\documentclass[acmsmall,screen]{acmart}

\acmPrice{}
\acmDOI{10.1145/3498691}
\acmYear{2022}
\copyrightyear{2022}
\acmSubmissionID{popl22main-p167-p}
\acmJournal{PACMPL}
\acmVolume{6}
\acmNumber{POPL}
\acmArticle{30}
\acmMonth{1}

\setcopyright{rightsretained}



\usepackage{booktabs}
\usepackage{caption}
\usepackage[defaultlines=3,all]{nowidow}
\usepackage{textgreek}

\usepackage[utf8]{inputenc} 
\usepackage[T1]{fontenc}    
\usepackage{amsfonts}       
\usepackage{nicefrac}       
\usepackage{microtype}      
\usepackage{amsthm}
\usepackage{xspace}
\usepackage[unicode]{hyperref}
\usepackage[capitalise]{cleveref}
\usepackage{enumitem}
\usepackage{multirow}

\usepackage{tikz}
\usetikzlibrary{positioning}
\usetikzlibrary{calc}
\usetikzlibrary{quantikz}
\usepackage{pgfplots}

\usepackage{listings}
\lstdefinelanguage{entanglement}
{morekeywords={
     fun,let,qubit,in,measure,bool,split,cast,entangle,type,qinit,'pure,'mixed,if,then,else
   },
 sensitive=true,
 morecomment=[n]{(*}{*)},
 morestring=[b]",
 escapechar=\%,
 columns=fullflexible,
 keepspaces=true,
 basicstyle=\ttfamily,
 mathescape=true,
}
\lstdefinestyle{tiny}{
    language=entanglement,
    basicstyle=\tiny\ttfamily,%
    keywordstyle=\color{blue!50!black}\ttfamily,%
    commentstyle=\color{gray}\ttfamily,
    alsoletter={'},
    columns=fullflexible,%
    keepspaces=true,%
    showstringspaces=false,
    numberstyle=\color{gray}\sffamily\tiny,
    numbersep=2mm,
    xleftmargin=2em,
}

\usepackage{mathpartir}
\usepackage{xargs}
\usepackage{pifont}
\usepackage{balance}
\usepackage{xcolor}

\usepackage{braket}
\usepackage{mathtools}
\usepackage{stmaryrd}
\usepackage{bbm}
\usepackage{cancel}

\newcommand{\LangName}{Twist}
\newcommand{\CoreLang}{\ensuremath{\mu Q}}
\newcommand{\CoreLangHeading}{\texorpdfstring{\ensuremath{\boldsymbol{\mu}}\textit{Q}}{𝜇Q}}

\title{\LangName{}: Sound Reasoning for Purity and Entanglement in Quantum Programs}

\author{Charles Yuan}
\affiliation{
  \institution{MIT CSAIL}
  \streetaddress{32 Vassar St}
  \city{Cambridge}
  \state{MA}
  \postcode{02139}
  \country{USA}
}
\email{chenhuiy@csail.mit.edu}

\author{Christopher McNally}
\affiliation{
  \institution{MIT RLE}
  \streetaddress{32 Vassar St}
  \city{Cambridge}
  \state{MA}
  \postcode{02139}
  \country{USA}
}
\email{mcnallyc@mit.edu}

\author{Michael Carbin}
\affiliation{
  \institution{MIT CSAIL}
  \streetaddress{32 Vassar St}
  \city{Cambridge}
  \state{MA}
  \postcode{02139}
  \country{USA}
}
\email{mcarbin@csail.mit.edu}

\DeclareMathOperator{\amp}{\&}

\DeclarePairedDelimiter\abs{\lvert}{\rvert}%

\newcommand{\qubit}{\ensuremath{\texttt{qubit}}}
\newcommand{\tBool}{\ensuremath{\texttt{bool}}}

\newcommand{\pure}{\ensuremath{\mathbf{P}}}
\newcommand{\mixed}{\ensuremath{\mathbf{M}}}
\newcommand{\purity}{\ensuremath{\pi}}

\newcommand{\qhastype}[3]{\ensuremath{\vdash_{#1} {#2} : {#3}}}
\newcommand{\hastype}[4]{\ensuremath{{#1} \vdash_{#2} {#3} : {#4}}}
\newcommand{\ctx}{\ensuremath{\Gamma}}
\newcommand{\rctx}{\ensuremath{\Delta}}
\newcommand{\qstate}{\ensuremath{\psi}}

\newcommand{\type}{\ensuremath{\tau}}
\newcommand{\qtype}{\ensuremath{\raisebox{-1.2pt}{\includegraphics[width=0.45em]{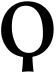}}}}

\makeatletter
\newcommand{\oset}[3][0ex]{%
  \mathrel{\mathop{#3}\limits^{
    \vbox to#1{\kern-2.75\ex@
    \hbox{$\scriptstyle#2$}\vss}}}}
\makeatother

\newcommand{\step}[5]{\ensuremath{{#1}; {#2}\ {\mapsto}_{#3}\ {#4}; {#5}}}
\newcommand{\stepo}[6]{\ensuremath{{#1}; {#2}\ {\oset{#3}{\longmapsto}}_{#4}\ {#5}; {#6}}}
\newcommand{\stepstar}[5]{\ensuremath{{#1}; {#2}\ {\mapsto}^*_{#3}\ {#4}; {#5}}}
\newcommand{\stepostar}[6]{\ensuremath{{#1}; {#2}\ {\oset{#3}{\longmapsto}}{}^*_{#4}\ {#5}; {#6}}}
\newcommand{\canstep}[2]{\ensuremath{{#1}; {#2}\ \mapsto \cdot}}
\newcommand{\val}[1]{\ensuremath{{#1}\ \textsf{val}}}
\newcommand{\stepv}[4]{\ensuremath{{#1}; {#2}\ {\mapsto}^*\ {#3}; {#4}}}

\newcommand{\iscompat}[2]{\ensuremath{{#1} \vDash {#2}}}

\newcommand{\Lam}[2]{\ensuremath{\lambda {#1}. {#2}}}
\newcommand{\App}[2]{\ensuremath{{#1}({#2})}}
\newcommand{\Let}[4]{\ensuremath{\mathtt{let}\ ({#1}, {#2}) = {#3}\ \mathtt{in}\ {#4}}}
\newcommand{\LetOne}[3]{\ensuremath{\mathtt{let}\ {#1} = {#2}\ \mathtt{in}\ {#3}}}
\newcommand{\Qinit}{\ensuremath{\mathtt{qinit}}}
\newcommand{\Uone}[1]{\ensuremath{U ({#1})}}
\newcommand{\Utwo}[1]{\ensuremath{U_2 ({#1})}}
\newcommand{\Measure}[1]{\ensuremath{\mathtt{measure}({#1})}}
\newcommand{\True}{\ensuremath{\mathsf{T}}}
\newcommand{\False}{\ensuremath{\mathsf{F}}}
\newcommand{\If}[3]{\ensuremath{\mathtt{if}\ {#1}\ \mathtt{then}\ {#2}\ \mathtt{else}\ {#3}}}

\newcommand{\Epair}[2]{\ensuremath{[{#1}, {#2}]}}
\newcommand{\Qref}[1]{\ensuremath{\mathtt{ref}[{#1}]}}
\newcommand{\Entangle}[3]{\ensuremath{\mathtt{entangle}_{#1}({#2}, {#3})}}
\newcommand{\EntangleOne}[2]{\ensuremath{\mathtt{entangle}_{#1}({#2})}}
\newcommand{\SplitA}[2]{\ensuremath{\mathtt{split}_{#1}({#2})}}
\newcommand{\Cast}[2]{\ensuremath{\mathtt{cast}_{#1}({#2})}}
\newcommand{\Qpurity}[2]{\ensuremath{{#1}^{#2}}}

\newcommand{\tPair}[2]{\ensuremath{{#1} \times {#2}}}
\newcommand{\tEpair}[2]{\ensuremath{{#1} \amp {#2}}}
\newcommand{\tPurity}[2]{\ensuremath{{#1}^{#2}}}

\newcommand{\hastypea}[3]{\ensuremath{{#1} \vdash_A {#2} : {#3}}}
\newcommand{\sempure}[2]{\ensuremath{{#2}; {#1}\ \mathsf{pure}}}
\newcommand{\sempuret}[3]{\ensuremath{{#2}; {#1}\ \mathsf{pure}_{#3}}}

\newcommand{\fail}[3]{\ensuremath{{#1}; {#2}\ {\not\mapsto}_{#3}}}

\newcommand{\SplitPoly}[2]{\ensuremath{\mathsf{Split}({#1}, {#2})}}
\newcommand{\CombinePoly}[2]{\ensuremath{\mathsf{Combine}({#1}, {#2})}}

\newcommand{\xmark}{\text{\ding{55}}}
\newcommand{\rulename}[1]{\textsc{#1}}

\newcommand{\pvar}[1]{\ensuremath{\mathtt{'}{#1}}}

\newcommand{\domain}[1]{\ensuremath{\mathrm{dom}\,{#1}}}
\newcommand{\trace}[1]{\ensuremath{\mathrm{tr}\,{#1}}}
\newcommand{\ptrace}[2]{\ensuremath{\mathrm{Tr}_{#1}({#2})}}

\newcommand{\SplitOp}{\ensuremath{\mathtt{split}}}
\newcommand{\CastOp}{\ensuremath{\mathtt{cast}}}
\newcommand{\SplitP}{\ensuremath{\SplitOp_\pure}}
\newcommand{\CastP}{\ensuremath{\CastOp_\pure}}
\newcommand{\SplitM}{\ensuremath{\SplitOp_\mixed}}
\newcommand{\CastM}{\ensuremath{\CastOp_\mixed}}

\newcommand{\Fun}[3]{\ensuremath{\mathtt{fun}\ {#1}\ {#2} = {#3}}}
\newcommand{\Main}[1]{\ensuremath{\mathtt{fun}\ \mathtt{main}\ () = {#1}}}
\newcommand{\ProgOk}[2]{\ensuremath{{#1} \vdash {#2}\ \mathsf{ok}}}
\newcommand{\denote}[1]{\ensuremath{\left\llbracket{#1}\right\rrbracket}}
\newcommand{\Denote}[3]{\ensuremath{\denote{#1}_{#2}\left({#3}\right)}}

\newcommand{\imeas}[4]{\ensuremath{{#1}; {#2} \Downarrow_{#3} {#4}}}

\newcommand{\vequiv}{\equiv}
\newcommand{\Refs}[1]{\ensuremath{\textsf{Refs}({#1})}}

\newcommand{\dyad}[1]{\ensuremath{\ket{#1}\hspace{-3.5pt}\bra{#1}}}
\newcommand{\vdenote}[1]{\ensuremath{\llparenthesis\,{#1}\,\rrparenthesis}}
\newcommand{\QrefM}[1]{\ensuremath{\mathtt{ref}^\star[{#1}]}}

\newcommand{\Comment}[1]{\ignorespaces}

\begin{document}

\begin{abstract}
\emph{Quantum programming languages} enable developers to implement algorithms for quantum computers that promise computational breakthroughs in classically intractable tasks.
Programming quantum computers requires awareness of \emph{entanglement}, the phenomenon in which measurement outcomes of qubits are correlated. Entanglement can determine the correctness of algorithms and suitability of programming patterns.

In this work, we formalize \emph{purity} as a central tool for automating reasoning about entanglement in quantum programs. A pure expression is one whose evaluation is unaffected by the measurement outcomes of qubits that it does not own, implying freedom from entanglement with any other expression in the computation.

We present \LangName{}, the first language that features a type system for sound reasoning about purity. The type system enables the developer to identify pure expressions using type annotations. \LangName{} also features purity assertion operators that state the absence of entanglement in the output of quantum gates. To soundly check these assertions, \LangName{} uses a combination of static analysis and runtime verification.

We evaluate \LangName{}'s type system and analyses on a benchmark suite of quantum programs in simulation, demonstrating that \LangName{} can express quantum algorithms, catch programming errors in them, and support programs that several languages disallow, while incurring runtime verification overhead of less than 3.5\%.
\end{abstract}

\maketitle

\section{Introduction}

\emph{Quantum programming languages}~\citep{qsharp, quipper, qwire, reqwire, liquid, ying-foundations, qml, qpl, quantum-lambda-calc, silq, clairambault20, staton17} allow programmers to utilize the computational primitives enabled by the quantum computers of today and tomorrow.
Algorithms for quantum computers offer computational breakthroughs in integer factorization~\citep{shor1997}, search~\citep{grover}, cryptographic and communication protocols~\citep{bennett, teleportation}, computational physics and chemistry~\citep{simulation-physics,simulation-chem}, and machine learning~\citep{biamonte2017}.

Quantum computation relies on the manipulation of \emph{quantum states} consisting of \emph{qubits}, the quantum analogs of classical data and bits.
A quantum state exists in a \emph{superposition}, a weighted sum over classical states. \emph{Measurement} causes a superposition to assume a classical state, with probability derived mathematically from the weight ascribed to that state in the sum. In the standard QRAM~\citep{knill} model, computations execute on a classical computer with access to a quantum device that supports initializing and operating on quantum states.

\subsection{Entanglement}

\emph{Entanglement}, the phenomenon of correlation between qubits,\footnotemark{} is critical to the quantum computational advantage~\citep{jozsa2003}.
Given a pair of entangled qubits, measuring one forces the other to assume a state consistent with the measurement.
Thus, the measurement outcome of one qubit causes operations on the other to potentially yield different behavior.%
\footnotetext{More precisely, the measurement outcomes of entangled quantum states have statistical correlations that cannot be explained by physical theories with local realism~\cite{nielsen_chuang_2010}.}

In many instances, reasoning about entanglement is either necessary or beneficial.
During debugging, determining whether two qubits are entangled at a particular point is a sanity check that an algorithm has been implemented correctly~\citep{huang2019}.
Another application is in the handling of temporary qubits used by several algorithms\footnote{Such as oracle functions in Grover's algorithm~\citep{grover} or modular multiplication in Shor's algorithm~\citep{shor1997}.} and programming patterns\footnote{Such as conditionally executing statements depending on some predicate over qubits~\citep{silq}.} that must be measured and deallocated so that physical qubits may be reused. If these temporaries remain entangled with the primary results of the computation, their measurement will cause the results to be incorrect~\citep{silq,reqwire,quipper}.

A further application is mitigation of information-leakage attacks on Shor's algorithm~\citep{azuma2017} and quantum bit-commitment schemes~\citep{lo1998} in which attackers introduce surreptitious entanglement that is rendered impossible if the sensitive state is verified to be not entangled.
Yet another application is to compiler analyses that leverage descriptions of which qubits are entangled to optimize usage of resources such as qubits in programs~\citep{haner2020}.

Entanglement is thus a key to reason about the correctness of an algorithm, verify the suitability of a programming pattern, and empower compiler analyses.
Though languages~\citep{silq,reqwire,amy} have recently been developed to reason about quantum phenomena such as reversibility of computation, prior work has yet to facilitate sound reasoning about entanglement in quantum programs.

\subsection{Purity}

We introduce the concept of the \emph{purity} of an expression, which enables reasoning about entanglement in quantum programs.
Operationally, quantum expressions contain references to qubits, analogous to pointers to classical memory, and evaluate by executing operations known as \emph{gates} on them. We say that an expression \emph{owns} a qubit when the final value to which it evaluates, such as a tuple of qubits, contains a reference to that qubit.

An expression is \emph{pure} if its evaluation is unaffected by the measurement outcome of any qubit it does not own, and \emph{mixed} otherwise.
Specifically, the qubits that a pure expression owns are only potentially entangled with each other and are \emph{separable}, or free of entanglement, from those in the remainder of the program. Pure and mixed expressions coincide with the established quantum mechanical definitions of pure and mixed states~\citep{nielsen_chuang_2010} in that evaluating a pure expression results in its owned qubits constituting a pure sub-state of the program state.

\subsection{Specification and Verification of Purity}

We present \LangName{}, the first language that enables programmers to specify that an expression is pure and soundly verify that the specification holds.
\LangName{} provides a type system to specify that an expression is of pure type,
purity assertions to declare the absence of entanglement, and a combination of static analysis and runtime verification to ensure that purity specifications hold.

\paragraph{Purity Types}
The type system enables a programmer to specify the type of a qubit or tuple of qubits as either pure -- unaffected by measurements of other qubits in the computation -- or mixed -- potentially affected by other qubits.
The type system reasons about purity by tracking potential entanglement between qubits, conservatively identifying that any quantum gate operating on two qubits may entangle them. We introduce a type of \emph{entangled pairs} of quantum data that are potentially entangled with each other, so that the outputs of multi-qubit quantum gates are of entangled tuple type. Elements projected from an entangled pair have mixed type to reflect their potential entanglement with each other.

\paragraph{Purity Assertions}
Because the application of a quantum gate may in fact remove the entanglement of a qubit with others in the computation, the language provides constructs to assert the absence of entanglement relationships in the program.
\LangName{} provides two purity assertions: \texttt{cast}, which asserts that an expression is pure, and \texttt{split}, which asserts that the components of a pure entangled pair are separable from each other and hence pure.
These two assertions work in concert to identify tuples of qubits that have no external entanglement, and then further assert the absence of entanglement between individual qubits within a tuple.

\paragraph{Purity Assertion Verification}
Despite progress in statically characterizing the effect of gates on entanglement~\citep{rand2021static}, in a general quantum program, verifying that a qubit is pure is at least as hard as simulating the program~\citep{gurvits,hayden}.
\LangName{} therefore relies on a combination of both static analysis and runtime verification to check purity assertions.

To verify the \texttt{cast} assertion, \LangName{} first executes a conservative static analysis to determine whether the set of qubits that an expression owns is separable from all others in the program.
To do so, the analysis identifies the qubits that ever share an entangled pair with any qubit in this set.
If all such qubits are also in this set, then the analysis concludes the expression is pure.

To verify \texttt{split}, \LangName{} tests at runtime whether two precise sets of qubits are separable from each other. These checks rely on a primitive that determines whether the runtime quantum state is separable or entangled. Similarly to the quantum assertions of \citet{huang2019,assertions,approximate-assertions}, we abort execution if the condition fails.

\paragraph{Separability Testing}
\LangName{}'s purity assertions rely on a primitive that tests whether the runtime quantum state is separable.
\citet{harrow-separability,walborn2006} have proposed implementations of such a test in hardware, which remain a topic of active research.
In this work, we implement \LangName{} on a state vector-based quantum simulator, examples of which include~\citet{qiskit,qpp}. We present a concrete implementation of separability testing in simulation based on the Schmidt decomposition~\citep{SchmidtZurTD} of quantum states.

\paragraph{Summary.} Together, \LangName{}'s purity types, purity assertions, and analysis techniques work to ensure the sound verification of purity specifications in quantum programs. The developer may use purity types to require pure expressions in computations, use purity assertion constructs to state conditions to be verified, and execute the static analysis and runtime verification on the assertions. The resulting guarantees of purity enable the programmer to debug algorithms, leverage idioms, and enjoy correctness guarantees in their programs.

\subsection{Contributions}

In this paper we present the following contributions:

\begin{itemize}[leftmargin=5.5mm]
    \item \emph{Purity}. We present the novel definition of \emph{pure} expressions in a quantum program, those that are unaffected by measurement outcomes of the remainder of the program. We formulate purity within the operational and denotational semantics of a functional quantum language.

	\item \emph{Purity Types}. We present a type system that identifies pure expressions, and prove that in it, expressions of pure type are in fact pure.

	\item \emph{Purity Assertions}. We present two types of purity assertions that state the absence of entanglement in the output of quantum gates: one stating that an expression is pure, and one stating that the two components of a pure entangled pair are individually pure.

	\item \emph{Purity Assertion Verification}. We present a static analysis and runtime verifications for the purity assertions, such that programs passing verification satisfy their purity specifications.

	\item \emph{Evaluation}. We implement \LangName{}, a language featuring purity types and assertions, in quantum simulation. We show that \LangName{} can express quantum algorithms and reject programming errors in them, that its runtime verification executes with overhead less than $3.5\%$, and that it can express semantically valid programs that existing languages disallow.
\end{itemize}

\noindent Our work introduces the powerful notion of purity to quantum programming, enabling sound reasoning for entanglement. Using \LangName{}, developers can recognize quantum entanglement not as a cognitive burden but rather as a clarifying tool to understanding the correctness of their programs.

\section{Background on Quantum Computation} \label{sec:background}

The following is an overview of key concepts in quantum computation relevant to this work and our notational choices. \citet{nielsen_chuang_2010} provide a comprehensive reference.

\subsection{Pure State Formalism}

We first define a \emph{pure quantum state} and the main formalism of quantum mechanics in this work.

\paragraph{Qubit}
The basic unit of quantum information is the \emph{qubit}, a linear combination $\gamma_0 \ket{0} + \gamma_1 \ket{1}$ known as a \emph{superposition}, where $\ket{0}$ and $\ket{1}$\footnote{The Dirac ket notation $\ket{\cdot}$ is customary in quantum mechanics. In our work, it simply denotes a quantum state.} are \emph{basis states} and $\gamma_0, \gamma_1 \in \mathbb{C}$ are \emph{amplitudes} satisfying $\abs{\gamma_0}^2 + \abs{\gamma_1}^2 = 1$ describing relative weights of basis states. Examples of qubits include the classical zero bit $\ket{0}$, classical one bit $\ket{1}$, and the superposition states ${(\ket{0} + \ket{1})}/{\sqrt{2}}$ and ${(i\ket{0} - \ket{1})}/{\sqrt{2}}$.

\paragraph{Quantum State}
A $2^n$-dimensional \emph{pure quantum state} $\ket{\qstate}$ is a superposition over $n$-bit strings. For example, ${(\ket{00}+\ket{11})}/{\sqrt{2}}$ is a quantum state over 2 qubits. Equivalently, we may represent any pure state as a \emph{state vector}, a length-$2^n$ vector of normalized complex amplitudes.\footnote{For example, the state vector corresponding to ${(\ket{00}+\ket{11})}/{\sqrt{2}}$ is $[\nicefrac{1}{\sqrt{2}}, 0, 0, \nicefrac{1}{\sqrt{2}}]^\top$, with the elements corresponding to amplitudes of $\ket{00}$, $\ket{01}$, $\ket{10}$, and $\ket{11}$ respectively.}

Multiple qubits form a quantum state system by means of the tensor product $\otimes$. Thus, the state $\ket{01}$ is equal to the product $\ket{0} \otimes \ket{1}$.
We use subscripts to denote the names of individual qubits or sets of qubits. For example, to denote a two-qubit system in which a qubit named $\alpha$ has state $\ket{0}$ and qubit $\beta$ has state $\ket{1}$, we write $\ket{0}_\alpha \otimes \ket{1}_\beta$.\footnote{In our notation, we treat the tensor product as commutative, so that $\ket{1}_\beta \otimes \ket{0}_\alpha$ refers to the same state as $\ket{0}_\alpha \otimes \ket{1}_\beta$, with the understanding that qubits are actually stored in some canonical order in the state.}

We define the \emph{empty state} $\ket{\cdot}$ to be the length-$2^0$ vector containing amplitude 1, which effectively describes a system of zero qubits. The tensor product of any $\ket{\qstate}$ and $\ket{\cdot}$ is accordingly $\ket{\qstate}$. The \emph{domain} of a state, $\domain{\ket{\qstate}}$, is the set of qubit names it contains.

\paragraph{Unitary Operators}
A $2^n$-dimensional \emph{unitary operator} is a linear operator on state vectors represented by an $2^n \times 2^n$ matrix $U$ that preserves inner products and whose inverse is its Hermitian adjoint. We denote the state produced by a unitary operator acting on qubit $\alpha$ in state $\ket{\qstate}$ by $U_\alpha\ket{\qstate}$.\footnote{In this work, we will work with one- and two-qubit unitary operators on named qubits with the understanding that they will be padded to all $n$ qubits in the system by tensor product with the identity matrix.}
In this work, we use single-qubit quantum gates such as:
\begin{itemize}
    \item X -- bit-flip (NOT) gate, which maps $\ket{0}$ to $\ket{1}$ and vice versa;
    \item Z -- phase-flip gate, which leaves $\ket{0}$ unchanged and maps $\ket{1}$ to $-\ket{1}$;
    \item H -- Hadamard gate, which maps $\ket{0}$ to ${(\ket{0} + \ket{1})}/{\sqrt{2}}$ and $\ket{1}$ to ${(\ket{0} - \ket{1})}/{\sqrt{2}}$.
\end{itemize}
We also use two-qubit gates, such as controlled-NOT (CNOT), controlled-Z (CZ), and SWAP. The controlled gates perform NOT or Z on their target qubit if their control qubit is in state $\ket{1}$.

The SWAP gate swaps two qubits in a quantum state, and inserting SWAP gates enables us to rename qubits at will. We use $\ket{\qstate}[\gamma / \alpha]$ to denote renaming qubit $\alpha$ to a new name $\gamma$ in $\ket{\qstate}$ by implicitly inserting SWAP gates. For example, $\ket{\qstate} = \ket{0}_\alpha \otimes \ket{1}_\beta$ becomes $\ket{\qstate}[\gamma / \alpha] = \ket{1}_\beta \otimes \ket{0}_\gamma$.

\paragraph{Measurement}
A quantum \emph{measurement} is a probabilistic operation over quantum states. When a qubit $\gamma_0 \ket{0} + \gamma_1 \ket{1}$ is measured,\footnote{We concern ourselves primarily with projective computational (i.e. 0/1) basis measurements, though our results can be generalized to other measurement forms.} the outcome is $\ket{0}$ with probability $\abs{\gamma_0}^2$ and $\ket{1}$ with probability $\abs{\gamma_1}^2$. Measuring a qubit within a larger quantum state will cause the entire state to probabilistically assume one of two outcomes. The outcome state after measurement is equal to the tensor product of the just-measured qubit in a basis state and the new state of the remainder of the system.

We denote the state produced by measurement of qubit $\alpha$ in state $\ket{\qstate}$ as $M_\alpha\ket{\qstate}$. To define measurement, we first rewrite the state into the form $\ket{\qstate} = \gamma_0 \ket{0}_\alpha \otimes \ket{\qstate_0}+ \gamma_1 \ket{1}_\alpha \otimes \ket{\qstate_1}$ where $\ket{\qstate_0}$ and $\ket{\qstate_1}$ are unique quantum states.\footnotemark{} Then, with probability $\abs{\gamma_0}^2$ we obtain $M_\alpha\ket{\qstate} = \ket{0}_\alpha \otimes \ket{\qstate_0}$, and with probability $\abs{\gamma_1}^2$ we obtain $M_\alpha\ket{\qstate} = \ket{1}_\alpha \otimes \ket{\qstate_1}$.
\footnotetext{Scaling a quantum state by a coefficient $e^{i\theta}$ known as a global phase factor produces another state indistinguishable from the original by any measurement. Thus, we define equality and uniqueness of states to be up to global phase.}

\paragraph{Pure and Mixed States}
Unlike the pure states defined so far,
the result of a measurement is a classical probability distribution over pure states, known as a \emph{mixed} state.

For example, ${(\ket{0} + \ket{1})}/{\sqrt{2}}$ is a pure state, whereas the distribution of $\ket{0}$ with probability $\nicefrac{1}{2}$ and $\ket{1}$ w.p. $\nicefrac{1}{2}$ is a mixed state.
Though measuring these two states immediately yields the same outcome distribution, they behave differently under unitary operators. For example, applying a Hadamard gate to ${(\ket{0} + \ket{1})}/{\sqrt{2}}$ always produces $\ket{0}$. Applying Hadamard to the mixed state instead produces another mixed state that when measured yields either 0 or 1 with equal probability.\footnote{In particular, this mixed state is ${(\ket{0} + \ket{1})}/{\sqrt{2}}$ with probability $\nicefrac{1}{2}$ and ${(\ket{0} - \ket{1})}/{\sqrt{2}}$ with probability $\nicefrac{1}{2}$.}

\paragraph{Entanglement and Separability}
A \emph{bipartite} quantum state $\ket{\qstate}$ is a state over the disjoint union of two qubit sets $A \sqcup B$. A bipartite state is \emph{separable} if it can be written as a tensor product of two states over each set, $\ket{\qstate_A} \otimes \ket{\qstate_B}$, or \emph{entangled} otherwise.

For example, the two-qubit state ${(\ket{00} + \ket{01} + \ket{10} + \ket{11})}/{2}$ is separable because it is the product of two copies of ${(\ket{0} + \ket{1})}/{\sqrt{2}}$. By contrast, the \emph{Bell state}~\citep{bell} ${(\ket{00} + \ket{11})}/{\sqrt{2}}$ is entangled because it cannot be written as the product of two single-qubit states.

Given the bipartite state $\ket{\qstate}$, measuring the qubits of $B$ will have different consequences for the remaining state of $A$, depending on whether $\ket{\qstate}$ is separable. If $\ket{\qstate}$ is separable, the measurement will leave $A$ in a pure state, and if it is entangled, measurement will leave $A$ in a mixed state. For example, measuring one of the qubits in the entangled Bell state ${(\ket{00} + \ket{11})}/{\sqrt{2}}$ results in the remaining qubit taking on a mixed state, $\ket{0}$ with probability $\nicefrac{1}{2}$ and $\ket{1}$ with probability $\nicefrac{1}{2}$.

\paragraph{Schmidt Decomposition}
Given a bipartite pure state $\ket{\psi}$ over qubit sets $A \sqcup B$, we may compute its unique \emph{Schmidt decomposition}~\citep{SchmidtZurTD}, $\ket{\psi} = \sum_j \lambda_j \ket{{\psi_A}_j} \otimes \ket{{\psi_B}_j}$ where $\ket{{\psi_A}_j}$ and $\ket{{\psi_B}_j}$ are states of $A$ and $B$ and $\lambda_j$ are positive real \emph{Schmidt coefficients} satisfying $\sum_j \lambda_j^2 = 1$. The Schmidt decomposition provides a criterion for separability -- $\ket{\psi}$ is separable if and only if it has only one nonzero Schmidt coefficient.

\subsection{Mixed State Formalism}

We next describe mixed states, a more expressive alternative formalism for quantum computation. Mixed states model statistical ensembles of states arising over multiple program executions.

\paragraph{Density Matrix}
Given a $2^n$-dimensional state vector $\ket{\qstate}$, we use $\dyad{\qstate}$ to denote its outer product with itself, which is a $2^n \times 2^n$ matrix.
A mixed state is mathematically represented as a \emph{density matrix} $\rho$, a linear combination $\rho = \sum_j p_j \dyad{\psi_j}$ where each $p_j > 0$ and $\rho$ is positive semidefinite.
The \emph{domain} of a density matrix, $\domain{\rho}$, is the set of qubits contained in each $\ket{\psi_j}$.

A density matrix is \emph{normalized} if $\sum_j p_j = 1$ and $\trace{\rho} = 1$.
\emph{Partial density matrices}~\citep{ying-foundations} relax the conditions to $\sum_j p_j \le 1$ and $\trace{\rho} \le 1$ and can be added to form normalized density matrices.

\paragraph{Unitary Operators}
A unitary quantum operator $U$ applies to a density matrix $\rho$ by matrix conjugation, so that the resulting matrix is $U\rho U^\dagger$ where $U^\dagger$ is the Hermitian adjoint of $U$.

\paragraph{Measurement}
A quantum measurement is represented by a set of projections $P_j$ corresponding to possible outcomes, where $\sum_j P_j = I$. Outcome $P_j$ occurs with probability $\trace{(\rho P_j)}$, and results in normalized density matrix ${P_j\rho P_j}/{\trace{(\rho P_j)}}$. $P_j \rho P_j$ is its probability-weighted partial density matrix.

The principle of \emph{deferred measurement} states that any computation that conditionally executes gates based on the measurement of a qubit produces the same mixed state as one that uses quantum conditioned gates and defers measurement until the end of the computation.\footnote{e.g. uses a CNOT gate rather than choosing whether to execute a NOT gate based on a classical measurement outcome.}

\paragraph{Product States and Separability}
We can construct a composite of mixed states using the tensor product $\rho_1 \otimes \rho_2$. We use the notation $\rho[\gamma / \alpha]$ to rename qubit $\alpha$ to $\gamma$ in $\rho$ by implicitly inserting SWAP gates.
A mixed state $\rho$ is \emph{simply separable} if there exist $\rho_1$ and $\rho_2$ where $\rho = \rho_1 \otimes \rho_2$.

\paragraph{Partial Trace}
The \emph{partial trace} of $\rho$ over $A$, $\ptrace{A}{\rho}$, is the unique linear operator satisfying:
\[ \ptrace{A}{\dyad{\psi_A} \otimes \dyad{\psi_B}} = \trace{(\dyad{\psi_A})}\dyad{\psi_B} \]
where $A = \domain{\ket{\psi_A}}$. In the special case where $\rho$ is simply separable as $\rho_A \otimes \rho_B$ where $\domain{\rho_A} = A$ and $\domain{\rho_B} = B$, we have $\ptrace{B}{\rho} = \rho_A$ and $\ptrace{A}{\rho} = \rho_B$.

\paragraph{Purity}
A mixed state $\rho$ is \emph{pure} if it is not an ensemble of more than one state: $\rho = \dyad{\qstate}$ for some pure state $\ket{\psi}$. The \emph{rank test} states that $\rho$ is pure if and only if $\trace{(\rho^2)} = (\trace{\rho})^2 = 1$. We similarly have that a partial density matrix $\rho = p \dyad{\qstate}$ if and only if $\trace{(\rho^2)} = (\trace{\rho})^2$ by linearity of trace.

\section{Example} \label{sec:examples}

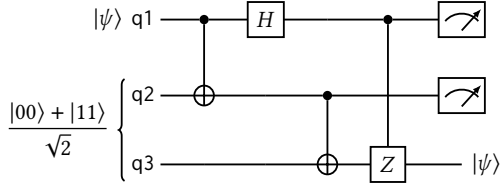
\begin{figure}
\resizebox{0.5\textwidth}{!}{
\begin{quantikz}
\lstick{$\ket{\psi}$} \mathtt{q1}\ & \ctrl{1} & \gate{H} & \qw & \ctrl{2} & \meter{} \\
\lstick[wires=2]{$\dfrac{\ket{00} + \ket{11}}{\sqrt{2}}$} \mathtt{q2}\ & \targ{} & \qw & \ctrl{1} & \qw & \meter{} \\
\mathtt{q3}\ & \qw & \qw & \targ{} & \gate{Z} & \qw \rstick{$\ket{\psi}$}
\end{quantikz}
}
\caption{Teleportation protocol with deferred measurement expressed as a quantum circuit. The inputs are \texttt{q1}, the qubit to be teleported, and a Bell pair of two qubits. The circuit applies conditional-NOT (CNOT) and conditional-Z (CZ) gates to \texttt{q3}. It measures two qubits and outputs one with the teleported state of \texttt{q1}.} \label{fig:circuit}
\end{figure}

We illustrate the value of purity and \LangName{} with the protocol of \emph{quantum teleportation}~\citep{teleportation}, a demonstration of the power of entanglement and a building block for techniques such as gate teleportation~\citep{Gottesman_1999}.
The protocol transmits the information stored in one qubit to a receiver an arbitrary distance away by transferring only two classical bits of information. Though the original protocol uses measurement-conditioned gates to emphasize classical information exchange, we examine a variant~\citep{miller,kumar} that instead uses quantum-conditioned gates and defers all measurement to the end of the program.

\Cref{fig:circuit} presents a quantum circuit for the deferred-measurement variant of teleportation that accepts as input a qubit \texttt{q1} to be teleported.\footnotemark{} We assume a pair of qubits $(\mathtt{q2}, \mathtt{q3})$ exist in a Bell state (\Cref{sec:background}).
The circuit entangles \texttt{q1} and \texttt{q2} by a CNOT gate, applies a Hadamard (H) gate to \texttt{q1}, and applies CNOT and CZ gates in succession on \texttt{q3}.
Finally, the circuit measures \texttt{q1} and \texttt{q2} and outputs \texttt{q3}, which now contains the original state of \texttt{q1}.
\footnotetext{For the purpose for demonstrating our results, understanding the rationale behind this circuit is not necessary; we refer the reader to~\citet{nielsen_chuang_2010} for a detailed explanation.}

\Cref{fig:teleport-program} presents a program implementing this circuit, as a function accepting a qubit \texttt{q1} and returning the teleported output, where the helper \texttt{bell\_pair} allocates a Bell pair.

\begin{figure}
\begin{lstlisting}[numbers=left]
fun teleport (q1 : qubit) : qubit =
  let (q2 : qubit, q3 : qubit) = bell_pair () in
  let (q1 : qubit, q2 : qubit) = CNOT (q1, q2) in
  let q1 : qubit = H (q1) in
  let (q2 : qubit, q3 : qubit) = CNOT (q2, q3) in
  let (q1 : qubit, q3 : qubit) = CZ (q1, q3) in
  let _ : bool * bool = measure (q1, q2) in q3
\end{lstlisting}
\caption{Implementation of the teleportation circuit of~\Cref{fig:circuit} as a quantum program.} \label{fig:teleport-program}
\end{figure}

\subsection{Entanglement and Purity}

The specification of the teleportation protocol is that the final state of \texttt{q3} is the same as the initial state of \texttt{q1}.
Thus, a central property that ensures the correctness of this program is that \texttt{q3} is \emph{pure} -- free of entanglement with \texttt{q1} and \texttt{q2} and hence unaffected by the measurements of these qubits.

Consider instead if one were to replace the CZ gate on line 6 of~\Cref{fig:teleport-program} with a different gate that causes \texttt{q3} to remain entangled with \texttt{q1} and \texttt{q2}, for example a CNOT gate. After replacing CZ with CNOT, \texttt{q3} would be affected by the measurements on line 7.

Because \texttt{q3} is now entangled with \texttt{q1} at the time of measurement, the program no longer satisfies the specification that \texttt{q3} assumes the original state of \texttt{q1}. Instead, if \texttt{q1} initially has the state
${(\ket{0} + \ket{1})}/{\sqrt{2}}$, then \texttt{q3} will assume a different state:
\[
  \begin{cases}
    \frac{\ket{0} + \ket{1}}{\sqrt{2}} & \text{with probability } \frac{1}{2} \\
    \frac{-\ket{0} + \ket{1}}{\sqrt{2}} & \text{with probability } \frac{1}{2}
  \end{cases}
\]
This probability distribution over pure states is a mixed state (\Cref{sec:background}) that stems from different measurement outcomes of \texttt{q1} and \texttt{q2}. By contrast, pure expressions must always evaluate to unique final states not dependent on the measurements of other qubits.

\subsection{Purity Types}

\begin{figure}
\vspace{-1em}
\begin{lstlisting}[numbers=left]
fun teleport (q1 : qubit<P>) : qubit<M> = (* mixed type *)
  let q23 : (qubit & qubit)<P> = bell_pair () in
  let (q2 : qubit<M>, q3 : qubit<M>) = q23 in
  let q12 : (qubit & qubit)<M> = CNOT (q1, q2) in
  let (q1 : qubit<M>, q2 : qubit<M>) = q12 in
  let q1 : qubit<M> = H (q1) in
  let (q2 : qubit<M>, q3 : qubit<M>) = CNOT (q2, q3) in
  let (q1 : qubit<M>, q3 : qubit<M>) = CZ (q1, q3) in
  let _ : bool * bool = measure (q1, q2) in q3
\end{lstlisting}
\caption{Teleportation program in \LangName{}, using purity types. The return type of the function is currently mixed.}
\label{fig:teleport-reject}
\end{figure}

\Cref{fig:teleport-reject} presents the teleportation program from \Cref{fig:teleport-program}, but written in \LangName{} with type annotations for purity. In a \LangName{} program, every quantum expression is of a type that is either pure or mixed.
Pure expressions are those that are unaffected by the measurement of other qubits whereas mixed expressions are those that may be affected.

We say that an expression \emph{owns} a qubit if it evaluates to a value that contains a reference to the qubit.
If a qubit owned by an expression is entangled with another unowned qubit, then the measurement outcome of the unowned qubit inevitably affects the state of the owned qubit. We describe entanglement in the type system using a type of \emph{entangled pairs}, denoted by the type constructor \texttt{\&}, such as on line 2 of~\Cref{fig:teleport-reject}.
An entangled pair stipulates that its two components potentially share entanglement.

Operations such as \texttt{bell\_pair} (line 2) and \texttt{CNOT} (line 4) return entangled pairs, indicating that two qubits are potentially entangled.
Projection from an entangled pair results in an expression of mixed type.
For example, when the program projects \texttt{q2} and \texttt{q3} from the entangled pair \texttt{q23} on line 3, each qubit has mixed type because they may be entangled with each other.
Finally, \texttt{q3} has mixed type according to the rules of the type system, denoting that this implementation may not satisfy its specification of returning a pure qubit.

\subsection{Purity Assertions}

\begin{figure}
\vspace{-0.5em}
\begin{lstlisting}[numbers=left]
fun teleport (q1 : qubit<P>) : qubit<P> = (* pure type *)
  let q23 : (qubit & qubit)<P> = bell_pair () in
  let (q2 : qubit<M>, q3 : qubit<M>) = q23 in
  let (q1 : qubit<M>, q2 : qubit<M>) = CNOT (q1, q2) in
  let q1 : qubit<M> = H (q1) in
  let (q2 : qubit<M>, q3 : qubit<M>) = CNOT (q2, q3) in
  let (q1 : qubit<M>, q3 : qubit<M>) = CZ (q1, q3) in
  let q123 : ((qubit & qubit) & qubit)<M> = ((q1, q2), q3) in
  (* assert ((q1, q2), q3) is pure; check statically *)
  let q123 : ((qubit & qubit) & qubit)<P> = cast<P>(q123) in
  (* assert q3 is separable from (q1, q2); check dynamically *)
  let (q12 : (qubit & qubit)<P>, q3 : qubit<P>) = split<P>(q123) in
  let _ : bool * bool = measure (q12) in q3
\end{lstlisting}
\caption{Teleportation example with purity types and assertions so that its return type is refined to be pure.}
\label{fig:teleport-fixed}
\end{figure}

\Cref{fig:teleport-fixed} presents an implementation of the \texttt{teleport} function that instead soundly returns a pure output by leveraging the purity assertions of Twist.
The program performs two steps to verify that \texttt{q3} is not entangled with any other qubits in the program.

In the first step, the program forms an entangled pair containing all three qubits on line 8 and then on line 10 uses \texttt{cast<P>}, a \emph{purifying-cast assertion} that states that this triple is pure and has no entanglement with the rest of the program.

In the second step, on line 12, the program uses \texttt{split<P>}, a \emph{purifying-split assertion} that states that the two components of the triple, one containing $\mathtt{q1}$ and $\mathtt{q2}$, and the other \texttt{q3}, are not entangled with each other, confirming that they are both individually pure.

\subsection{Purity Assertion Verification}

\LangName{} uses a static analysis to verify \texttt{cast}, whereas it verifies \texttt{split} at runtime.

\paragraph{Verifying \texttt{cast}}
The static analysis determines that an expression is pure if it contains either zero or both components of every entangled pair created by the execution of the program.
In~\Cref{fig:teleport-fixed}, \LangName{} determines that \texttt{q123} is not entangled with the rest of the program, because it contains both components of the pure entangled pair \texttt{q23} created by \texttt{bell\_pair} on line 2, along with the pure function argument \texttt{q1}, all of which never interact with any other qubits.

\paragraph{Verifying \texttt{split}}
In contrast to using static analysis as with \texttt{cast}, \LangName{} verifies \texttt{split} at runtime, because statically checking separability necessitates direct reasoning about the quantum semantics of gates, which to the best of our knowledge is at least as hard as simulation except for circuits constructed from restricted classes of gates~\citep{hayden}.

Specifically, \LangName{} determines whether the two components of a pure entangled pair are entangled with each other via a runtime separability test.
Our implementation of this test obtains the simulated state vector of the program at line 12, and determines whether it is separable along the components of the pair by means of the Schmidt decomposition (\Cref{sec:background}).
Because \texttt{q3} is in fact pure in this program, the assertion always succeeds.

\begin{figure}
\begin{lstlisting}[numbers=left]
fun teleport (q1 : qubit<P>) : qubit<P> =
  let q23 : (qubit & qubit)<P> = bell_pair () in
  let (q2 : qubit<M>, q3 : qubit<M>) = q23 in
  let (q1 : qubit<M>, q2 : qubit<M>) = CNOT (q1, q2) in
  let q1 : qubit<M> = H (q1) in
  let (q2 : qubit<M>, q3 : qubit<M>) = CNOT (q2, q3) in
  let (q1 : qubit<M>, q3 : qubit<M>) = CZ (q1, q3) in
  let _ : bool * bool = measure (q1, q2) in
  cast<P>(q3) (* direct assertion of q3 as pure *)
\end{lstlisting}
\caption{Teleportation program that asserts the purity of \texttt{q3} with a single purity assertion on line 9. Checking the assertion that \texttt{q3} is pure after the measurements of \texttt{q1} and \texttt{q2} requires reasoning on mixed states, which is inefficient in practice. Consequently, our static analysis rejects this program.}
\label{fig:teleport-single}
\end{figure}

\paragraph{Efficiency}
We have designed \texttt{split} and \texttt{cast} to be used in concert to make their verification more efficient. Consider instead the program in~\Cref{fig:teleport-single} that gives \texttt{q3} pure type, but does not use \texttt{split} and only uses \texttt{cast} on line 9.
While this program type checks and has a valid semantic meaning, the static analysis rejects this program by virtue of the fact that the operand of the \texttt{cast} (\texttt{q3}) does not contain all potentially entangled qubits (\texttt{q1} and \texttt{q2}) as was the case in \Cref{fig:teleport-fixed}.

The reason we have designed the static analysis to reject this program is because the measurements of \texttt{q1} and \texttt{q2} (line 8) imply that \texttt{q3} could be in a mixed state, if \texttt{q3} were in fact entangled with either \texttt{q1} or \texttt{q2}.
It is not possible to dynamically verify purity for a qubit within a mixed state by separability tests on the runtime quantum state alone. To perform this check, \LangName{} would instead have to simulate the full density matrix of the program. However, density matrix simulation is inefficient compared to state vector simulation. In a quantum simulator, density matrices have quadratic space overhead over state vectors, and standard simulators such as~\citet{qiskit} support only half as many qubits in mixed-state simulation.%
\footnote{In quantum mechanics, mixed states have more general definitions of entanglement and separability than pure states. We show in~\Cref{sec:app-entanglement} that the special case of \emph{simple separability} is appropriate for reasoning about purity in \LangName{}. This is fortunate, as testing for more general separability is an open problem for non-trivially-small cases~\citep{horodecki1996,harrow-separability}, and was proved by~\citet{gurvits} to be \textbf{NP}-hard in the dimension of the density matrix.}\nowidow

By only performing runtime verification of \texttt{split}, whose argument is of pure type, \LangName{} may exploit the available relative efficiency of testing separability of a runtime pure quantum state.
Put together, \LangName{} balances both static analysis to identify pure types where statically practical -- discharging the \texttt{cast} operator in \Cref{fig:teleport-fixed} -- as well as runtime verification to dynamically enforce purity where it too can be done practically -- the \texttt{split} operator on a pure quantum state.

\subsection{Discarding Pure Values}

\LangName{}'s purity types enable us to write the teleportation program even more concisely than~\Cref{fig:teleport-fixed}. Line 13 of~\Cref{fig:teleport-fixed} explicitly measures \texttt{q12}, a value the program no longer needs. Instead, in~\Cref{fig:teleport-discard} we remove this measurement operation and simply discard \texttt{q12}.

The reason that discarding unused qubits is a potential point of contention is the deferred measurement principle (\Cref{sec:background}). The principle states that discarding a qubit or allowing it to leave scope has the same effect as measuring it at its last point of use, meaning that discarding a value is always akin to measuring it. Measurement may in general affect the states of values elsewhere in the program, which would be an unintuitive consequence if it occurred when we implicitly discarded a variable.

However, the measurement outcomes of a pure expression cannot affect the states of qubits it does not own, and thus the outcome of the measurement of \texttt{q12} cannot have any impact on the remaining computation. In general, \LangName{} supports implicitly measuring a pure expression when it is discarded and goes out of scope.

Existing languages such as~\citet{qwire, qpl} forbid programs from implicitly discarding quantum data.
The language Silq~\citep{silq} can automatically \emph{uncompute} certain temporary qubits, restoring their value to zero and obviating the need to explicitly measure them. However, given a direct translation that eliminates the purity annotations from the program in~\Cref{fig:teleport-fixed}, Silq will not uncompute \texttt{q12} unless the user modifies the program. The reason why Silq does not automatically uncompute \texttt{q12} when it drops from scope is that Silq does not automatically uncompute qubits affected by gates with phase-level effects such as Hadamard, as is the case for \texttt{q12} in \texttt{teleport}.

To write this example in Silq, the developer may either 1) manually invoke a \texttt{measure} operator on \texttt{q12}, 2) manually invoke a \texttt{forget} operator on \texttt{q12}, or 3) apply the \texttt{qfree} annotation to the sequence of gates that produced it. In the first two cases, neither \texttt{measure} nor \texttt{forget} guarantees that its argument is separable from the remaining computation, meaning that invoking either operator may in general leave the computation in a mixed state if the developer does not accurately reason about the purity of the computation. With the third strategy, applying the \texttt{qfree} annotation invokes the requirement of the Silq type system that \texttt{qfree} functions do not use gates such as Hadamard that have phase effects, disqualifying teleportation and programs such as the efficient quantum adder of~\citet{draper2000addition}. Silq could in principle be extended with an unsafe cast to \texttt{qfree} that would permit this example and other programs with phase effects, with the proviso that the developer must themselves accurately reason about the safety of such casts in general.

By contrast, \LangName{} leverages the guarantee made by its purity type system that the discarded value cannot impact the remaining computation, resulting in a concise and safe program.

\begin{figure}
\begin{lstlisting}[numbers=left]
fun teleport (q1 : qubit<P>) : qubit<P> =
  let (q2 : qubit<M>, q3 : qubit<M>) = bell_pair () in
  let (q1 : qubit<M>, q2 : qubit<M>) = CNOT (q1, q2) in
  let q1 : qubit<M> = H (q1) in
  let (q2 : qubit<M>, q3 : qubit<M>) = CNOT (q2, q3) in
  let (q1 : qubit<M>, q3 : qubit<M>) = CZ (q1, q3) in
  let q123 : ((qubit & qubit) & qubit)<M> = ((q1, q2), q3) in
  let (q12 : (qubit & qubit)<P>, q3 : qubit<P>) = split<P>(cast<P>(q123)) in
  q3 (* safely discard pure expression q12 *)
\end{lstlisting}
\caption{Concise teleportation program in \LangName{} that implicitly discards a pure value \texttt{q12} on line 9.} \label{fig:teleport-discard}
\end{figure}

\section{\CoreLangHeading{} Language} \label{sec:core-language}

To formally define evaluation and purity in quantum programs, we present \CoreLang{}, a small quantum language. \CoreLang{} is a functional language featuring classical control and quantum data, in the style of the linear quantum $\lambda$-calculus of~\citet{quantum-lambda-calc}.

In this section, we present \CoreLang{}'s syntax and operational semantics. In~\Cref{sec:semantic-properties}, we develop its semantic properties, including purity, in terms of executions in the operational semantics. In~\Cref{sec:core-denotational}, we present \CoreLang{}'s denotational semantics, which provides a more concise definition of purity useful to formulate \LangName{}'s purity assertions. In~\Cref{sec:language}, we develop the language \LangName{} by adding to \CoreLang{} the purity assertions that reason about purity within the language itself.

\subsection{Syntax}

The syntax of \CoreLang{} consists of types, expressions, and programs.
\begin{align*}
    \textsf{Type}\ \type \Coloneqq{} & \tBool \mid \tPair{\type_1}{\type_2} \mid \type_1 \to \type_2 \mid \qubit \\
    \textsf{Expression}\ e \Coloneqq{} & x \mid f \mid \App{e_1}{e_2} \mid (e_1, e_2) \mid \Let{x}{y}{e_1}{e_2} \mid \If{e}{e_1}{e_2} \mid \True \mid \False \\
    \mid{} & \Qinit\ () \mid \Uone{e} \mid \Utwo{e} \mid \Measure{e} \mid \Qref{\alpha} \\
    \textsf{Program}\ m \Coloneqq{} & \Fun{f}{x}{e}; m \mid \Main{e}
\end{align*}
\CoreLang{} features the types $\type$ of classical Booleans, pairs, functions, and qubits.
Its expressions $e$ include the classical constructs of variables $x$, functions $f$, applications, pairs, and \texttt{let}\footnote{We also use \LetOne{x}{e_1}{e_2} as syntactic sugar for \Let{x}{\_}{(e_1, \True)}{e_2} (where \True{} is arbitrarily chosen).} and \texttt{if}-expressions. Boolean literals, which we denote with metavariable $b$, are \True{} and \False{}.

Other expressions interact with the quantum state. The operator $\Qinit\ ()$ creates a new qubit initialized to $\ket{0}$. The operator $\Uone{e}$ applies a single-qubit unitary to qubit $e$, and $\Utwo{e}$ applies a two-qubit unitary gate to a pair $e$ of two qubits.\footnote{For simplicity, we do not represent larger gates, which can be decomposed into single- and two-qubit operators using constructions like~\citet{Kitaev_1997}.}
The operator $\Measure{e}$ performs a quantum measurement of $e$, returning the classical outcome. The last operator \Qref{\alpha} is a \emph{qubit reference} that only appears in intermediate evaluations, as \CoreLang{} programs do not expose concrete qubit names. Finally, every program $m$ is a sequence of function declarations followed by a \texttt{main} function.

\subsection{Type System}

\Cref{fig:core-types} presents the judgment $\hastype{\ctx}{\rctx}{e}{\type}$ that assigns the expression $e$ the type $\type$ given the \emph{context} $\ctx$ and the \emph{qubit context} $\rctx$.
A context $\Gamma$ maps variables to types, and a qubit context $\rctx$ is a set of allocated qubits. We define the \emph{classical types} to be Booleans, functions, and pairs of classical types. Both contexts are linear, except for classical types, which may be freely duplicated or discarded.\footnotemark{} To ensure that qubits are not duplicated, we use the disjoint set union $\sqcup$, defined only when its arguments are disjoint.
\footnotetext{To prevent qubits from being duplicated or discarded due to the quantum no-cloning~\citep{Wootters1982ASQ} and no-deleting theorems~\citep{Pati2000ImpossibilityOD}, the type system of \CoreLang{} does not allow the structural rules of contraction or weakening for quantum data.}
We populate the initial context with types of function declarations using the judgment $\ProgOk{\ctx}{m}$ (defined in~\Cref{sec:app-semantics}), stating that a program $m$ is well-formed.

The typing rules for variables, functions, applications, pairs, and \texttt{let}-expressions are standard as in a linear $\lambda$-calculus.
We define the \emph{quantum types} to be qubits and pairs of quantum types.
The rule for \texttt{if} imposes a condition that the type of the branches is a quantum type, which will simplify our presentation. Booleans have Boolean type and $\Qinit\ ()$ has qubit type.
Unitary operators operate on one qubit or a pair of qubits, and have the same type as their argument.
Measurement of a qubit results in a classical Boolean type.\footnote{We also allow measurement of tuples as syntactic sugar for sequentially measuring each component.} Qubit references have qubit type.
\begin{figure}
\resizebox{\textwidth}{!}{
\parbox{1.1\textwidth}{
\begin{mathpar}
    \inferrule[T-Var]{\vphantom{\ctx}}{\hastype{x : \type}{\emptyset}{x}{\type}}

    \inferrule[T-Fun]{\vphantom{\ctx}}{\hastype{f : \type \to \type'}{\emptyset}{f}{\type \to \type'}}

    \inferrule[T-App]{\hastype{\ctx_1}{\rctx_1}{e_1}{\type_1 \to \type_2} \\ \hastype{\ctx_2}{\rctx_2}{e_2}{\type_1}}{\hastype{\ctx_1, \ctx_2}{\rctx_1 \sqcup \rctx_2}{\App{e_1}{e_2}}{\type_2}}

    \inferrule[T-Pair]{\hastype{\ctx_1}{\rctx_1}{e_1}{\type_1} \\ \hastype{\ctx_2}{\rctx_2}{e_2}{\type_2}}{\hastype{\ctx_1, \ctx_2}{\rctx_1 \sqcup \rctx_2}{(e_1, e_2)}{\tPair{\type_1}{\type_2}}}

    \inferrule[T-Let]{\hastype{\ctx_1}{\rctx_1}{e_1}{\tPair{\type_1}{\type_2}} \\ \hastype{\ctx_2, x : \type_1, y : \type_2}{\rctx_2}{e_2}{\type}}{\hastype{\ctx_1, \ctx_2}{\rctx_1 \sqcup \rctx_2}{\Let{x}{y}{e_1}{e_2}}{\type}}

    \inferrule[T-If]{\hphantom{\hastype{\ctx_1}{\rctx_1}{e}{\tBool}} \\ \text{$\type$ is a quantum type} \\\\ \hastype{\ctx_1}{\rctx_1}{e}{\tBool} \\ \hastype{\ctx_2}{\rctx_2}{e_1}{\type} \\ \hastype{\ctx_2}{\rctx_2}{e_2}{\type}}{\hastype{\ctx_1, \ctx_2}{\rctx_1 \sqcup \rctx_2}{\If{e}{e_1}{e_2}}{\type}}

    \inferrule[T-Bool]{\vphantom{\ctx}}{\hastype{\cdot}{\emptyset}{b}{\tBool}}

    \inferrule[T-Qinit]{\vphantom{\ctx}}{\hastype{\cdot}{\emptyset}{\Qinit\ ()}{\qubit}} \\

    \inferrule[T-U1]{\hastype{\ctx}{\rctx}{e}{\qubit}}{\hastype{\ctx}{\rctx}{\Uone{e}}{\qubit}}

    \inferrule[T-U2]{\hastype{\ctx}{\rctx}{e}{\tPair{\qubit}{\qubit}}}{\hastype{\ctx}{\rctx}{\Utwo{e}}{\tPair{\qubit}{\qubit}}}

    \inferrule[T-Measure]{\hastype{\ctx}{\rctx}{e}{\qubit}}{\hastype{\ctx}{\rctx}{\Measure{e}}{\tBool}}

    \inferrule[T-Ref]{\vphantom{\ctx}}{\hastype{\cdot}{\{\alpha\}}{\Qref{\alpha}}{\qubit}}
\end{mathpar}
}
}
\caption{\CoreLang{} type system.} \label{fig:core-types}
\end{figure}

\subsection{Operational Semantics}

\Cref{fig:core-dynamics} presents an operational semantics for \CoreLang{} as a probabilistic transition system over \emph{program states} -- pairs of quantum states $\ket{\qstate}$ and classical expressions $e$.\footnote{Any program state can take exactly one or two transitions, with probabilities adding to one. For a formal model of probability in our semantics, see the probabilistic reduction system introduced by \citet{quantum-lambda-calc}.} The semantics depend on particular measurement outcomes described by an \emph{outcome map} $O$. An outcome map is a set of pairs $(\alpha, b)$, each meaning that qubit $\alpha$ was measured, and classical outcome $b$ was observed for it.

\paragraph{Judgments}
The value judgment $\val{v}$ states that functions, Boolean literals, qubit references, and pairs of values are values.
The step judgment $\stepo{\ket{\qstate}}{e}{O}{p}{\ket{\qstate'}}{e'}$ states that the expression $e$ under state $\ket{\qstate}$ steps to $e'$ and new state $\ket{\qstate'}$ after observing the measurement outcomes in $O$ with probability $p$.
The evaluation judgment $\stepostar{\ket{\qstate}}{e}{O}{p}{\ket{\qstate'}}{v}$ states that expression $e$ under state $\ket{\qstate}$ evaluates to a value $v$ and state $\ket{\qstate'}$ having observed outcomes $O$ with probability $p$.

\paragraph{Functions}
The semantics makes use of an unchanging global function context $\phi$ that maps function names $f$ to definitions $\Lam{x}{e}$. A program executes by collecting all function definitions into $\phi$ and then evaluating the body of the \texttt{main} function until it reaches a value.

\begin{figure}
\resizebox{\textwidth}{!}{
\parbox{1.1\textwidth}{
\begin{mathpar}
    \inferrule[S-App]{\phi(f) = \Lam{x}{e} \\ \val{e'}}{\stepo{\ket{\qstate}}{\App{f}{e'}}{\emptyset}{1}{\ket{\qstate}}{[e'/x]e}}

    \inferrule[S-Let]{\val{e_1} \\ \val{e_2}}{\stepo{\ket{\qstate}}{\Let{x}{y}{(e_1, e_2)}{e'}}{\emptyset}{1}{\ket{\qstate}}{[e_1, e_2/x, y]e'}}

    \inferrule[S-Qinit]{\alpha \text{ fresh in } \ket{\qstate}}{\stepo{\ket{\qstate}}{\Qinit\ ()}{\emptyset}{1}{\ket{\qstate} \otimes \ket{0}_\alpha}{\Qref{\alpha}}}

    \inferrule[S-U1]{\vphantom{\ctx}}{\stepo{\ket{\qstate}}{\Uone{\Qref{\alpha}}}{\emptyset}{1}{U_\alpha\ket{\qstate}}{\Qref{\alpha}}}

    \inferrule[S-U2]{\vphantom{\ctx}}{\stepo{\ket{\qstate}}{\Utwo{\Qref{\alpha}, \Qref{\beta}}}{\emptyset}{1}{U_{\alpha, \beta}\ket{\qstate}}{(\Qref{\alpha}, \Qref{\beta})}} \\

    \inferrule[S-IfT]{\vphantom{\ctx}}{\stepo{\ket{\qstate}}{\If{\True}{e_1}{e_2}}{\emptyset}{1}{\ket{\qstate}}{e_1}}

    \inferrule[S-IfF]{\vphantom{\ctx}}{\stepo{\ket{\qstate}}{\If{\False}{e_1}{e_2}}{\emptyset}{1}{\ket{\qstate}}{e_2}}

    \inferrule[S-MeasureT]{M_\alpha\ket{\qstate} = \ket{1}_\alpha \otimes \ket{\qstate'} \text{ w.p. } p \\ O = \{(\alpha, \True)\}}{\stepo{\ket{\qstate}}{\Measure{\Qref{\alpha}}}{O}{p}{\ket{\qstate'}}{\True}}

    \inferrule[S-MeasureF]{M_\alpha\ket{\qstate} = \ket{0}_\alpha \otimes \ket{\qstate'} \text{ w.p. } p \\ O = \{(\alpha, \False)\}}{\stepo{\ket{\qstate}}{\Measure{\Qref{\alpha}}}{O}{p}{\ket{\qstate'}}{\False}} \\

    \inferrule[E-Val]{\val{v}}{\stepostar{\ket{\qstate}}{v}{\emptyset}{1}{\ket{\qstate}}{v}}

    \inferrule[E-Step]{\stepo{\ket{\qstate}}{e}{O_1}{p_1}{\ket{\qstate'}}{e'} \\ \stepostar{\ket{\qstate'}}{e'}{O_2}{p_2}{\ket{\qstate''}}{v} \\ O' = O_1 \cup O_2}{\stepostar{\ket{\qstate}}{e}{O'}{p_1 p_2}{\ket{\qstate''}}{v}}
\end{mathpar}
}
}
\caption{Selected \CoreLang{} operational semantics rules. The full rules are given in~\Cref{sec:app-semantics}.} \label{fig:core-dynamics}
\end{figure}

\paragraph{Operators}
\Cref{fig:core-dynamics} presents a selection of the rules, with the remaining left-to-right call-by-value rules presented in~\Cref{sec:app-semantics}.
The first two rules step application and \texttt{let}-expressions in the standard way, with no effect on $\ket{\qstate}$.
The $\Qinit\ ()$ operator adds a new qubit named $\alpha$ in state $\ket{0}$ to $\ket{\qstate}$ and steps to a reference to $\alpha$.
A single-qubit unitary operator on a qubit reference to $\alpha$ steps to its argument and a new state with $U$ applied on qubit $\alpha$. Similarly, a two-qubit unitary operator applied on qubits $\alpha$ and $\beta$ steps to its argument and a new state with $U$ applied on $\alpha$ and $\beta$.
An \texttt{if}-expression chooses a branch to step to depending on the condition.
Measurement of a qubit has two probabilistic outcomes, classical true or false, with a new state under each outcome. The step rules for measurement transition to each outcome with its occurrence probability (\Cref{sec:background}).

\paragraph{Notation} We define the following new notations for judgments.

\begin{itemize}
    \item $\canstep{\ket{\qstate}}{e}$ means that there exists a set $I$ where for each $i \in I$, there exist $p_i > 0$, $O_i$, $\ket{\qstate_i}$ and $e_i$ where $\stepo{\ket{\qstate}}{e}{O_i}{p_i}{\ket{\qstate_i}}{e_i}$ and $\sum_i p_i = 1$.
    \item $\step{\ket{\qstate}}{e}{}{\ket{\qstate'}}{e'}$ means that there exist $O$ and $p$ where $\stepo{\ket{\qstate}}{e}{O}{p}{\ket{\qstate'}}{e'}$.
    \item $\stepstar{\ket{\qstate}}{e}{}{\ket{\qstate'}}{e'}$ means that there exist $O$ and $p$ where $\stepostar{\ket{\qstate}}{e}{O}{p}{\ket{\qstate'}}{e'}$.
\end{itemize}

\subsection{Type Safety}

We now state the type safety properties of \CoreLang{}.
Progress states that a well-typed expression is a value or can step under a quantum state containing all qubits that the expression references.

\begin{theorem}[Progress] \label{thm:core-progress}
    If $\hastype{\cdot}{\rctx}{e}{\type}$, then $\val{e}$ or for all $\ket{\qstate}$ where $\rctx \subseteq \domain{\ket{\qstate}}$, $\canstep{\ket{\qstate}}{e}$.
\end{theorem}

The proof is by induction on the derivation of $\hastype{\cdot}{\rctx}{e}{\type}$.
The preservation theorem states that a step preserves the type of the expression under the new qubit context.

\begin{theorem}[Preservation] \label{thm:core-preservation}
    If $\hastype{\ctx}{\rctx}{e}{\type}$ and $\rctx \subseteq \domain{\ket{\qstate}}$ and $\step{\ket{\qstate}}{e}{}{\ket{\qstate'}}{e'}$, then we have $\hastype{\ctx}{\rctx'}{e'}{\type}$ where $\rctx' \subseteq \domain{\ket{\qstate'}}$.
\end{theorem}

The proof is by induction on the derivation of $\step{\ket{\qstate}}{e}{}{\ket{\qstate'}}{e'}$.

\section{Semantic Properties} \label{sec:semantic-properties}

In this section, we define the semantic property of purity using the operational semantics of \CoreLang{}.
Purity states that an expression executes independently of measurement outcomes of unowned qubits, or equivalently, the qubits owned by a pure expression are separable from those it does not own. Formally, purity states that there is only one possible final program state after evaluating an expression and measuring all qubits that it does not own.

\subsection{Implicit Measurement}

An expression $e$ evaluates under a quantum state $\ket{\qstate}$ to a value $v$ and state $\ket{\qstate'}$. The resulting state $\ket{\qstate'}$ may contain qubits to which $v$ does not refer. If these qubits are measured later in the program, then their measurement outcomes may affect the state of the qubits in $v$ through entanglement.

To capture the effect of the eventual measurement of unowned qubits, we define a relation called \emph{implicit measurement}. Given a state $\ket{\qstate}$ and a value $v$, implicit measurement measures all qubits in $\ket{\qstate}$ to which $v$ does not refer. Define $\Refs{v}$ to be the sequence, or set when order is irrelevant, of qubit names referenced in $v$.
\begin{definition}[Implicit Measurement]
    A state $\ket{\qstate}$ \emph{implicitly measures} modulo a value $v$ to produce a new state $\ket{\qstate_v}$ with probability $p$, written $\imeas{\ket{\qstate}}{v}{p}{\ket{\qstate_v}}$, when $M_{A}\ket{\qstate} = \ket{\qstate_A} \otimes \ket{\qstate_v}$ with probability $p$ where $A = \domain{\ket{\qstate}} \setminus \Refs{v}$ and $\domain{\ket{\qstate_v}} = \Refs{v}$.
\end{definition}
We use notation $\imeas{\ket{\qstate}}{v}{}{\ket{\qstate'}}$ to denote that there exists some $p$ such that $\imeas{\ket{\qstate}}{v}{p}{\ket{\qstate'}}$.
A value has a unique implicit measurement in $\ket{\qstate}$ if and only if its referenced qubits are separable in $\ket{\qstate}$.

\subsection{Qubit Equivalence}

We consider two program states $(\ket{\qstate_1}, v_1)$ and $(\ket{\qstate_2}, v_2)$ to be equivalent if they are equal up to consistent renaming of qubits. For example, we define
\[ (\ket{0}_\alpha \otimes \ket{1}_\beta, (\Qref{\alpha}, \Qref{\beta})) \vequiv (\ket{1}_\gamma \otimes \ket{0}_\delta, (\Qref{\delta}, \Qref{\gamma})) \]
because they can be substituted for each other in a program with no operational effect.

\begin{definition}[Qubit Equivalence]
$(\ket{\qstate_1}, v_1) \vequiv (\ket{\qstate_2}, v_2)$ holds when $\ket{\qstate_2}[\ell_1 / \ell_2] = \ket{\qstate_1}$ and $v_2[\ell_1 / \ell_2] = v_1$ for two duplicate-free sequences of qubit names $\ell_1, \ell_2$ of equal length.
\end{definition}

$\vequiv$ is reflexive, symmetric, and transitive, making it an equivalence relation.
We define the action of rewriting using qubit equivalence as replacing a value $v_1$ of quantum type with any qubit-equivalent value $v_2$ of the same type by also modifying the quantum state correspondingly.

\subsection{Purity}

The \emph{purity} property states that a state and expression always evaluate and implicitly measure to a unique final state and value up to qubit equivalence.
\begin{definition}[Purity] \label{def:purity}
    An expression $e$ is \emph{pure} under state $\ket{\qstate}$, denoted $\sempure{e}{\ket{\qstate}}$, when if $\stepstar{\ket{\qstate}}{e}{}{\ket{\qstate_1}}{v_1}$ and $\imeas{\ket{\qstate_1}}{v_1}{}{\ket{\qstate'_1}}$, and also $\stepstar{\ket{\qstate}}{e}{}{\ket{\qstate_2}}{v_2}$ and $\imeas{\ket{\qstate_2}}{v_2}{}{\ket{\qstate'_2}}$, then we have $(\ket{\qstate'_1}, v_1) \vequiv (\ket{\qstate'_2}, v_2)$.
\end{definition}
Purity asserts that under state $\ket{\qstate}$, expression $e$ evaluates to a unique value $v$ and final state $\ket{\qstate'}$ where $v$ has a unique implicit measurement in $\ket{\qstate'}$. This definition formalizes the intuition that the eventual measurement outcome of unowned qubits cannot affect the state of those it owns.

\section{Denotational Semantics of \CoreLangHeading{}} \label{sec:core-denotational}

The operational definition of purity quantifies over all possible executions of a \CoreLang{} program. In this section, we present the denotational semantics for \CoreLang{}, which enables a more concise definition of purity based on the closed-form denotation of an expression.
Denotational semantics reasons directly about an expression's effect on the distribution of the program state across all executions, using mixed states represented as density matrices to describe distributions over pure states.

\paragraph{Eliminating Nondeterminism}
A \CoreLang{} program executes nondeterministically in terms of both its quantum and classical state. An expression may evaluate to multiple distinct values when measurement outcomes influence classical control flow. Because this nondeterminism in the program's classical state complicates the development of a denotational semantics, we instead force all nondeterminism to occur in the quantum state.

Any expression of function type always evaluates to a unique value, by a simple induction argument over the operational semantics. For Booleans, we utilize the deferred measurement principle (\Cref{sec:background}) to interpret Booleans as deferred qubit measurements that are not resolved until the end of the program.
For quantum types, we canonicalize values that differ only in the order of appearance of qubits. To do so, we use rewriting under qubit equivalence (\Cref{sec:semantic-properties}) to dynamically rename qubits so that every value refers to the same qubits in the same order.

\subsection{Semantics}

\Cref{fig:core-denotation} presents the denotational semantics of \CoreLang{}. The denotation of an expression $e$ is a function from a context $\gamma$ mapping variables to values and an input partial density matrix $\rho$ to the final partial density matrix $\rho'$ and value $v$ to which $e$ evaluates. The denotation of a program $m$ adds functions to the initial $\gamma$ and is always a normalized density matrix.

\paragraph{Basic Operators}
The values of variables reside in the context $\gamma$. Values do not evaluate further. Pairs and \texttt{let}-bindings propagate the state through their evaluation. The denotation of function application is the application of $f$, the denotation of $e_1$, to the denotation of $e_2$. The denotation of $\Qinit\ ()$ adds a new qubit $\alpha$ in state $\ket{0}$ by tensor product with its density matrix form, which is an outer product $\dyad{0}$. The denotation of unitary operators performs matrix conjugation (\Cref{sec:background}).

\paragraph{Measurement and Conditional Branches}
The denotations of \texttt{if} and \texttt{measure} encode deferred measurement. A measurement has no immediate effect and simply evaluates its argument, where the notation $\QrefM{\alpha}$ differentiates a measured from an unmeasured qubit. An \texttt{if}-expression first determines whether its condition is a literal \True{} or \False{} and if so executes an appropriate branch.

Otherwise, the condition is a deferred measurement of a qubit $\alpha$.
The denotation computes the partial density matrices corresponding to each outcome using the matrix $P_\alpha$ projecting qubit $\alpha$ to $\ket{1}$ in $\rho$. It executes the corresponding branches to obtain matrices $\rho_1$ and $\rho_2$. It next corrects their dimensions using \textsf{match\_sizes}, defined as taking tensor product with copies of $\dyad{0}$ on the smaller matrix. This operation allocates the same number of qubits created by \texttt{qinit} to both branches.

The values $v_1$ and $v_2$ returned by the two branches are of the same type but may refer to different qubits. The denotation unifies $v_1$ and $v_2$ by renaming each qubit in $v_2$ and $\rho'_2$ into the corresponding reference in $v_1$. It adds the two partial density matrices to weigh each outcome by its probability.

\begin{figure}
\footnotesize
\begin{align*}
\denote{\Fun{f}{x}{e}; m}_\gamma ={}& \text{let } \gamma' = \gamma[\lambda v. \denote{e}_{\gamma[v/x]} / f] \text{ in} \denote{m}_{\gamma'} \\
\denote{\Main{e}}_\gamma ={}& \Denote{e}{\gamma}{\dyad{\cdot}} \\[0.5em]
\Denote{x}{\gamma}{\rho} ={}& (\rho, \gamma(x)) \\
\Denote{v}{\gamma}{\rho} ={}& (\rho, v) \text{ when } \val{v} \\
\Denote{(e_1, e_2)}{\gamma}{\rho} ={}& \textrm{let}\ \rho_1, v_1 = \Denote{e_1}{\gamma}{\rho}\textrm{in}\ \textrm{let}\ \rho_2, v_2 = \Denote{e_2}{\gamma}{\rho_1}\textrm{in}\ (\rho_2, (v_1, v_2)) \\
\Denote{\Let{x}{y}{e_1}{e_2}}{\gamma}{\rho} ={}& \textrm{let}\ \rho_1, (v_1, v_2) = \Denote{e_1}{\gamma}{\rho}\textrm{in}\Denote{e_2}{\gamma[v_1, v_2 / x, y]}{\rho_1} \\
\Denote{\App{e_1}{e_2}}{\gamma}{\rho} ={}& \textrm{let}\ \rho_1, f = \Denote{e_1}{\gamma}{\rho}\textrm{in}\ \textrm{let}\ \rho_2, v = \Denote{e_2}{\gamma}{\rho_1}\textrm{in}\ f(v)(\rho_2) \\
\Denote{\Qinit\ ()}{\gamma}{\rho} ={}& (\rho \otimes \dyad{0}_\alpha, \Qref{\alpha}) \text{ where $\alpha$ is fresh in $\rho$} \\
\Denote{\Uone{e}}{\gamma}{\rho} ={}& \textrm{let}\ \rho', \Qref{\alpha} = \Denote{e}{\gamma}{\rho}\textrm{in}\ (U_{\alpha} \rho' U^\dagger_{\alpha}, \Qref{\alpha}) \\
\Denote{\Utwo{e}}{\gamma}{\rho} ={}& \textrm{let}\ \rho', (\Qref{\alpha}, \Qref{\beta}) = \Denote{e}{\gamma}{\rho}\textrm{in}\ (U_{\alpha, \beta} \rho' U^\dagger_{\alpha, \beta}, (\Qref{\alpha}, \Qref{\beta})) \\
\Denote{\Measure{e}}{\gamma}{\rho} ={}& \textrm{let}\ \rho', \Qref{\alpha} = \Denote{e}{\gamma}{\rho}\textrm{in}\ (\rho', \QrefM{\alpha}) \\
\Denote{\If{e}{e_1}{e_2}}{\gamma}{\rho} ={}& \textrm{let}\ \rho', v = \Denote{e}{\gamma}{\rho} \textrm{in\ case\ } v \textrm{ of } \True \to \Denote{e_1}{\gamma}{\rho'} \mid \False \to \Denote{e_2}{\gamma}{\rho'} \mid \QrefM{\alpha} \to \\
& \text{let } P_\alpha = \dyad{1}_\alpha \otimes I_{\domain{\rho} \setminus \{\alpha\}} \text{ in} \\
& \text{let } \rho_1, v_1 = \denote{e_1}_\gamma({P_\alpha \rho P_\alpha}) \text{ and } \rho_2, v_2 = \denote{e_2}_\gamma({(I_{\domain{\rho}} - P_\alpha) \rho (I_{\domain{\rho}} - P_\alpha)}) \text{ in} \\
& \text{let } \rho'_1, \rho'_2 = \textsf{match\_sizes}(\rho_1, \rho_2) \text{ in}\ (\rho'_1 + \rho'_2[\Refs{v_1} / \Refs{v_2}], v_1)
\end{align*}
\caption{\CoreLang{} denotational semantics.}
\label{fig:core-denotation}
\end{figure}

\subsection{Semantics Equivalence}

In this section, we show that the operational and denotational semantics are equivalent in terms of final program states under particular measurement outcomes.
Let $\phi$ denote the initial context containing only top-level function declarations.
Given an outcome map $O$ and a value $v$, we define $\textsf{apply}(O, v)$ to be the value syntactically identical to $v$ except that every instance of $\QrefM{\alpha}$ is replaced with $b$ where $(\alpha, b) \in O$.
In the opposite direction, we define $\textsf{defer}(O)$ to be the tensor product of all outcomes in $O$ expressed as quantum states:\nowidow\vspace*{-\baselineskip}
\begin{align*}
    \textsf{defer}(O) = \bigotimes \left\{\text{if $b = \True{}$ then $\ket{1}_\alpha$ else $\ket{0}_\alpha$} \mid (\alpha, b) \in O\right\}
\end{align*}
The following theorem states that the denotation of an expression captures every execution up to qubit equivalence. $\rho$ contains all final $\ket{\qstate_i}$ and outcomes $O_i$ and some number of unused padding qubits, and deferred and operational measurement outcomes align in the final value $v$.
\begin{theorem} \label{thm:core-sem-equiv}
    Given $\ket{\qstate}$ and $e$, let the multiset $S = [ (p_i, \ket{\qstate_i}, v_i, O_i) \mid \stepostar{\ket{\qstate}}{e}{O_i}{p_i}{\ket{\qstate_i}}{v_i} ]$. Then,
    \begin{align*}
        \Denote{e}{\phi}{\dyad{\qstate}} = (\rho, v)
        \text{ where }& \rho = \textstyle\sum_{(p_i, \ket{\qstate_i}, v_i, O_i) \in S} p_i \dyad{\qstate'_i}
        \text{and } (\ket{\qstate'_i}, v) \vequiv (\ket{\qstate''_i}, v') \\
        \text{ where }& \ket{\qstate_i} \otimes \textsf{defer}(O_i) \otimes \ket{0}^* = \ket{\qstate''_i}
        \text{and } v_i = \textsf{apply}(O_i, v').
    \end{align*}
\end{theorem}

The proof is by induction on the derivations of $\stepostar{\ket{\qstate}}{e}{O_i}{p_i}{\ket{\qstate_i}}{v_i}$. Most operators in the denotational semantics follow a left-to-right eager evaluation of the arguments and return the same value as the step rule in the operational semantics. The interesting cases are for unitary operators, \texttt{if}-expressions, and measurement, where we appeal to the correspondence between the state vector and density matrix models of quantum mechanics as well as the principle of deferred measurement.

\subsection{Purity}

The following corollary provides an equivalent definition of purity using denotational semantics.

\begin{corollary} \label{thm:purity-denotation}
    $\sempure{e}{\ket{\qstate}}$ holds if and only if $ \Denote{e}{\phi}{\dyad{\qstate}} = (\rho \otimes \rho', v) $ where $\domain{\rho} = \Refs{v}$ and $\rho$ is a pure state, i.e. $\rho = \dyad{\qstate'}$ for some $\ket{\qstate'}$.
\end{corollary}

Thus, computing the denotation of $e$ under $\ket{\qstate}$ gives us a direct way of testing whether $\sempure{e}{\ket{\qstate}}$. We leverage this fact in the following sections to define the \emph{purity assertion} operators of \LangName{}.

\section{\LangName{} Language} \label{sec:language}

In this section, we present the formal core of the \LangName{} language. \LangName{} extends \CoreLang{} with a \emph{purity type} system that specifies which expressions are pure. We present the two \emph{purity assertion} operators that specify and check purity in a quantum program.

We present both a denotational and operational semantics of \LangName{}. Denotationally, we define purity assertions using mathematical \emph{separability conditions} on mixed-state denotations. Operationally, we implement them using a \emph{separability test} primitive on runtime pure states.

\subsection{Syntax}

To develop \LangName{}, we augment the syntax of \CoreLang{} as follows, where new syntax is in black:\footnote{We describe in~\Cref{sec:higher-level} additional syntactic features of \LangName{}, such as implicit measurement of discarded pure values.}
\begin{align*}
    \textsf{Purity}\ \purity \Coloneqq{} & \pure \mid \mixed \\
    \textsf{Quantum type}\ \qtype \Coloneqq{} & \textcolor{gray}{\qubit} \mid \tEpair{\qtype_1}{\qtype_2} \\
    \textsf{Type}\ \type \Coloneqq{} & \textcolor{gray}{\tBool \mid \tPair{\type_1}{\type_2} \mid \type_1 \to \type_2} \mid \tPurity{\qtype}{\purity} \\
    \textsf{Quantum value}\ q \Coloneqq{} & \textcolor{gray}{\Qref{\alpha}} \mid \Epair{q_1}{q_2} \\
    \textsf{Expression}\ e \Coloneqq{} & \textcolor{gray}{x \mid f \mid \App{e_1}{e_2} \mid (e_1, e_2) \mid \Let{x}{y}{e_1}{e_2} \mid \If{e}{e_1}{e_2} \mid \True \mid \False} \\
    \textcolor{gray}{\mid{}} & \textcolor{gray}{\Qinit\ () \mid \Uone{e} \mid \Utwo{e} \mid \Measure{e}} \mid \Qpurity{q}{\purity}\\
    \mid{} & \EntangleOne{\purity}{e} \mid \SplitA{\purity}{e} \mid \Cast{\purity}{e}
\end{align*}
We introduce purity annotations $\purity$, either pure, $\pure$, or mixed, $\mixed$.
We also introduce a sort of \emph{quantum types} $\qtype$,\footnote{Lowercase Greek letter qoppa.} either a \qubit{} or an \emph{entangled pair} $\tEpair{\qtype_1}{\qtype_2}$ of two quantum types.
They denote \emph{quantum values} $q$, either a reference to a qubit or an entangled pair $\Epair{q_1}{q_2}$ of two quantum values.

The new forms of expressions manipulate purity and have no effect on the quantum state. A purity-annotated quantum value $\Qpurity{q}{\purity}$ is an expression that only arises in intermediate evaluations.
The $\EntangleOne{\purity}{e}$ operator constructs an entangled pair from its argument, the $\SplitA{\purity}{e}$ operator destructs an entangled pair into two components, and $\Cast{\purity}{e}$ modifies the purity of an expression.

\subsection{Type System}

\Cref{fig:types} presents the type system of \LangName{}, extending the type system of \CoreLang{} to the new operators.
The new judgment $\qhastype{\rctx}{q}{\qtype}$ assigns the quantum value $q$ the quantum type $\qtype$ under the context $\rctx$. A reference to qubit $\alpha$ has qubit type under a context only containing $\alpha$, and an entangled pair has entangled pair type if its components have quantum types.

We next modify the typing judgment $\hastype{\ctx}{\rctx}{e}{\type}$, showing only the rules that change, with the remainder in~\Cref{sec:app-semantics}.
The rule for \texttt{if} requires its branches to have the same quantum type and now returns a mixed version of that type.
The $\Qinit\ ()$ operator returns a pure qubit. Unitary operators operate on a qubit or entangled pair of any purity and have the same purity as their argument.
The \Measure{e} operator accepts a qubit of any purity and returns a Boolean.
A purity-annotated quantum value $\Qpurity{q}{\purity}$ has type $\tPurity{\qtype}{\purity}$ if $q$ has quantum type $\qtype$.

\paragraph{Entangled Pairs and Purity Assertions}

The purpose of the entangled pair type is to denote values that may have been entangled by a two-qubit unitary operator.
The $\mathtt{entangle}_{\purity}$ operator creates an entangled pair of purity $\purity$ from an ordinary pair of quantum values of that same purity.

When a program extracts the two components of an entangled pair of any purity, conservatively the type system assumes they became entangled and now constitute mixed states, and assigns them mixed type.
The operator $\SplitA{\mixed}{e}$ destructs a mixed entangled pair into two mixed components, and $\Cast{\mixed}{e}$ sets the purity of any expression to mixed.

Two \emph{purity assertion} operators obtain pure types for quantum expressions. The \emph{purifying-split} operator $\SplitP$ destructs a pure entangled pair into two pure components, and the \emph{purifying-cast} operator $\CastP$ sets the purity of any expression to pure.
To ensure that expressions of pure type are actually pure, the semantics of the purity assertions impose conditions on their usage.
\begin{figure}
\resizebox{\textwidth}{!}{
\parbox{1.1\textwidth}{
\begin{mathpar}
    \inferrule[Q-Ref]{\vphantom{\ctx}}{\qhastype{\{\alpha\}}{\Qref{\alpha}}{\qubit}}

    \inferrule[Q-Pair]{\qhastype{\rctx_1}{q_1}{\qtype_1} \\ \qhastype{\rctx_2}{q_2}{\qtype_2}}{\qhastype{\rctx_1 \sqcup \rctx_2}{\Epair{q_1}{q_2}}{\tEpair{\qtype_1}{\qtype_2}}}

    \inferrule[T-If]{\hastype{\ctx_1}{\rctx_1}{e}{\tBool} \\ \hastype{\ctx_2}{\rctx_2}{e_1}{\tPurity{\qtype}{\purity}} \\ \hastype{\ctx_2}{\rctx_2}{e_2}{\tPurity{\qtype}{\purity}}}{\hastype{\ctx_1, \ctx_2}{\rctx_1 \sqcup \rctx_2}{\If{e}{e_1}{e_2}}{\tPurity{\qtype}{\mixed}}}

    \inferrule[T-Qinit]{\vphantom{\ctx}}{\hastype{\cdot}{\emptyset}{\Qinit\ ()}{\tPurity{\qubit}{\pure}}}

    \inferrule[T-U1]{\hastype{\ctx}{\rctx}{e}{\tPurity{\qubit}{\purity}}}{\hastype{\ctx}{\rctx}{\Uone{e}}{\tPurity{\qubit}{\purity}}}

    \inferrule[T-U2]{\hastype{\ctx}{\rctx}{e}{\tPurity{(\tEpair{\qubit}{\qubit})}{\purity}}}{\hastype{\ctx}{\rctx}{\Utwo{e}}{\tPurity{(\tEpair{\qubit}{\qubit})}{\purity}}}

    \inferrule[T-Measure]{\hastype{\ctx}{\rctx}{e}{\tPurity{\qubit}{\purity}}}{\hastype{\ctx}{\rctx}{\Measure{e}}{\tBool}}

    \inferrule[T-Qval]{\qhastype{\rctx}{q}{\qtype}}{\hastype{\cdot}{\rctx}{\Qpurity{q}{\purity}}{\tPurity{\qtype}{\purity}}}

    \inferrule[T-Entangle]{\hastype{\ctx}{\rctx}{e}{\tPair{\tPurity{\qtype_1}{\purity}}{\tPurity{\qtype_2}{\purity}}}}{\hastype{\ctx}{\rctx}{\EntangleOne{\purity}{e}}{\tPurity{(\tEpair{\qtype_1}{\qtype_2})}{\purity}}}

    \inferrule[T-Split]{\hastype{\ctx}{\rctx}{e}{\tPurity{(\tEpair{\qtype_1}{\qtype_2})}{\purity}}}{\hastype{\ctx}{\rctx}{\SplitA{\purity}{e}}{\tPair{\tPurity{\qtype_1}{\purity}}{\tPurity{\qtype_2}{\purity}}}}

    \inferrule[T-Cast]{\hastype{\ctx}{\rctx}{e}{\tPurity{\qtype}{\purity'}}}{\hastype{\ctx}{\rctx}{\Cast{\purity}{e}}{\tPurity{\qtype}{\purity}}}
\end{mathpar}
}
}
\caption{\LangName{} purity type system. Only rules changed from~\Cref{fig:core-types} shown; full rules in~\Cref{sec:app-semantics}.} \label{fig:types}
\end{figure}

\subsection{Denotational Semantics}

\Cref{fig:denotation} presents the denotational semantics of \LangName{}, extending the semantics of \CoreLang{} to the new operators. We define $\vdenote{v}$ to strip purity annotations and replace entangled pairs with ordinary pairs in $v$, so that the denotation of a value $v$ is $\vdenote{v}$. Operators $\texttt{entangle}_\purity$, $\CastM$, and $\SplitM$ only manipulate purity annotations, and so their denotation is simply the denotation of their argument.

\SplitP{} and \CastP{} evaluate their arguments and then assert a condition about the resulting partial density matrix. If the condition holds, the operator leaves the state unchanged. Otherwise, its denotation is the special element $\bot$ corresponding to the program aborting at runtime.

\paragraph{Separability Conditions}
Given a partial density matrix $\rho$ and a partition of its domain, a \emph{separability condition} states that $\rho$ is simply separable into sub-states whose domains are the qubit sets of the partition. The \SplitP{} operator asserts that $\rho$ is simply separable into, i.e. is the product of three sub-states, where two are pure and correspond to $q_1$ and $q_2$. The definition of \CastP{} is similar, asserting that the state has a pure sub-state corresponding to $q$.

\paragraph{Classical Control}
A purity assertion under one branch of an \texttt{if}-expression checks that an expression is pure across states satisfying that branch of the \texttt{if}-condition. The typing rule for \texttt{if} ensures that a pure value inside a branch will be considered mixed at the end of the \texttt{if}-expression. For example, suppose $x$ and $y$ are two qubits in a Bell pair. Each \CastP{} within
\[
\If{\Measure{x}}{\Cast{\pure}{y}}{\Cast{\pure}{y}} : \tPurity{\qubit}{\mixed}
\]
is valid because if $x$ was measured to be $\ket{0}$, $y$ must also be in the pure state $\ket{0}$, and measuring $\ket{1}$ for $x$ would likewise yield $\ket{1}$ for $y$. Nevertheless, the expression overall is mixed -- the type of \texttt{if} is always mixed because the type system cannot assume any particular outcome for $x$.

\begin{figure}
\footnotesize
\centering
\begin{minipage}[t]{.35\textwidth}
    \centering
    \begin{align*}
      \vdenote{v} ={}& v \text{ when $v$ is $f$, $b$, or $\Qref{\alpha}$} \\
      \vdenote{\Qpurity{q}{\purity}} ={}& \vdenote{q} \\
      \vdenote{\Epair{q_1}{q_2}} ={}& (\vdenote{q_1}, \vdenote{q_2})
      \end{align*}
  \end{minipage}%
\begin{minipage}[t]{.5\textwidth}
  \centering
  \begin{align*}
    \Denote{v}{\gamma}{\rho} ={}& (\rho, \vdenote{v}) \text{ when $\val{v}$} \\
    \denote{\EntangleOne{\purity}{e}} ={}& \denote{\Cast{\mixed}{e}} = \denote{\SplitA{\mixed}{e}} = \denote{e} \\
    \Denote{\SplitA{\pure}{e}}{\gamma}{\rho} ={}& \Denote{e}{\gamma}{\rho} \text{ if } \Denote{e}{\gamma}{\rho} = (\rho_1 \otimes \rho_2 \otimes \rho_0, (q_1, q_2)) \text{ else } \bot \\
    & \text{where } \domain{\rho_1} = \Refs{q_1} \text{ and } \domain{\rho_2} = \Refs{q_2} \\
    & \text{and $\rho_1$ and $\rho_2$ are pure} \\
    \Denote{\Cast{\pure}{e}}{\gamma}{\rho} ={}& \Denote{e}{\gamma}{\rho} \text{ if } \Denote{e}{\gamma}{\rho} = (\rho_1 \otimes \rho_0, q) \text{ else } \bot \\
    & \text{where } \domain{\rho_1} = \Refs{q} \text{ and $\rho_1$ is pure}
    \end{align*}
\end{minipage}
\caption{Selected \LangName{} denotational semantics. Only rules changed from~\Cref{fig:core-denotation} shown.}
\label{fig:denotation}
\end{figure}

\paragraph{Verifying Separability Conditions}
Soundly verifying \LangName{}'s purity assertions requires verifying their separability conditions. If desired, a mixed-state quantum simulator can be used for this purpose. The simulator determines whether a density matrix is simply separable by taking its partial trace and executing the rank test~(\Cref{sec:background}). Though this approach supports simulating all well-typed \LangName{} programs, it is tied to a computationally inefficient mixed-state simulator.

\subsection{Operational Semantics}

\begin{figure}
\resizebox{\textwidth}{!}{
\parbox{1.1\textwidth}{
\begin{mathpar}
    \inferrule[S-IfT]{\vphantom{\ctx}}{\stepo{\ket{\qstate}}{\If{\True}{e_1}{e_2}}{\emptyset}{1}{\ket{\qstate}}{\Cast{\mixed}{e_1}}}

    \inferrule[S-IfF]{\vphantom{\ctx}}{\stepo{\ket{\qstate}}{\If{\False}{e_1}{e_2}}{\emptyset}{1}{\ket{\qstate}}{\Cast{\mixed}{e_2}}}

    \inferrule[S-Entangle]{\vphantom{\ctx}}{\stepo{\ket{\qstate}}{\Entangle{\purity}{\Qpurity{q_1}{\purity}}{\Qpurity{q_2}{\purity}}}{\emptyset}{1}{\ket{\qstate}}{\Qpurity{\Epair{q_1}{q_2}}{\purity}}}

    \inferrule[S-SplitMixed]{\vphantom{\ctx}}{\stepo{\ket{\qstate}}{\SplitA{\mixed}{\Qpurity{\Epair{q_1}{q_2}}{\mixed}}}{\emptyset}{1}{\ket{\qstate}}{(\Qpurity{q_1}{\mixed}, \Qpurity{q_2}{\mixed})}}

    \inferrule[S-SplitPure]{\ket{\qstate} = \ket{\qstate_1} \otimes \ket{\qstate_2} \otimes \ket{\qstate_0} \\\\ \domain{\ket{\qstate_1}} = \Refs{q_1} \\ \domain{\ket{\qstate_2}} = \Refs{q_2}}{\stepo{\ket{\qstate}}{\SplitA{\pure}{\Qpurity{\Epair{q_1}{q_2}}{\pure}}}{\emptyset}{1}{\ket{\qstate}}{(\Qpurity{q_1}{\pure}, \Qpurity{q_2}{\pure})}}

    \inferrule[S-Cast]{(unchecked)}{\stepo{\ket{\qstate}}{\Cast{\purity}{\Qpurity{q}{\purity'}}}{\emptyset}{1}{\ket{\qstate}}{\Qpurity{q}{\purity}}}
\end{mathpar}
}
}
\caption{Selected \LangName{} operational semantics. The full rules are given in~\Cref{sec:app-semantics}. Rule (\rulename{S-Cast}) is unsound because purity cannot be verified using $\ket{\qstate}$ on one execution alone. Combining this semantics with the static analysis in~\Cref{sec:analysis} enables building a sound interpreter.} \label{fig:dynamics}
\end{figure}

The operational semantics of \LangName{} manipulates a pure quantum state over an individual program execution. The semantics verifies separability conditions concretely using a \emph{separability test} primitive that determines whether the runtime state is separable or entangled.

\Cref{fig:dynamics} presents the updated step judgment, and in~\Cref{sec:app-semantics}, we modify the value judgment to state that $\val{\Qpurity{q}{\purity}}$. For \texttt{if}-expressions, because a different evaluation of the condition could have resulted in taking the other branch, the result may depend on the measurement outcome of some unspecified qubit. Thus, the rule for \texttt{if} casts its output to mixed.
The $\texttt{entangle}_\purity$ operator creates an entangled pair, and \SplitM{} destructs an entangled pair into an ordinary pair containing its components annotated as mixed. We next define semantics for \SplitP{} and \CastP{}.

\paragraph{Separability Tests} Given the pure state $\ket{\qstate}$ and a partition of its domain, a \emph{separability test} determines whether $\ket{\qstate}$ is separable into pure sub-states whose domains are the sets in the partition.

\paragraph{Verifying \texttt{split}} The \texttt{split} operator soundly verifies purity using a separability test on the runtime quantum state $\ket{\qstate}$. Sound verification is possible because the typing rule for $\SplitA{\pure}{e}$ guarantees that $e$ is pure.
The expression $\SplitA{\pure}{e}$ evaluates under state $\ket{\qstate}$ by first evaluating $e$ to a unique value $\Qpurity{\Epair{q_1}{q_2}}{\pure}$ and state $\ket{\qstate'}$. Purity guarantees that no qubit in $\Epair{q_1}{q_2}$ is entangled with the rest of $\ket{\qstate'}$. Thus, $\ket{\qstate'} = \ket{\qstate_{12}} \otimes \ket{\qstate_0}$ where $\domain{\ket{\qstate_{12}}} = \Refs{\Epair{q_1}{q_2}}$, and furthermore $\ket{\qstate_{12}}$ is identical across all executions.
To verify the purity of $q_1$ and $q_2$, the premises of the \SplitP{} step rule test whether $\ket{\qstate_{12}} = \ket{\qstate_1} \otimes \ket{\qstate_2}$ where $\domain{\ket{\qstate_1}} = \Refs{q_1}$ and $\domain{\ket{\qstate_2}} = \Refs{q_2}$.

\paragraph{Verifying \texttt{cast}} By contrast, the state of the current execution alone cannot indicate whether a \CastP{} is sound. For example, suppose that $x$ and $y$ are in a Bell pair, and consider the expression:
\[
(\Measure{x}, \Cast{\pure}{y}) : \tBool \times \tPurity{\qubit}{\pure}
\]
This expression returns a qubit $y$ that is in a mixed state because $y$ was entangled with $x$ before $x$ was measured, and hence the \CastP{} assertion is invalid.
However, on any particular execution of the program, $y$ appears to take on a pure state of either $\ket{0}$ or $\ket{1}$. Upon reaching the $\CastP{}$, it is not possible to determine that $y$ is mixed, because the probabilistic branch has already occurred.

Thus, the operational semantics for \CastP{} cannot precisely verify its separability condition and is necessarily unsound or incomplete. \Cref{fig:dynamics} presents an operational semantics in which \CastP{} is unsound and does not verify purity.
Instead, we present in~\Cref{sec:analysis} a sound static analysis that guarantees that all uses of \CastP{} occur on pure expressions.
Combining this analysis with the operational semantics allows \LangName{} to verify a large class of programs featuring purity assertions using only the runtime pure state of the program.

\paragraph{Executing Separability Tests}
The operational semantics uses a separability test primitive that operates on the runtime pure quantum state. Though this work does not focus on the hardware implementation of this primitive, we overview in~\Cref{sec:app-hardware} a proposed implementation based on work by~\citet{harrow-separability}.
In a pure-state simulator, \LangName{} verifies \SplitP{} using the Schmidt decomposition (\Cref{sec:background}).
\LangName{} computes the Schmidt coefficients of a state vector by interpreting it as a matrix and taking its singular value decomposition using well-studied algorithms~\citep{svd}. Then, \LangName{} checks that there is only one nonzero coefficient.

\subsection{Type Safety}

We now state the type safety properties for \LangName{}.
We introduce a judgment to state that an attempt to step a \SplitP{} aborts the program if its separability condition fails, denoted $\fail{\ket{\qstate}}{e}{\SplitOp}$ and fully defined in~\Cref{sec:app-semantics}. The progress theorem states that well-typed closed expressions are values, can step, or will fail at runtime due to failing a \SplitP{} separability condition:\noclub

\begin{theorem}[Progress] \label{thm:progress}
If $\hastype{\cdot}{\rctx}{e}{\type}$, then either $\val{e}$ or for all $\ket{\qstate}$ where $\rctx \subseteq \domain{\ket{\qstate}}$, we have either $\canstep{\ket{\qstate}}{e}$ or $\fail{\ket{\qstate}}{e}{\SplitOp}$.
\end{theorem}

The proof is given in~\Cref{sec:progress-proof} and is analogous to~\Cref{thm:core-progress}, extended to the operators in \LangName{}. The preservation theorem is identical to~\Cref{thm:core-preservation} and proved analogously.

\subsection{Semantics Equivalence}

The denotational semantics constrain valid operational executions in the presence of the \CastP{} operator. The following theorem modifies~\Cref{thm:core-sem-equiv} by stating that if the denotation of a program is a valid density matrix, then it agrees with the operational semantics:
\begin{theorem} \label{thm:sem-equiv}
    Let the multiset $S = [ (p_i, \ket{\qstate_i}, v_i, O_i) \mid \stepostar{\ket{\qstate}}{e}{O_i}{p_i}{\ket{\qstate_i}}{v_i} ]$. Then, if
    \begin{align*}
        \Denote{e}{\phi}{\dyad{\qstate}} = (\rho, v)
        \text{ then }& \rho = \textstyle\sum_{(p_i, \ket{\qstate_i}, v_i, O_i) \in S} p_i \dyad{\qstate'_i}
        \text{and } (\ket{\qstate'_i}, v) \vequiv (\ket{\qstate''_i}, v') \\
        \text{ where }& \ket{\qstate_i} \otimes \textsf{defer}(O_i) \otimes \ket{0}^* = \ket{\qstate''_i}
        \text{and } v_i = \textsf{apply}(O_i, v').
    \end{align*}
\end{theorem}

The proof of this theorem adds cases for $\texttt{entangle}_\purity$, \CastM{}, and \SplitM{}, which evaluate to values with the same denotation as their argument, as well as \SplitP{}, whose operational and denotational separability conditions align, and \CastP{} when its separability condition holds.

Similarly, we weaken~\Cref{thm:purity-denotation} to state that the denotation having a pure sub-state is a sufficient condition for purity. \Cref{def:purity} for purity holds vacuously when the program aborts at runtime and there is no execution to a final value. While the denotation of an illegal assertion is always $\bot$, corresponding to the program aborting at runtime, the operational semantics does not verify the separability condition of \CastP{} and may not abort.

\subsection{Purity Soundness}

We now prove that under the operational semantics, the purity type system excluding the \CastP{} operator is sound. In this section, we assume expressions contain no instances of \CastP{}. We first establish a relationship stating that quantum states respect pure annotations in expressions:

\begin{definition}[Compatibility]
    An expression $e$ is \emph{compatible} with quantum state $\ket{\qstate}$, denoted $\iscompat{\ket{\qstate}}{e}$, if for every $\Qpurity{q}{\pure}$ appearing within $e$, we have $\sempure{\Qpurity{q}{\pure}}{\ket{\qstate}}$.
\end{definition}

For example, the expression $\Qpurity{\Qref{\alpha}}{\pure}$ of pure type is only pure if qubit $\alpha$ is separable from the rest of the system in $\ket{\qstate}$.
We maintain the compatibility property through a preservation theorem for \LangName{} that augments~\Cref{thm:core-preservation}:

\begin{theorem}[Preservation] \label{thm:preservation}
    If $\hastype{\ctx}{\rctx}{e}{\type}$ and $\rctx \subseteq \domain{\ket{\qstate}}$ and $\iscompat{\ket{\qstate}}{e}$ and $\step{\ket{\qstate}}{e}{}{\ket{\qstate'}}{e'}$, then $\hastype{\ctx}{\rctx'}{e'}{\type}$ where $\rctx' \subseteq \domain{\ket{\qstate'}}$ and $\iscompat{\ket{\qstate'}}{e'}$.
\end{theorem}

The proof is by induction on the derivation of $\step{\ket{\qstate}}{e}{}{\ket{\qstate'}}{e'}$ and given in~\Cref{sec:preservation-proof}.
The main soundness theorem states that an expression with pure type is pure. Assuming the expression satisfies runtime verification for \SplitP{}, it evaluates to a unique final value and state.

\begin{theorem}[Purity Soundness] \label{thm:purity}
    If $\hastype{\cdot}{\rctx}{e}{\tPurity{\qtype}{\pure}}$, $\rctx \subseteq \domain{\ket{\qstate}}$, and $\iscompat{\ket{\qstate}}{e}$, then $\sempure{e}{\ket{\qstate}}$.
\end{theorem}

The proof is by logical relations. For the relation, we define purity at a type $\type$, lifting purity to function types by stating that they take pure inputs to pure outputs and to pairs by stating that their components are pure.
We give the full proof in~\Cref{sec:purity-proof}, strengthening the theorem to open terms using a substitution of free variables for pure expressions, and then proceeding by induction on the derivation of $\hastype{\ctx}{\rctx}{e}{\type}$ using the strengthened inductive hypothesis.

\section{Static Analysis for Purity} \label{sec:analysis}

In this section, we present a static analysis guaranteeing that all uses of \CastP{} operators are sound, obviating the need to verify them using mixed-state simulation.
This analysis verifies that the qubits owned by an expression are separable from all others in the system, including those that were measured.
The analysis is sound and conservative, relying on the fact that qubit $\alpha$ may only become entangled with qubit $\beta$ by entering the same entangled pair as $\beta$ or a qubit $\gamma$ that is entangled with $\beta$. Thus, the analysis tracks possibly-entangled qubits by a variant of data-dependence analysis.

\subsection{Tracking Split Entangled Pairs}

If a pure entangled pair is split into two components $e_1$ and $e_2$, any expression containing only $e_1$ or $e_2$ is potentially mixed. For an expression to be pure it must contain, for every $\SplitM{}$ in the program, either zero or both of its components.
Based on this observation, the analysis tracks for each expression the fraction of each \SplitM{} in the program it contains.

Our approach is similar to fractional permissions~\citep{boyland} in that we associate each type with a fractional quantity that is divided upon destructing a type into constituents.
Specifically, we associate each expression with one fraction per entangled pair created by the program, representing the components of the pair.
We modify the type system from~\Cref{sec:language} to generalize purities $\pure$ and $\mixed$ to a data structure that we call a \emph{history}.

\begin{definition}
    A \emph{history} is a linear combination $\sum_i c_i x_i$ where each $c_i \in \mathbb{Q}$, $0 \le c_i < 1$ and each $x_i$ is a symbol distinguishable from $x_j$ when $i \neq j$.
\end{definition}

The analysis assigns each instance of \SplitM{} in the program a unique index $i$ and associates it with the symbol $x_i$. An expression's history specifies the fraction of each \SplitM{} it contains.

\paragraph{Split and Combine}
We define two operations to manipulate histories. Letting histories $f = \sum_{i=0}^N a_i x_i$ and $g = \sum_{i=0}^N b_i x_i$, and a fresh index $j > N$, define:
\begin{align*}
    \textstyle
    \SplitPoly{f}{j} = \frac{1}{2}\left(f + x_j\right) \text{ and }\,
    \CombinePoly{f}{g} = \sum_{i=0}^N \mathrm{frac}\left(a_i + b_i\right)x_i
\end{align*}
The $\SplitPoly{f}{j}$ operator adds a new term $x_j$ to $f$ and halves every coefficient, representing the components that each hold half of the parent. The $\CombinePoly{f}{g}$ operator adds two histories term-wise, taking the fractional part of each coefficient so that fractions adding to one cancel. The definition performs this cancellation because if an expression contains either zero or both components of every \SplitM{} in the program, it must be pure.

\paragraph{Pure Expressions}
The history containing zero terms contains no fractional component of any $\SplitM{}$ in the program, meaning it represents a pure expression. We denote this history $\pure$ for sake of continuity.
As an example, suppose $\Epair{q_1}{q_2}$ has history $f = \frac{1}{2}x_1 + \frac{3}{4}x_2$. Applying $\SplitM{}$ produces two expressions with history $\SplitPoly{f}{3} = \frac{1}{4}x_1 + \frac{3}{8}x_2 + \frac{1}{2}x_3$. Neither expression can become pure unless combined with the other to cancel the $x_3$ term. Now suppose that value $q_3$ has $g = \frac{1}{2}x_1 + \frac{1}{4}x_2$. Then $\Epair{\Epair{q_1}{q_2}}{q_3}$ has history $\CombinePoly{f}{g} = \pure$ and is pure. Thus, $\Epair{\Epair{q_1}{q_2}}{q_3}$ has no outside entanglements because it cannot contain a fraction of any $\SplitM{}$ in the program.

\paragraph{Analysis Type System}
Formally, the analysis accepts a program by assigning it a type under a modified type system with purities replaced by histories.
\Cref{fig:analysis-types} presents the typing judgment $\hastypea{\ctx}{e}{\type}$ used by the analysis.
Nearly all of its defining rules are derived directly from the original type system, and only the shown typing rules are modified substantially.

\begin{figure}
\resizebox{\textwidth}{!}{
\parbox{1.1\textwidth}{
\begin{mathpar}
    \inferrule[A-If]{\hastypea{\ctx_1}{e}{\tBool} \\ \hastypea{\ctx_2}{e_1}{\tPurity{\qtype}{f}} \\ \hastypea{\ctx_2}{e_2}{\tPurity{\qtype}{g}}}{\hastypea{\ctx_1, \ctx_2}{\If{e}{e_1}{e_2}}{\tPurity{\qtype}{\mixed}}}

    \inferrule[A-Entangle]{\hastypea{\ctx}{e}{\tPair{\tPurity{\qtype_1}{f}}{\tPurity{\qtype_2}{g}}} \quad h = \CombinePoly{f}{g}}{\hastypea{\ctx}{\EntangleOne{\purity}{e}}{\tPurity{(\tEpair{\qtype_1}{\qtype_2})}{h}}}

    \inferrule[A-SplitMixed]{\hastypea{\ctx}{e}{\tPurity{(\tEpair{\qtype_1}{\qtype_2})}{f}} \quad j\ \mathsf{fresh} \quad g = \SplitPoly{f}{j}}{\hastypea{\ctx}{\SplitA{\mixed}{e}}{\tPair{\tPurity{\qtype_1}{g}}{\tPurity{\qtype_2}{g}}}}

    \inferrule[A-CastMixed]{\hastypea{\ctx}{e}{\tPurity{\qtype}{f}}}{\hastypea{\ctx}{\Cast{\mixed}{e}}{\tPurity{\qtype}{f}}}

    \inferrule[A-CastPure]{\hastypea{\ctx}{e}{\tPurity{\qtype}{\pure}}}{\hastypea{\ctx}{\Cast{\pure}{e}}{\tPurity{\qtype}{\pure}}}
\end{mathpar}
}
}
\caption{Selected rules of analysis type system. The full rules are given in~\Cref{sec:app-analysis}.} \label{fig:analysis-types}
\end{figure}

An \texttt{if}-expression obtains a special history $\mixed$ where $\SplitPoly{\mixed}{j} = \mixed$ and $\CombinePoly{f}{\mixed} = \mixed$ for any history $f$. The reason is that the static analysis cannot know which branch the \texttt{if} takes.
The rule for $\mathtt{entangle}_{\purity}$ invokes $\CombinePoly{f}{g}$ on its argument, disregarding the annotation $\purity$. The rule for \SplitM{} invokes $\SplitPoly{f}{j}$ on its arguments to compute the type of the result. The \CastM{} operator has the same history as its argument, and the \CastP{} operator is only valid when the analysis knows its argument to be pure.

\subsection{Purity Soundness}

The following theorem states that a well-typed program that passes the analysis will satisfy its purity specification at runtime. An expression of pure type is pure, and assuming it satisfies runtime verification for \SplitP{}, it evaluates to a unique final value and state.

\begin{theorem}[Purity Soundness] \label{thm:analysis-soundness}
    If $\hastypea{\cdot}{e}{\tPurity{\qtype}{\pure}}$, then $\sempure{e}{\ket{\cdot}}$.
\end{theorem}

We give the full proof in~\Cref{sec:analysis-soundness-proof}, proceeding by induction on $\hastypea{\cdot}{e}{\tPurity{\qtype}{\pure}}$ to show that an expression with pure history does not own any qubit that is entangled with any unowned qubit. A program passing the static analysis may only invoke \CastP{} on expressions of history $\pure$, which are never the results of \texttt{if}-expressions or entangled with unowned qubits, and hence are pure.

\section{Evaluation} \label{sec:evaluation}

We now implement the type checker, static analysis, and interpreter for \LangName{},\footnotemark{} and use them to analyze and execute a set of benchmark programs and answer the following research questions:
\footnotetext{The implementation is available at \url{https://www.github.com/psg-mit/twist-popl22}.}

\begin{enumerate}[label=\emph{RQ\arabic*.}]
\item
Is \LangName{} expressive enough to permit writing standard quantum algorithms?

\item
Does \LangName{} reject programs that contain bugs caused by violating purity specifications?

\item
How does \LangName{}'s runtime performance compare between pure- and mixed-state simulators?

\item
What is the runtime overhead of \LangName{}'s purity assertions in simulation?

\item
Is \LangName{} expressive enough to permit programs that existing languages disallow?
\end{enumerate}

\subsection{Implementation}

We implemented the interpreter in OCaml, using Quantum++~\citep{qpp}, a state-of-the-art C++ quantum simulator. We perform measurements and unitary gates by invoking Quantum++ functions. The implementation adds support for three-qubit Toffoli (CCNOT) and Fredkin (CSWAP) gates and arbitrary (controlled) phase rotation gates. We implemented purity assertions using both pure- and mixed-state simulation (\Cref{sec:language}), which Quantum++ natively supports.

\subsection{Methodology}

For RQ1, we wrote a set of benchmark programs, described in the next section, and annotated each with purity specifications. For RQ2, we modified several programs to introduce a small bug, for example deleting a vital unitary gate or using an incorrect gate, such that the program would yield incorrect results. We then executed the type checker, static analysis, and both pure- and mixed-state runtime verification on the programs. We did not execute later analysis passes on ill-typed programs, nor did we execute the pure-state simulator when the static analysis failed.

For RQ3 and RQ4, we wrote a family of programs invoking runtime verification on an increasing number of qubits, described in the next section. We executed them using both the pure- and mixed-state simulators and measured their execution time, specifically the portion of time spent performing runtime verification. We executed all benchmarks on a MacBook Pro with 2.4GHz 8-core Intel Core i9 processor and 64 GB of RAM. We invoked optimization level -O3 and enabled OpenMP (used by Quantum++). All reported timings are the average of 10 executions.\noclub

For RQ5, we compare the results of \LangName{}'s analyses on the benchmarks with the type system of Silq~\citep{silq}, a recent quantum programming language.

\subsection{Benchmark Programs}

We implemented a range of programs featuring entangled states, including well-known quantum algorithms. These benchmarks include \emph{Teleport-Deferred}, the example from~\Cref{sec:examples} of deferred-measurement teleportation; a classical AND oracle function; a faulty substitution of a Bell state by a Greenberger-Horne-Zeilinger state~\citep{ghz}; Deutsch, Deutsch-Jozsa~\citep{deutsch}, and Grover's~\citep{grover} algorithms; a quantum Fourier transform~\citep{qft}; and Shor's nine-qubit error-correcting code~\citep{shor-code}. For the benchmarks \emph{AndOracle}, \emph{Deutsch}, \emph{DeutschJozsa}, \emph{Grover}, and \emph{ShorCode}, we also implemented erroneous variants that violate their purity specifications and yield incorrect results.

\begin{figure}[t]
\begin{lstlisting}[numbers=left]
fun teleport (q1 : qubit<P>) : qubit<P> =
  let (q2 : qubit<M>, q3 : qubit<M>) = bell_pair () in
  let (q1 : qubit<M>, q2 : qubit<M>) = CNOT (q1, q2) in
  let q1 : qubit<M> = H (q1) in
  let q3 = if measure (q2) then X (q3) else q3 in
  let q3 = if measure (q1) then Z (q3) else q3 in
  cast<P>(q3)
\end{lstlisting}
\caption{\emph{Teleport-Measure} benchmark, a variant of \emph{Teleport-Deferred} that uses classical \texttt{if}-expressions.} \label{fig:teleport-measure}
\end{figure}

The \emph{Teleport-Measure} benchmark (\Cref{fig:teleport-measure}) is a variant of \emph{Teleport-Deferred} that does not use quantum conditional gates. Instead, the program measures \texttt{q1} and \texttt{q2} and uses the classical outcomes to conditionally execute gates on \texttt{q3} via \texttt{if}-expressions followed by a purifying cast.

To study the performance of runtime verification on increasingly complex programs, we implemented a benchmark \emph{ModMul($n$)}, inspired by~\citet{huang2019}, which implements modular multiplication for a number of qubits $n$ ranging from 4 to 22, with a different version of the program for each input size. The programs use purity assertions to verify that after a conditional modular multiplication followed by its inverse, the condition qubit is separable from the multiplicand. We also implemented an erroneous variant \emph{ModMul($n$)-NotInverse} where the inverse operation is incorrect, resulting in the condition qubit remaining entangled.

Full descriptions of all benchmark programs may be found in~\Cref{sec:app-examples}, and their full sources are provided in~\Cref{sec:app-full-sources}. The sources utilize syntax extensions to \LangName{} described in~\Cref{sec:higher-level}.

\subsection{RQ1 and RQ2: Analysis Results}

We list in~\Cref{tbl:results} the analysis outputs for each benchmark compared to ground-truth knowledge. Detailed descriptions of the analysis results may be found in~\Cref{sec:app-examples}.

\paragraph{Research Question 1}
We can express quantum algorithms such as \emph{Deutsch}, \emph{DeutschJozsa}, \emph{Grover}, \emph{QFT}, and \emph{ShorCode}, and \LangName{} correctly determines that they satisfy their purity specification.

\paragraph{Research Question 2}
The type checker correctly rejects three of the benchmarks, \emph{AndOracle-NotUncomputed}, \emph{Bell-GHZ}, and \emph{ShorCode-Drop}, which use mixed expressions in contexts that require pure expressions. The static analysis correctly rejects two of the benchmarks, \emph{DeutschJozsa-MixedInit}, and \emph{Grover-BadOracle}, that inappropriately use the purifying-cast operator to coerce a mixed expression into a pure one. Runtime verification correctly rejects three of the remaining benchmarks, \emph{Teleport-NoCZ}, \emph{Deutsch-BadResultBasis}, and \emph{ModMul($n$)-NotInverse}, due to failing the separability condition for the purifying-split operator.

For \emph{Teleport-Measure}, the static analysis is imprecise. The static analysis rejects the final purifying-cast assertion and does not permit the result to be annotated as pure, even though it is pure. As a result, we must use the mixed-state simulator to determine that the program is valid.

\newcommand{\na}{\textcolor{gray}{\footnotesize N/A}}
\newcommand{\micros}{\ensuremath{\mu s}}
\newcommand{\subname}{\ \raisebox{2pt}{\rotatebox[origin=c]{180}{\textcolor{gray}{$\neg$}}}}

\begin{table}[t]
\caption{Evaluation results on benchmark programs. The column ``valid'' denotes the ground truth of whether the program is valid under its purity specification, and the following three columns state whether the program passed the type check, static analysis, and runtime verification respectively. Pure- and mixed-state simulations agreed in all instances where both were run. If type checking failed, the later analyses were not executed.\\\hspace{\textwidth}
${}^\ast$ For tests that failed the static analysis, only the mixed-state simulator was executed.\\\hspace{\textwidth}
${}^\dagger$ For \emph{Teleport-Measure}, the static analysis was sound but imprecise (overly conservative).}
\begin{tabular}{ l c c c c c c c c }
\toprule
     &           & \multicolumn{4}{c}{purity specifications} & \multirow[b]{2}{*}{\raisebox{1.1pt}{\begin{tabular}{@{}c@{}}\LangName{} \\ correct\end{tabular}}} \\
\cmidrule{3-6}
Name & \# qubits & valid & types & static & dynamic & \\
\midrule
Teleport-Deferred & 3 & \checkmark & \checkmark & \checkmark & \checkmark & \checkmark \\
\subname{} Teleport-NoCZ & 3 & \xmark & \checkmark & \checkmark & \xmark & \checkmark \\
\subname{} Teleport-Measure & 3 & \checkmark & \checkmark & \xmark & \hphantom{$^\ast$}$\checkmark^\ast$ & \hphantom{$^\dagger$}${-}^\dagger$ \\
AndOracle & 3 & \checkmark & \checkmark & \checkmark & \checkmark & \checkmark \\
\subname{} AndOracle-NotUncomputed & 3 & \xmark & \xmark & \na & \na & \checkmark \\
Bell-GHZ & 3 & \xmark & \xmark & \na & \na & \checkmark \\
Deutsch & 2 & \checkmark & \checkmark & \checkmark & \checkmark & \checkmark \\
\subname{} Deutsch-BadResultBasis & 2 & \xmark & \checkmark & \checkmark & \xmark & \checkmark \\
DeutschJozsa & 3 & \checkmark & \checkmark & \checkmark & \checkmark & \checkmark \\
\subname{} DeutschJozsa-MixedInit & 3 & \xmark & \checkmark & \xmark & \hphantom{$^\ast$}$\xmark^\ast$ & \checkmark \\
Grover & 4 & \checkmark & \checkmark & \checkmark & \checkmark & \checkmark \\
\subname{} Grover-BadOracle & 4 & \xmark & \checkmark & \xmark & \hphantom{$^\ast$}$\xmark^\ast$ & \checkmark \\
QFT & 3 & \checkmark & \checkmark & \checkmark & \checkmark & \checkmark \\
ShorCode & 9 & \checkmark & \checkmark & \checkmark & \checkmark & \checkmark \\
\subname{} ShorCode-Drop & 9 & \xmark & \xmark & \na & \na & \checkmark \\
ModMul($n$) & 4--22 & \checkmark & \checkmark & \checkmark & \checkmark & \checkmark \\
\subname{} ModMul($n$)-NotInverse & 4--22 & \xmark & \checkmark & \checkmark & \xmark & \checkmark \\
\bottomrule
\end{tabular}
\label{tbl:results}
\end{table}

\subsection{RQ3 and RQ4: Timing Results}

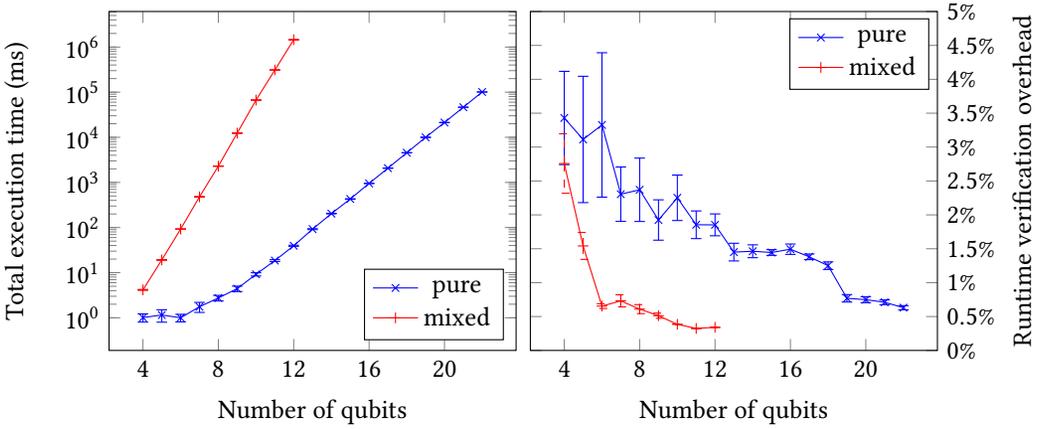
\begin{figure}
\begin{tikzpicture}
\begin{axis}[
    xlabel=Number of qubits,
    ylabel=Total execution time (ms),
    ymode=log,
    ytick={0,1,1e1,1e2,1e3,1e4,1e5,1e6,1e7},
    xtick distance=4,
    width=7cm,
    legend pos=south east
]
\addplot[blue, mark=x, error bars]
plot [error bars/.cd, y dir=both, y explicit]
coordinates {
    (4, 1.011491) +- (0.209637816008944, 0.209637816008944)
    (5, 1.152729) +- (0.353348266050649, 0.353348266050649)
    (6, 1.00112) +- (0.191191720776119, 0.191191720776119)
    (7, 1.756216) +- (0.440267203977565, 0.440267203977565)
    (8, 2.748419) +- (0.390483568706661, 0.390483568706661)
    (9, 4.421783) +- (0.610191657357753, 0.610191657357753)
    (10, 9.193992) +- (0.758592934641792, 0.758592934641792)
    (11, 18.544388) +- (0.880256771093021, 0.880256771093021)
    (12, 39.481378) +- (1.10862563199847, 1.10862563199847)
    (13, 92.261601) +- (2.06925187766897, 2.06925187766897)
    (14, 203.816509) +- (2.47594456722433, 2.47594456722433)
    (15, 428.893066) +- (5.03222833607985, 5.03222833607985)
    (16, 950.522472) +- (4.62294843922646, 4.62294843922646)
    (17, 2070.960187) +- (11.4085263466331, 11.4085263466331)
    (18, 4550.88296) +- (7.57481149467718, 7.57481149467718)
    (19, 9976.399374) +- (98.4764911795562, 98.4764911795562)
    (20, 21251.917957) +- (28.4233275722753, 28.4233275722753)
    (21, 46356.893419) +- (102.054430719963, 102.054430719963)
    (22, 100984.694291) +- (282.31773262886, 282.31773262886)
};
\addlegendentry{pure}
\addplot[red, mark=+, error bars]
plot [error bars/.cd, y dir=both, y explicit]
coordinates {
    (4,  4.15935754) +- (0.09011420048822495, 0.09011420048822495)
    (5,  19.16473872) +- (0.29388141973025156, 0.29388141973025156)
    (6,  92.54250282999999) +- (0.6100137435256759, 0.6100137435256759)
    (7,  478.91998269999993) +- (4.315470876432495, 4.315470876432495)
    (8,  2296.5412616999997) +- (6.942379748617561, 6.942379748617561)
    (9,  12375.3449917) +- (7.904342068599695, 7.904342068599695)
    (10, 66822.87957660001) +- (179.96674864412313, 179.96674864412313)
    (11, 310134.82446659997) +- (1267.2110525696905, 1267.2110525696905)
    (12, 1450989.8498535) +- (635.9310644672839, 635.9310644672839)
};
\addlegendentry{mixed}
\end{axis}
\end{tikzpicture}%
\hspace{5pt}%
\begin{tikzpicture}
\begin{axis}[
    ymin=0, ymax=0.05,
    xlabel=Number of qubits,
    ylabel=Runtime verification overhead,
    ylabel near ticks,
    scaled ticks=false,
    ytick distance=0.005,
    yticklabel=\pgfmathparse{100*\tick}\pgfmathprintnumber{\pgfmathresult}\%,
    yticklabel style={/pgf/number format/.cd,fixed,precision=2},
    xtick distance=4,
    width=7cm,
    ylabel near ticks, yticklabel pos=right,
]
\addplot[blue, mark=x, error bars]
plot [error bars/.cd, y dir=both, y explicit]
coordinates {
    (4, 0.0342959057470605) +- (0.00688173704597687, 0.00688173704597687)
    (5, 0.0311270038317766) +- (0.00931965792821487, 0.00931965792821487)
    (6, 0.0332667412498002) +- (0.0106514636292966, 0.0106514636292966)
    (7, 0.0230512647646986) +- (0.00401471829286768, 0.00401471829286768)
    (8, 0.0236990793616257) +- (0.00467480030664182, 0.00467480030664182)
    (9, 0.0192442731812031) +- (0.00297781329920975, 0.00297781329920975)
    (10, 0.0225194888139994) +- (0.00335336517073461, 0.00335336517073461)
    (11, 0.0185392475610411) +- (0.00204087951981662, 0.00204087951981662)
    (12, 0.0185287352432329) +- (0.00161270628047409, 0.00161270628047409)
    (13, 0.0145064358898346) +- (0.00130521223356406, 0.00130521223356406)
    (14, 0.0146297373781434) +- (0.00094812654890137, 0.00094812654890137)
    (15, 0.0144476548846793) +- (0.000421348617686832, 0.000421348617686832)
    (16, 0.0149325517471827) +- (0.000765912674253485, 0.000765912674253485)
    (17, 0.0138109410212433) +- (0.000419490006138882, 0.000419490006138882)
    (18, 0.0125058399656141) +- (0.000560171619049788, 0.000560171619049788)
    (19, 0.00769031662865745) +- (0.00055366401969092, 0.00055366401969092)
    (20, 0.00748865303931683) +- (0.000440855083015007, 0.000440855083015007)
    (21, 0.00710414887001512) +- (0.000347477012989469, 0.000347477012989469)
    (22, 0.0063039484396076) +- (0.000305229005638907, 0.000305229005638907)
};
\addlegendentry{pure}
\addplot[red, mark=+, error bars/error bar style={dashed}]
plot [error bars/.cd, y dir=both, y explicit]
coordinates {
    (4,  0.02757255631358875) +- (0.004382555677290433, 0.004382555677290433)
    (5,  0.015416370362079217) +- (0.001984784028211189, 0.001984784028211189)
    (6,  0.006541107507239008) +- (0.000351210693531613, 0.000351210693531613)
    (7,  0.007312110428671825) +- (0.0008960177176545982, 0.0008960177176545982)
    (8,  0.006094903424308024) +- (0.000664987415128514, 0.000664987415128514)
    (9,  0.005126496169807785) +- (0.000376597857900788, 0.000376597857900788)
    (10, 0.003852014933072968) +- (0.00005407676339508646, 0.00005407676339508646)
    (11, 0.003246923709492828) +- (0.000012241227547381364, 0.000012241227547381364)
    (12, 0.0034092423410125533) +- (0.000021468626800372702, 0.000021468626800372702)
};
\addlegendentry{mixed}
\end{axis}
\end{tikzpicture}%
\caption{Execution performance of \emph{ModMul($n$)} programs in \LangName{}. The first plot depicts total execution time as the number of qubits $n$ increases, and the second depicts the percentage of total execution time spent on runtime verification. Error bars display the standard error of the mean of ten runs.} \label{fig:multiply}
\end{figure}

We display in~\Cref{fig:multiply} the runtime performance of the \emph{ModMul($n$)} family of programs. All static analyses terminated within 50 ms and are not counted as part of runtime.

\paragraph{Research Question 3} Mixed-state simulation rapidly becomes impractical to execute compared to pure-state simulation, taking longer on 11 qubits than the pure-state simulator did on 22 qubits.

\paragraph{Research Question 4} In pure-state simulation, runtime verification is not a large performance burden. On a program with 4 qubits, the relative overhead of runtime verification is the largest, at 3.5\% of total runtime. As the number of gates and qubits increases, the relative overhead of runtime verification diminishes, and is approximately 0.5\% for 22-qubit programs.

Mixed-state simulation is similar, with runtime verification overhead being about 3\% for 4-qubit programs and approximately 0.3\% for 12-qubit programs. Though the relative overhead is ostensibly lower for mixed than pure states, the baseline is much slower for mixed-state simulation.

The reason why the relative overhead decreases with more qubits in this application is that the larger test cases require more gates to encode multiplication over more qubits. However, each program needs only one purity assertion to enforce its purity specification.

\subsection{RQ5: Comparison to Silq}

We compare \LangName{} with Silq~\citep{silq}, a recent language that enables high-level programming through type annotations expressing freedom from effects such as mutation and measurement, letting it automatically uncompute certain temporary qubits that exit scope. We chose to compare against Silq because it is designed to express high-level constructs while preventing unintuitive or physically unrealizable behavior, and because it subsumes features of languages such as~\citet{qwire, quipper, qsharp}. Because our benchmarks extensively utilize Twist's annotations, we evaluated whether Silq's type system accepted a direct translation of the program without any \LangName{}-specific features while still allocating and discarding the same qubits.

Silq rejects a direct translation of the \emph{Teleport-Deferred}, \emph{Deutsch}, and \emph{ShorCode} benchmarks, which, as written in Twist, each discard a temporary expression that has become separable at the end of the computation without measuring its value. Though Silq supports automatically uncomputing certain qubits when they exit scope, it does so only for qubits that are introduced within the same scope by an operation that Silq's type system deems to be \texttt{qfree}, i.e. one that does not use unsupported gates with phase-level effects such as Hadamard. In these benchmarks, Silq cannot determine whether an arbitrary expression, such as the input or output of a function, is separable and safe to discard via uncomputation.

To translate these benchmarks to Silq, the developer may either 1) manually invoke a \texttt{measure} operator on the expression being discarded, 2) manually invoke a \texttt{forget} operator on the expression, or 3) apply the \texttt{qfree} annotation to the operation that produced it. In the first two cases, neither \texttt{measure} nor \texttt{forget} guarantees that the expression is separable from the remaining computation, meaning that invoking either operator may in general leave the computation in a mixed state if the developer does not accurately reason about the purity of the computation. In the third case, the \texttt{qfree} annotation invokes the requirement of the Silq type system that disqualifies the use of phase-level effects within a \texttt{qfree} operation to compute the expression, for example the efficient quantum adder of~\citet{draper2000addition}. As with \texttt{forget}, Silq could in principle be extended with an unsafe cast to \texttt{qfree} that would permit this implementation, with the proviso that the developer must themselves accurately reason about the safety of such casts in general.

By contrast, \LangName{} soundly accepts these programs when annotated with purity specification, because it can determine that measuring a pure expression cannot affect the ongoing computation, and thus permits implicitly measuring it.

\section{Extensions and Limitations} \label{sec:limitations}

\paragraph{Language Features} We overview in~\Cref{sec:higher-level} the higher-level syntactic features that we have already implemented in \LangName{} to enable writing more concise programs. Adding additional features such as arrays, loops, quantum conditional blocks, and automatic adjoints may further improve ease of use. Currently, the user must annotate purities in types, and type inference would make programming faster and more concise, especially if the developer could use it as a tool to immediately see the inferred purity specification of their programs.

\paragraph{Classical Control}
Classical control affects the precision of our static analysis, as seen in the \emph{Teleport-Measure} example. Programs that use classical control currently require the mixed-state approach of purity assertion verification, which is inefficient. Further work may improve the precision of the static analysis on classical control and enable execution on a pure-state simulator.

\paragraph{Hardware Execution}
In this work, we operate in an idealized model of quantum computation and describe how the purity assertions of \LangName{} can be implemented in simulation using the primitive of separability testing.
In quantum hardware, determining whether a state is separable is a form of quantum state tomography~\citep{tomography}, where ascertaining properties of an unknown quantum state generally requires many copies of the state to obtain high accuracy. In~\Cref{sec:app-hardware}, we describe a procedure to determine with high probability whether a pure quantum state is separable, based on~\citet{harrow-separability}. Additionally, existing runtime quantum assertion frameworks such as \citet{huang2019,assertions,approximate-assertions} support some form of separability testing which \LangName{} could leverage. These methods indicate that separability testing is potentially achievable natively in quantum hardware.

\paragraph{Simulation Accuracy}
\LangName{}'s runtime verification as implemented relies on numeric floating-point arithmetic in the quantum simulator, which may introduce the possibility of imprecision error. Our attempts to exploit imprecision to cause separability tests to pass on an entangled state were unsuccessful, but it is possible that very slightly entangled states or excessive error accumulation in the simulator could lead to unsoundness, which would require further effort to mitigate.

\section{Related Work} \label{sec:related-work}

\paragraph{Entanglement Reasoning}

\citet{perdrix2007,perdrix-ai,prost2009,Honda2015AnalysisOQ} present logical systems for reasoning about entanglement in quantum programs based on type systems and abstract interpretation. These frameworks track fine-grained entanglement between specific qubits, making them limited in scale and unable to handle more complex programs such as teleportation.

\citet{rand2021static,gottesman-beyond} propose a type system establishing circumstances under which gate outputs are separable, based on their Heisenberg representation. Unlike ours, it does not guarantee purity, and can only determine separability in specific bases. However, this direction may enable better static checking for separability conditions, increasing the utility of purity specifications.

Researchers have developed compilers that perform reasoning about entanglement in a quantum circuit. For example, ScaffCC~\citep{JavadiAbhari_2015} provides a disentanglement check warning the user when possibly entangled qubits exit scope. Unlike \LangName{}, this check is purely syntactic and cannot be refined by semantic knowledge about the program. \citet{haner2020} leverage entanglement annotations to optimize circuit gate count, but do not ensure that the annotations are sound. By contrast, we proved that pure-typed expressions are in fact pure.

Frameworks for quantum runtime assertions support reasoning about entangled and separable states, such as~\citet{huang2019} who use statistical hypothesis testing on measurement outcomes of repeated program executions, \citet{assertions} who check predicates using projection operators, and~\citet{approximate-assertions} who propose quantum circuits implementing approximate assertions. Though these tools do not support sound reasoning for whole programs, their techniques may be useful to implement runtime verification in \LangName{}.

\paragraph{Quantum Program Verification}

Separation logic~\citep{reynolds} is a well-studied formalism for reasoning about memory aliasing in classical programs, and researchers have extended it into the probabilistic~\citep{barthe2020} and quantum~\citep{zhou2021quantum} realms, generalizing the notion of separation to probabilistic independence and quantum separability respectively. They have also adapted classical techniques such as model checking~\citep{gay2008qmc}, abstract interpretation~\citep{yu}, and Hoare logic~\citep{singhal2020quantum,unruh2019quantum}, and developed relational proof strategies~\citep{relational} for quantum programs. Verification tools have enabled applications such as analysis of robustness against error~\citep{robustness, glepnir} and mechanized soundness proofs for optimizations~\citep{verified-optimizer, shi2020certiq}.

\LangName{}'s pure annotation can be construed as separable conjunction in separation logic, a connection worthy of additional study. However, an advantage of our type system is that it is directly usable during programming -- purity is a first-class construct, and the type checker and analyses automate reasoning about purity. The programmer may simply write a program with annotations and execute it to be confident that it satisfies the specification. Our automated reasoning does not require the programmer to understand a proof framework external to the language.

\LangName{}'s static analysis relies on manipulations of fractions in a similar way as fractional permissions~\citep{boyland}, originally introduced to reason about mutable effects and which has been extended to separation logic~\citep{bornat}, symbolic rather than concrete quantities~\citep{heule}, and applications in memory management~\citep{suenaga}. To our knowledge, we are the first to use fractional permissions-style reasoning for quantum entanglement.

\paragraph{Ancilla Correctness}

One application of purity guarantees in a quantum program is the correctness of temporary qubits, known as ancillas, and verifying uncomputation of ancillas. Silq~\citep{silq} automatically inserts sound uncomputation for ancillas or rejects programs not supporting uncomputation. \citet{unqomp} provide a general scheme to synthesize uncomputation for circuits, and \citet{reqwire} provide a mechanized system for proving ancilla correctness. In general, \LangName{} benefits from the presence of automatic or provable uncomputation, because it can trust the correctness of the uncomputation and elide the runtime verification.

In turn, systems with automatic uncomputation benefit from soundness guarantees of purity types and convenience of implicitly discarding pure values statically known to not require uncomputation. We demonstrated that \LangName{} can detect incorrect programs that these languages cannot. In addition, \LangName{} does not require distinguishing particular qubits as ancillas, and its runtime verification can perform reasoning that they cannot, for example recognizing gates like Hadamard as self-inverse.

\section{Conclusion}

Quantum computing presents unique challenges to programmers who must reason about phenomena such as entanglement that have no analog in the classical world, and if improperly addressed can result in unintuitive bugs. So far, quantum programming languages have sought to ensure valid semantics under the laws of quantum mechanics, but have not developed comprehensive means of understanding entanglement. The result is that the developer must manually determine whether their computations are affected by measurement outcomes of seemingly unrelated qubits.

In this work, we introduce \LangName{}, the first language with sound reasoning for purity, the property of an expression that states its evaluation is unaffected by measurement outcomes of unowned qubits. We present language constructs to assert purity of expressions and verifications for these assertions. \LangName{} enjoys a soundness guarantee stating that in programs that pass its verifications, every expression of pure type is in fact pure and free from entanglement.

To our knowledge, this work is the first to define the powerful notion of purity, which enables sound reasoning about entanglement in quantum programs. We hope this work paves the way to languages featuring abstractions that align with the complex and unintuitive phenomena inherent in quantum computing, allowing classical programmers to reap its computational benefits.

\begin{acks}
We would like to thank Alex Renda, Eric Atkinson, Jesse Michel, Tian Jin, Alex Lew, Sara Achour, Robert Rand, Yipeng Huang, Axel Feldmann, and Rohan Yadav, who all provided feedback on drafts of this paper.
This work was supported in part by the MIT-IBM Watson AI Lab, the Intel Probabilistic Computing Center, and the Sloan Foundation. Any opinions, findings, and conclusions or recommendations expressed in this material are those of the authors and do not necessarily reflect the views of the funding agencies.
\end{acks}

\bibliography{biblio.bib}

\newpage

\appendix

\section{Entanglement for Pure and Mixed States} \label{sec:app-entanglement}

In this work, we focus on the analysis of entanglement and separability in the sense of pure states rather than mixed states. However, from the perspective of quantum mechanics, the notions of separability for pure states and mixed states are not equivalent, and it is conceivable that one would want to analyze quantum programs using the mixed-state definition of separability instead.

In this section, we briefly illustrate differences between the two formalisms and argue that the pure-state definition, known as \emph{simple separability} when applied to density matrices, is appropriate for reasoning about purity in quantum programs. Thus, the separability tests powering \LangName{}'s purity assertions are precise and avoid the extra complexity of general separability of mixed states.

\subsection{Definitions of Entanglement: Concurring Case}

We present example programs that we analyze using the frameworks of pure- and mixed-state entanglement. First, we show a case where the pure- and mixed-state definitions align.

\begin{figure}
\begin{lstlisting}[numbers=left]
fun phi_plus () : (qubit & qubit)<P> = CNOT (H (qinit ()), qinit ())
fun phi_minus () : (qubit & qubit)<P> = CNOT (H (X (qinit ())), qinit ())
fun flip_second (xy : (qubit & qubit)<P>) : (qubit & qubit)<P> =
  let (x : qubit<M>, y : qubit<M>) = xy in
  let (xy : (qubit & qubit)<P>) = (x, X (y)) in xy
fun psi_plus () : (qubit & qubit)<P> = flip_second (phi_plus ())
fun psi_minus () : (qubit & qubit)<P> = flip_second (phi_minus ())
\end{lstlisting}
\caption{Definitions of the four Bell states. The listing uses \LangName{} syntax extensions described in~\Cref{sec:higher-level}.} \label{fig:bell-pairs}
\end{figure}

\Cref{fig:bell-pairs} presents \LangName{} functions that produce the four maximally entangled \emph{Bell states}~\citep{bell}, the first of which is used throughout examples in this work:
\begin{align*}
    \ket{\Phi^+} = \frac{\ket{00} + \ket{11}}{\sqrt{2}} &\qquad \ket{\Phi^-} = \frac{\ket{00} - \ket{11}}{\sqrt{2}} \\
    \ket{\Psi^+} = \frac{\ket{01} + \ket{10}}{\sqrt{2}} &\qquad \ket{\Psi^-} = \frac{\ket{01} - \ket{10}}{\sqrt{2}}
\end{align*}

In \LangName{}, pure quantum values have the property that even when the overall program is in a mixed state, the qubits owned by a pure value constitute a pure sub-state of the program state. For example, consider the program in~\Cref{fig:pure-substate}, in which we produce a qubit \texttt{x} in state $\ket{1}$ and two qubits \texttt{y} and \texttt{z} in the Bell state $\ket{\Phi^+}$. When we measure \texttt{z}, the whole program enters a mixed state: $\ket{1}_\texttt{x} \otimes \ket{0}_\texttt{y}$ with probability $\nicefrac{1}{2}$ and $\ket{1}_\texttt{x} \otimes \ket{1}_\texttt{y}$ with probability $\nicefrac{1}{2}$. Nevertheless, \texttt{x} is pure, as evidenced by the fact that in either branch, the program state is separable into some sub-state for \texttt{y} and a sub-state for \texttt{x} which is always the same: $\ket{1}_\texttt{x}$.

We can equivalently express the program state, and this property of the pure expression \texttt{x}, using the mixed state formalism. Mathematically, the density matrix corresponding to the program state immediately after the measurement of \texttt{z} is:
\begin{align*}
    \rho &= \frac{\dyad{\psi_0} + \dyad{\psi_1}}{2} \text{ where } \ket{\psi_0} = \ket{1}_\texttt{x} \otimes \ket{0}_\texttt{y}, \ket{\psi_1} = \ket{1}_\texttt{x} \otimes \ket{1}_\texttt{y} \\
    &= \dyad{1}_\texttt{x} \otimes \frac{\dyad{0}_\texttt{y} + \dyad{1}_\texttt{y}}{2}
\end{align*}
The fact that the density matrix $\rho$ is separable into density matrices for \texttt{x} and \texttt{y}, where the factor $\dyad{1}_\texttt{x}$ is a pure density matrix, indicates in this specific example that the expression \texttt{x} is pure. Here, the pure- and mixed-state notions of separability coincide.

\begin{figure}
\begin{lstlisting}[numbers=left]
let x : qubit<P> = X (qinit ()) in
let (y : qubit<M>, z : qubit<M>) = phi_plus () in
let _ = measure (z) in (x, y)
\end{lstlisting}
\caption{A quantum program that produces a mixed state with a sub-state for the qubit \texttt{x} that is pure.} \label{fig:pure-substate}
\end{figure}

\subsection{Definitions of Entanglement: Contrasting Case}

\begin{figure}
\begin{lstlisting}[numbers=left]
fun random_bool () : bool = measure (H (qinit ()))
fun random_bell () : (qubit & qubit)<M> =
  if random_bool () then if random_bool () then phi_plus () else phi_minus ()
  else if random_bool () then psi_plus () else psi_minus ()
\end{lstlisting}
\caption{A quantum program that produces one of the four Bell states at random.} \label{fig:bell-program}
\end{figure}

We now examine a case where the mixed-state definition of entanglement yields a different conclusion from the pure-state definition.
\Cref{fig:bell-program} presents a program that produces, uniformly randomly, one of the four Bell states.
The function \texttt{random\_bell} returns an entangled pair of qubits whose state is one of these four with $\nicefrac{1}{4}$ probability each. The result is not a pure expression, because the final state is a mixed state and each individual execution depends on the measurement outcomes of the intermediate qubits used by \texttt{random\_bool} as a source of randomness.

In the output of \texttt{random\_bell}, both qubits of the pair are also mixed when considered in isolation because they are entangled with each other. If we, for example, extract one qubit (and measure the other), there is no sound manner in \LangName{} to coerce it to a qubit of pure type. The static and dynamic verifications will reject attempts using the purifying-cast and split operators because on any execution, regardless of which of the four Bell states is produced, the two qubits are entangled.

However, remarkably, the mixed state corresponding to the output of \texttt{random\_bell} is separable by the mixed state definition. Mathematically, we can confirm:
\begin{align*}
    \rho &= \frac{\dyad{\Phi^+} + \dyad{\Phi^-} + \dyad{\Psi^+} + \dyad{\Psi^-}}{4} \\
    &= \begin{pmatrix}
        \nicefrac{1}{4} & 0 & 0 & 0 \\
        0 & \nicefrac{1}{4} & 0 & 0 \\
        0 & 0 & \nicefrac{1}{4} & 0 \\
        0 & 0 & 0 & \nicefrac{1}{4}
    \end{pmatrix}
    = \begin{pmatrix}
        \nicefrac{1}{2} & 0 \\
        0 & \nicefrac{1}{2}
    \end{pmatrix} \otimes \begin{pmatrix}
        \nicefrac{1}{2} & 0 \\
        0 & \nicefrac{1}{2}
    \end{pmatrix} \\
    &= \frac{\dyad{0} + \dyad{1}}{2} \otimes \frac{\dyad{0} + \dyad{1}}{2}
\end{align*}

Thus, quantum mechanically, the uniform mixture of the four Bell states is indistinguishable (by any measurement process) from two independent copies of one single qubit in a uniform mixture of $\ket{0}$ and $\ket{1}$. The mixed-state definition of separability thus recovers a fact that the pure-state definition, which always concludes that \texttt{x} and \texttt{y} are not separable, cannot.

\subsection{Sufficiency of Pure-State Entanglement}

As the example demonstrates, there are situations where using the mixed-state formalism provides more fine-grained information about how the program state factors into independent sub-states.
However, for the purposes of \LangName{}, knowing that the program state is separable into two mixed states does not always refine our reasoning about the purity of expressions. After all, we are interested in whether there is a pure sub-state in the program, not a mixed one, and ultimately, \LangName{} is concerned not with separability but with purity.

Effectively, the mixed-state definition of separability is too general for \LangName{}. Mixed states can be separable in ways that do not contribute to purity reasoning:
\begin{itemize}
    \item The mixed state is separable, but the components are not pure states, and thus do not provide information about purity of sub-expressions (\texttt{random\_bell}). In this case, the sub-expressions are individually in mixed states.
    \item The mixed state is separable, but only as a convex combination of tensor products. For example, decomposing $\rho = \frac{\rho_1 \otimes \rho_2 + \rho_3 \otimes \rho_4}{2}$ also does not provide information about purity of sub-expressions. In this case, the sub-expressions may be classically correlated.
\end{itemize}

Only if a density matrix $\rho$ is \emph{simply separable} into $\rho_1 \otimes \rho_2$ where one factor is a pure state does separability guarantee that a sub-expression is pure in the sense of \LangName{}. Simple separability is stronger than separability for mixed states because it guarantees that there exists no classical correlation between the two sub-states.

However, we can recover the information of simple separability through pure-state reasoning. After all, assuming without loss of generality that $\rho_1 = \dyad{\phi}$ and $\rho_2 = \sum_j p_j \dyad{\psi_j}$, we have
\begin{align*}
    \rho &= \dyad{\phi} \otimes \sum_j p_j \dyad{\psi_j} = \sum_j p_j \left(\ket{\phi} \otimes \ket{\psi_j}\right)\left(\bra{\phi} \otimes \bra{\psi_j}\right)
\end{align*}
which is equivalent to the statement that all executions of the program yield states that share a common factor $\ket{\phi}$, and can be checked by pure-state separability tests alone.
\section{Separability Testing in Hardware} \label{sec:app-hardware}

\newcommand{\Xvec}{\mathbf{X}}

In this section, we briefly discuss a potential implementation of separability testing on a hardware quantum computer. We utilize concepts from the density matrix formalism of quantum mechanics, including reduced density matrices and the quantity of purity. We refer the reader to~\citet{nielsen_chuang_2010} for the technical background to this section.

Adapting the following procedure to \LangName{} requires re-executing a program to produce new copies of a state that is being subject to a separability test, as separability testing is a form of quantum state tomography~\citep{tomography} that in general requires multiple copies of a state to characterize it.

The \emph{SWAP test}~\citep{buhrman2001} is a separability-testing scheme that consumes two copies of a $d$-qubit quantum state $\ket{\psi}$, and in $O(\log d)$ time accepts separable states with probability 1, using $O(1)$ extra qubits. Given a state $\ket{\psi}$, we divide it into subsystems $A$ and $B$ of dimensions $d_A$ and $d_B$ with reduced density matrices $\rho_A$ and $\rho_B$.

First, we produce two copies $\ket{\psi} \otimes \ket{\psi}$ of the state under test. We swap the $A$ subsystems of the two states, conditioned on an ancilla qubit, and apply the same procedure to the $B$ subsystems. If a subsystem is a single qubit, this amounts to a single Fredkin gate; in general, it requires $\log d$ Fredkin gates. Finally, we perform a measurement to detect the phase acquired by the ancilla qubits.

The SWAP test accepts with probability $p = \frac{1}{2} + \frac{1}{2}\mathfrak{p}$, where $\mathfrak{p} = \mathrm{tr}(\rho_A^2) = \mathrm{tr}(\rho_B^2)$ is the subsystem purity.
During the swaps, the ancilla qubits acquire a phase depending on the purity of the respective subsystems, which are equal for a bipartite state. The phase can be detected by an interferometric measurement, yielding the expression for $p$. More precisely, define the \emph{distance} $\epsilon$ between $\ket{\psi}$ and the nearest separable state as:
\[
    \epsilon = 1 - \mathrm{max} \left\{ \abs{\braket{\psi|\varphi}}^2 :\ket{\varphi} \text{is a separable state} \right\}
\]
This testing scheme then yields a false positive rate that depends on $\epsilon$, such that the test accepts with probability $1 - \alpha \epsilon \le p \le 1 - \beta \epsilon$, where $\alpha = 2$ and $\beta = 11/512$ \citep{harrow-separability}. The test has been realized experimentally \citep{walborn2006}. However, these bounds given assume perfect operation of the quantum hardware; accounting for the cost of imperfect gates in near-term hardware remains an open problem.

To bound $\epsilon$ above with high probability, one could run $n$ repetitions of the SWAP test, registering observations $\Xvec = (X_1, X_2, \dots, X_n)$. If $\ket{\psi}$ is a separable state, and hardware operations are perfect, these will result in $X_i = 1$ for all $i$. To control the distance parameter $\epsilon$, we can express the test as a parameter estimation problem, and bound the probability that the actual distance is greater than some fixed $\epsilon^*$ conditioned on the observations by Markov's inequality:
\[
\Pr [\epsilon > \epsilon^* \mid \Xvec] \le  \frac{\mathbb{E}[\epsilon \mid \Xvec]}{\epsilon^*}
\]
where the conditional expectation term may be obtained from the known quantity $\Pr[\Xvec \mid \epsilon] = p^n$ and a prior distribution over $\epsilon$ given by a Haar-uniform distribution~\citep{Haar1933DerMI} over states. This quantity could equally well be computed in terms of the purity $\mathfrak{p}$, rather than $\epsilon$, and closed-form expressions for the corresponding distribution are known \citep{giraud2007}. In this way, one can achieve any desired level of confidence in the separability of $\ket{\psi}$, by setting the parameter $\epsilon^*$ and adjusting the number of iterations $n$ to upper-bound the probability of a false negative.
\section{Additional Syntactic Features of \LangName{}} \label{sec:higher-level}

In this section, we discuss syntactic features of \LangName{} that enable writing more concise programs.

\subsection{Affine Pure Types}

As discussed, \LangName{} allows pure expressions to be discarded and implicitly measured.
Formally, we perform a syntactic translation to insert measurements of unused variables of type $\tPurity{\qtype}{\pure}$ and discard the result of the measurement. Specifically, we translate functions and \texttt{let}-bindings:
\begin{align*}
\Lam{x}{e} &\leadsto \Lam{x}{\LetOne{\_}{\Measure{x}}{e}} \\
\Let{x}{y}{e'}{e} &\leadsto \Let{x}{y}{e'}{\LetOne{\_}{\Measure{x}}{e}}
\end{align*}
when $x$ has type $\tPurity{\qtype}{\pure}$ and does not appear in $e$, and likewise for $y$ if necessary. It is irrelevant when the implicit measurement takes place, and we do so before evaluating $e$. To discard a pair containing only pure and classical types, we recursively measure and discard its contents.

\subsection{Inferring Conversion Operators}

Converting between ordinary and entangled product types, as well as between different purities, requires appropriate \texttt{entangle}, \texttt{split}, and \texttt{cast} operators.
Instead of forcing the user to write these operators, \LangName{} exposes a generalized \texttt{let}-expression using type annotations from which the language can automatically insert appropriate operators.
The generalized \texttt{let}-expression binds an expression to a \emph{pattern} $p$ which is either a type-annotated variable or a pair of patterns:
\begin{align*}
\textsf{Pattern}\ p &\Coloneqq x : \type \mid (p_1, p_2)
\end{align*}
This expression is a derived form that recursively translates into core constructs:
\begin{align*}
  \Let{p_1}{p_2}{e}{e'} &\leadsto \Let{x}{y}{e}{\LetOne{p_1}{x}{\LetOne{p_2}{y}{e'}}}
\end{align*}
The procedure to infer conversion operators is type-directed. For every \texttt{let}-expression binding the expression $e$ to a pattern $p$, it synthesizes the type $\tau$ of $e$ and the type $\tau'$ that $p$ expects to bind. It then follows rules to transform $e$ by inserting conversions so that its type becomes $\tau'$:
\begin{itemize}
    \item If $\tau = \tau'$, do nothing to $e$.
    \item If $\tau = \tPurity{\qtype}{\pure}$ and $\tau' = \tPurity{\qtype}{\mixed}$, replace with $\Cast{\mixed}{e}$.
    \item If $\tau = \tPurity{\qtype}{\mixed}$ and $\tau' = \tPurity{\qtype}{\pure}$, replace with $\Cast{\pure}{e}$.
    \item If $\tau = \tPurity{(\tEpair{\qtype_1}{\qtype_2})}{\purity}$ and $\tau' = \tPair{\tPurity{\qtype_1}{\purity}}{\tPurity{\qtype_2}{\purity}}$, replace with $\SplitA{\purity}{e}$.
    \item If $\tau = \tPair{\tPurity{\qtype_1}{\purity}}{\tPurity{\qtype_2}{\purity}}$ and $\tau' = \tPurity{(\tEpair{\qtype_1}{\qtype_2})}{\purity}$, replace with $\Let{x}{y}{e}{\Entangle{\purity}{x}{y}}$.
    \item In other cases, recursively descend into $\tau$ and $\tau'$ and apply the above rules.
\end{itemize}

As an example, a program that takes an entangled pair of qubits as input and performs a phase flip conditioned on the first qubit requires a number of conversions:
\begin{lstlisting}
fun f (qs : (qubit & qubit)<P>) : (qubit & qubit)<P> =
  let (q1 : qubit<M>, q2 : qubit<M>) = split<M>(cast<M>(qs)) in
  let q1 : qubit<M> = Z (q1) in
  cast<P>(entangle<M>(q1, q2))
\end{lstlisting}
Conversion operator inference allows us to express the program much more concisely as:
\begin{lstlisting}
fun f (qs : (qubit & qubit)<P>) : (qubit & qubit)<P> =
    let (q1 : qubit<M>, q2 : qubit<M>) = qs in
    let q1 = Z (q1) in
    let out : (qubit & qubit)<P> = (q1, q2) in out
\end{lstlisting}

\subsection{Polymorphism over Purity}

The fact that functions must provide specific purities in their types can result in code duplication with pure and mixed variants. Thus, the language supports generic purity parameters where functions instantiate at a given purity at call site. Polymorphism allows more concise programs that more accurately describe the effect of functions on purity, and also allows the static analysis of~\Cref{sec:analysis} to be more precise.
We extend the syntax of purities to allow \emph{variables}:
\begin{align*}
  \purity \Coloneqq \pure \mid \mixed \mid \pvar{\alpha}
\end{align*}
We require that every purity variable be introduced exactly once in the argument of the function in which it is used.
We do not permit $\SplitA{\purity}{e}$ for generic $\purity$ because its operation fundamentally differs for $\pure$ or $\mixed$,\footnote{One may cast to $\mixed$ to invoke $\texttt{split}_\mixed$, then later cast back.} but we permit generics in \texttt{entangle} and \texttt{cast}, which are statically checked.

To check a function with generic purity, the static analysis instantiates for each purity variable a unique index $i$ and associated history $x_i$, guaranteeing that it cannot be conflated with any other in the system. We also augment the type system to instantiate a function of polymorphic type when it is applied to an argument of specific type.

Generic purities increase the precision of the static analysis. Consider the following program:
\begin{lstlisting}
fun f (q : qubit<M>) : qubit<M> = q
fun g (q : qubit<P>) : qubit<P> = cast<P>(f (cast<M>(q)))
\end{lstlisting}
This program does not pass the static analysis because the function \texttt{f} returns a qubit annotated as mixed rather than pure. However, \texttt{f} is over-specified to only take mixed qubits to mixed qubits, and simply inlining its definition into \texttt{g} results in a well-typed program. A more sophisticated interprocedural static analysis might realize this fact, but a superior solution is to allow \texttt{f} to be polymorphic in the purity of \texttt{q}.
We can annotate the argument to \texttt{f} with the generic purity \texttt{'p}, which \texttt{g} then instantiates with $\pure$:
\begin{lstlisting}
fun f (q : qubit<'p>) : qubit<'p> = q
fun g (q : qubit<P>) : qubit<P> = f (q)
\end{lstlisting}
The example now passes the static analysis and more clearly expresses the effect of \texttt{f} on purity.
\section{Full Language Semantics} \label{sec:app-semantics}

\subsection{\CoreLangHeading{} Language}

\Cref{fig:prog-ok} defines the typing judgment for programs in \CoreLang{} and \LangName{}.
\Cref{fig:core-value-judgment}, \Cref{fig:core-lambda-dynamics}, and \Cref{fig:core-quantum-dynamics} contain the full operational semantics of \CoreLang{}.

\begin{figure}
    \footnotesize
    \begin{mathpar}
        \inferrule{\hastype{\ctx, x : \type}{\emptyset}{e}{\type'} \\ \ProgOk{\ctx, f : \type \to \type'}{m}}{\ProgOk{\ctx}{\Fun{f}{x}{e}; m}}

        \inferrule{\hastype{\ctx}{\emptyset}{e}{\type}}{\ProgOk{\ctx}{\Main{e}}}
    \end{mathpar}
    \caption{Definition of well-formed programs in \CoreLang{} and \LangName{}.} \label{fig:prog-ok}
\end{figure}

\begin{figure}
    \footnotesize
    \begin{mathpar}
        \inferrule[V-Fun]{\vphantom{\ctx}}{\val{f}}

        \inferrule[V-Bool]{\vphantom{\ctx}}{\val{b}}

        \inferrule[V-Qref]{\vphantom{\ctx}}{\val{\Qref{\alpha}}}

        \inferrule[V-Pair]{\val{e_1} \\ \val{e_2}}{\val{(e_1, e_2)}} \\
    \end{mathpar}
    \caption{Definition of value judgment in \CoreLang{}.} \label{fig:core-value-judgment}
\end{figure}

\begin{figure}
    \footnotesize
    \begin{mathpar}
        \inferrule[S-App]{\phi(f) = \Lam{x}{e} \\ \val{e'}}{\stepo{\ket{\qstate}}{\App{f}{e'}}{\emptyset}{1}{\ket{\qstate}}{[e'/x]e}}

        \inferrule[S-AppL]{\stepo{\ket{\qstate}}{e_1}{O}{p}{\ket{\qstate'}}{e_1'}}{\stepo{\ket{\qstate}}{\App{e_1}{e_2}}{O}{p}{\ket{\qstate'}}{\App{e_1'}{e_2}}}

        \inferrule[S-AppR]{\val{e_1} \\ \stepo{\ket{\qstate}}{e_2}{O}{p}{\ket{\qstate'}}{e_2'}}{\stepo{\ket{\qstate}}{\App{e_1}{e_2}}{O}{p}{\ket{\qstate'}}{\App{e_1}{e_2'}}}

        \inferrule[S-PairL]{\stepo{\ket{\qstate}}{e_1}{O}{p}{\ket{\qstate'}}{e_1'}}{\stepo{\ket{\qstate}}{(e_1, e_2)}{O}{p}{\ket{\qstate'}}{(e_1', e_2)}}

        \inferrule[S-PairR]{\val{e_1} \\ \stepo{\ket{\qstate}}{e_2}{O}{p}{\ket{\qstate'}}{e_2'}}{\stepo{\ket{\qstate}}{(e_1, e_2)}{O}{p}{\ket{\qstate'}}{(e_1, e_2')}}

        \inferrule[S-Let]{\val{e_1} \\ \val{e_2}}{\stepo{\ket{\qstate}}{\Let{x}{y}{(e_1, e_2)}{e'}}{\emptyset}{1}{\ket{\qstate}}{[e_1, e_2/x, y]e'}}

        \inferrule[S-LetS]{\stepo{\ket{\qstate}}{e_1}{O}{p}{\ket{\qstate'}}{e_1'}}{\stepo{\ket{\qstate}}{\Let{x}{y}{e_1}{e_2}}{O}{p}{\ket{\qstate'}}{\Let{x}{y}{e_1'}{e_2}}}

        \inferrule[E-Val]{\val{v}}{\stepostar{\ket{\qstate}}{v}{\emptyset}{1}{\ket{\qstate}}{v}}

        \inferrule[E-Step]{\stepo{\ket{\qstate}}{e}{O_1}{p_1}{\ket{\qstate'}}{e'} \\ \stepostar{\ket{\qstate'}}{e'}{O_2}{p_2}{\ket{\qstate''}}{v} \\ O' = O_1 \cup O_2}{\stepostar{\ket{\qstate}}{e}{O'}{p_1 p_2}{\ket{\qstate''}}{v}}
    \end{mathpar}
    \caption{Operational semantics for classical constructs in \CoreLang{} and \LangName.} \label{fig:core-lambda-dynamics}
\end{figure}

\begin{figure}
    \footnotesize
    \begin{mathpar}
        \inferrule[S-Qinit]{\alpha \text{ fresh in } \ket{\qstate}}{\stepo{\ket{\qstate}}{\Qinit\ ()}{\emptyset}{1}{\ket{\qstate} \otimes \ket{0}_\alpha}{\Qref{\alpha}}}

        \inferrule[S-U1]{\vphantom{\ctx}}{\stepo{\ket{\qstate}}{\Uone{\Qref{\alpha}}}{\emptyset}{1}{U_\alpha\ket{\qstate}}{\Qref{\alpha}}}

        \inferrule[S-U1S]{\stepo{\ket{\qstate}}{e}{O}{p}{\ket{\qstate'}}{e'}}{\stepo{\ket{\qstate}}{\Uone{e}}{O}{p}{\ket{\qstate'}}{\Uone{e'}}}

        \inferrule[S-U2]{\vphantom{\ctx}}{\stepo{\ket{\qstate}}{\Utwo{\Qref{\alpha}, \Qref{\beta}}}{\emptyset}{1}{U_{\alpha, \beta}\ket{\qstate}}{(\Qref{\alpha}, \Qref{\beta})}}

        \inferrule[S-U2S]{\stepo{\ket{\qstate}}{e}{O}{p}{\ket{\qstate'}}{e'}}{\stepo{\ket{\qstate}}{\Utwo{e}}{O}{p}{\ket{\qstate'}}{\Utwo{e'}}}

        \inferrule[S-MeasureT]{M_\alpha\ket{\qstate} = \ket{1}_\alpha \otimes \ket{\qstate'} \text{ w.p. } p \\ O = \{(\alpha, \True)\}}{\stepo{\ket{\qstate}}{\Measure{\Qref{\alpha}}}{O}{p}{\ket{\qstate'}}{\True}}

        \inferrule[S-MeasureF]{M_\alpha\ket{\qstate} = \ket{0}_\alpha \otimes \ket{\qstate'} \text{ w.p. } p \\ O = \{(\alpha, \False)\}}{\stepo{\ket{\qstate}}{\Measure{\Qref{\alpha}}}{O}{p}{\ket{\qstate'}}{\False}}

        \inferrule[S-MeasureS]{\stepo{\ket{\qstate}}{e}{O}{p}{\ket{\qstate'}}{e'}}{\stepo{\ket{\qstate}}{\Measure{e}}{O}{p}{\ket{\qstate'}}{\Measure{e'}}}

        \inferrule[S-IfT]{\vphantom{\ctx}}{\stepo{\ket{\qstate}}{\If{\True}{e_1}{e_2}}{\emptyset}{1}{\ket{\qstate}}{e_1}}

        \inferrule[S-IfF]{\vphantom{\ctx}}{\stepo{\ket{\qstate}}{\If{\False}{e_1}{e_2}}{\emptyset}{1}{\ket{\qstate}}{e_2}}

        \inferrule[S-IfS]{\stepo{\ket{\qstate}}{e}{O}{p}{\ket{\qstate'}}{e'}}{\stepo{\ket{\qstate}}{\If{e}{e_1}{e_2}}{O}{p}{\ket{\qstate'}}{\If{e'}{e_1}{e_2}}}
    \end{mathpar}

    \caption{Full operational semantics of quantum constructs in \CoreLang{}.} \label{fig:core-quantum-dynamics}
\end{figure}

\subsection{\LangName{} Language}

\Cref{fig:full-types} contains the full type system of \LangName{}.
\Cref{fig:value-judgment}, \Cref{fig:core-lambda-dynamics}, and \Cref{fig:full-dynamics} contain the full operational semantics of \LangName{}.

\begin{figure}
    \footnotesize
    \begin{mathpar}
        \inferrule[Q-Ref]{\vphantom{\ctx}}{\qhastype{\{\alpha\}}{\Qref{\alpha}}{\qubit}}

        \inferrule[Q-Pair]{\qhastype{\rctx_1}{q_1}{\qtype_1} \\ \qhastype{\rctx_2}{q_2}{\qtype_2}}{\qhastype{\rctx_1 \sqcup \rctx_2}{\Epair{q_1}{q_2}}{\tEpair{\qtype_1}{\qtype_2}}}

        \inferrule[T-Var]{\vphantom{\ctx}}{\hastype{x : \type}{\emptyset}{x}{\type}}

        \inferrule[T-Fun]{\vphantom{\ctx}}{\hastype{f : \type \to \type'}{\emptyset}{f}{\type \to \type'}}

        \inferrule[T-App]{\hastype{\ctx_1}{\rctx_1}{e_1}{\type_1 \to \type_2} \\ \hastype{\ctx_2}{\rctx_2}{e_2}{\type_1}}{\hastype{\ctx_1, \ctx_2}{\rctx_1 \sqcup \rctx_2}{\App{e_1}{e_2}}{\type_2}}

        \inferrule[T-Pair]{\hastype{\ctx_1}{\rctx_1}{e_1}{\type_1} \\ \hastype{\ctx_2}{\rctx_2}{e_2}{\type_2}}{\hastype{\ctx_1, \ctx_2}{\rctx_1 \sqcup \rctx_2}{(e_1, e_2)}{\tPair{\type_1}{\type_2}}}

        \inferrule[T-Let]{\hastype{\ctx_1}{\rctx_1}{e_1}{\tPair{\type_1}{\type_2}} \\\\ \hastype{\ctx_2, x : \type_1, y : \type_2}{\rctx_2}{e_2}{\type}}{\hastype{\ctx_1, \ctx_2}{\rctx_1 \sqcup \rctx_2}{\Let{x}{y}{e_1}{e_2}}{\type}}

        \inferrule[T-If]{\hastype{\ctx_1}{\rctx_1}{e}{\tBool} \\ \hastype{\ctx_2}{\rctx_2}{e_1}{\tPurity{\qtype}{\purity}} \\ \hastype{\ctx_2}{\rctx_2}{e_2}{\tPurity{\qtype}{\purity}}}{\hastype{\ctx_1, \ctx_2}{\rctx_1 \sqcup \rctx_2}{\If{e}{e_1}{e_2}}{\tPurity{\qtype}{\mixed}}}

        \inferrule[T-Bool]{\vphantom{\ctx}}{\hastype{\cdot}{\emptyset}{b}{\tBool}}

        \inferrule[T-Qinit]{\vphantom{\ctx}}{\hastype{\cdot}{\emptyset}{\Qinit\ ()}{\tPurity{\qubit}{\pure}}}

        \inferrule[T-U1]{\hastype{\ctx}{\rctx}{e}{\tPurity{\qubit}{\purity}}}{\hastype{\ctx}{\rctx}{\Uone{e}}{\tPurity{\qubit}{\purity}}}

        \inferrule[T-U2]{\hastype{\ctx}{\rctx}{e}{\tPurity{(\tEpair{\qubit}{\qubit})}{\purity}}}{\hastype{\ctx}{\rctx}{\Utwo{e}}{\tPurity{(\tEpair{\qubit}{\qubit})}{\purity}}}

        \inferrule[T-Measure]{\hastype{\ctx}{\rctx}{e}{\tPurity{\qubit}{\purity}}}{\hastype{\ctx}{\rctx}{\Measure{e}}{\tBool}}

        \inferrule[T-Qval]{\qhastype{\rctx}{q}{\qtype}}{\hastype{\cdot}{\rctx}{\Qpurity{q}{\purity}}{\tPurity{\qtype}{\purity}}}

        \inferrule[T-Entangle]{\hastype{\ctx}{\rctx}{e}{\tPair{\tPurity{\qtype_1}{\purity}}{\tPurity{\qtype_2}{\purity}}}}{\hastype{\ctx}{\rctx}{\EntangleOne{\purity}{e}}{\tPurity{(\tEpair{\qtype_1}{\qtype_2})}{\purity}}}

        \inferrule[T-Split]{\hastype{\ctx}{\rctx}{e}{\tPurity{(\tEpair{\qtype_1}{\qtype_2})}{\purity}}}{\hastype{\ctx}{\rctx}{\SplitA{\purity}{e}}{\tPair{\tPurity{\qtype_1}{\purity}}{\tPurity{\qtype_2}{\purity}}}}

        \inferrule[T-Cast]{\hastype{\ctx}{\rctx}{e}{\tPurity{\qtype}{\purity'}}}{\hastype{\ctx}{\rctx}{\Cast{\purity}{e}}{\tPurity{\qtype}{\purity}}}
    \end{mathpar}
    \caption{Full type system of \LangName{}.} \label{fig:full-types}
    \end{figure}

\begin{figure}
    \footnotesize
    \begin{mathpar}
        \inferrule[V-Fun]{\vphantom{\ctx}}{\val{f}}

        \inferrule[V-Bool]{\vphantom{\ctx}}{\val{b}}

        \inferrule[V-Qval]{\vphantom{\ctx}}{\val{\Qpurity{q}{\purity}}}

        \inferrule[V-Pair]{\val{e_1} \\ \val{e_2}}{\val{(e_1, e_2)}} \\
    \end{mathpar}
    \caption{Definition of value judgment in \LangName{}.} \label{fig:value-judgment}
\end{figure}

\begin{figure}
    \footnotesize
    \begin{mathpar}
        \inferrule[S-Qinit]{\alpha \text{ fresh in } \ket{\qstate}}{\stepo{\ket{\qstate}}{\Qinit\ ()}{\emptyset}{1}{\ket{\qstate} \otimes \ket{0}_\alpha}{\Qpurity{\Qref{\alpha}}{\purity}}}

        \inferrule[S-U1]{\vphantom{\ctx}}{\stepo{\ket{\qstate}}{\Uone{\Qpurity{\Qref{\alpha}}{\purity}}}{\emptyset}{1}{U_\alpha\ket{\qstate}}{\Qpurity{\Qref{\alpha}}{\purity}}}

        \inferrule[S-U1S]{\stepo{\ket{\qstate}}{e}{O}{p}{\ket{\qstate'}}{e'}}{\stepo{\ket{\qstate}}{\Uone{e}}{O}{p}{\ket{\qstate'}}{\Uone{e'}}}

        \inferrule[S-U2]{\vphantom{\ctx}}{\stepo{\ket{\qstate}}{\Utwo{\Qpurity{\Epair{\Qref{\alpha}}{\Qref{\beta}}}{\purity}}}{\emptyset}{1}{U_{\alpha, \beta}\ket{\qstate}}{\Qpurity{\Epair{\Qref{\alpha}}{\Qref{\beta}}}{\purity}}}

        \inferrule[S-U2S]{\stepo{\ket{\qstate}}{e}{O}{p}{\ket{\qstate'}}{e'}}{\stepo{\ket{\qstate}}{\Utwo{e}}{O}{p}{\ket{\qstate'}}{\Utwo{e'}}}

        \inferrule[S-MeasureT]{M_\alpha\ket{\qstate} = \ket{1}_\alpha \otimes \ket{\qstate'} \text{ w.p. } p \\ O = \{(\alpha, \True)\}}{\stepo{\ket{\qstate}}{\Measure{\Qpurity{\Qref{\alpha}}{\purity}}}{O}{p}{\ket{\qstate'}}{\True}}

        \inferrule[S-MeasureF]{M_\alpha\ket{\qstate} = \ket{0}_\alpha \otimes \ket{\qstate'} \text{ w.p. } p \\ O = \{(\alpha, \False)\}}{\stepo{\ket{\qstate}}{\Measure{\Qpurity{\Qref{\alpha}}{\purity}}}{O}{p}{\ket{\qstate'}}{\False}}

        \inferrule[S-MeasureS]{\stepo{\ket{\qstate}}{e}{O}{p}{\ket{\qstate'}}{e'}}{\stepo{\ket{\qstate}}{\Measure{e}}{O}{p}{\ket{\qstate'}}{\Measure{e'}}}

        \inferrule[S-IfT]{\vphantom{\ctx}}{\stepo{\ket{\qstate}}{\If{\True}{e_1}{e_2}}{\emptyset}{1}{\ket{\qstate}}{\Cast{\mixed}{e_1}}}

        \inferrule[S-IfF]{\vphantom{\ctx}}{\stepo{\ket{\qstate}}{\If{\False}{e_1}{e_2}}{\emptyset}{1}{\ket{\qstate}}{\Cast{\mixed}{e_2}}}

        \inferrule[S-IfS]{\stepo{\ket{\qstate}}{e}{O}{p}{\ket{\qstate'}}{e'}}{\stepo{\ket{\qstate}}{\If{e}{e_1}{e_2}}{O}{p}{\ket{\qstate'}}{\If{e'}{e_1}{e_2}}}

        \inferrule[S-Entangle]{\vphantom{\ctx}}{\stepo{\ket{\qstate}}{\EntangleOne{\purity}{\Qpurity{q_1}{\purity}, \Qpurity{q_2}{\purity}}}{\emptyset}{1}{\ket{\qstate}}{\Qpurity{\Epair{q_1}{q_2}}{\purity}}}

        \inferrule[S-EntangleS]{\stepo{\ket{\qstate}}{e}{O}{p}{\ket{\qstate'}}{e'}}{\stepo{\ket{\qstate}}{\EntangleOne{\purity}{e}}{O}{p}{\ket{\qstate'}}{\EntangleOne{\purity}{e'}}}

        \inferrule[S-SplitS]{\stepo{\ket{\qstate}}{e}{O}{p}{\ket{\qstate'}}{e'}}{\stepo{\ket{\qstate}}{\SplitA{\purity}{e}}{O}{p}{\ket{\qstate'}}{\SplitA{\purity}{e'}}}

        \inferrule[S-SplitMixed]{\vphantom{\ctx}}{\stepo{\ket{\qstate}}{\SplitA{\mixed}{\Qpurity{\Epair{q_1}{q_2}}{\mixed}}}{\emptyset}{1}{\ket{\qstate}}{(\Qpurity{q_1}{\mixed}, \Qpurity{q_2}{\mixed})}}

        \inferrule[S-SplitPure]{\ket{\qstate} = \ket{\qstate_1} \otimes \ket{\qstate_2} \otimes \ket{\qstate_0} \\\\ \domain{\ket{\qstate_1}} = \Refs{q_1} \\ \domain{\ket{\qstate_2}} = \Refs{q_2}}{\stepo{\ket{\qstate}}{\SplitA{\pure}{\Qpurity{\Epair{q_1}{q_2}}{\pure}}}{\emptyset}{1}{\ket{\qstate}}{(\Qpurity{q_1}{\pure}, \Qpurity{q_2}{\pure})}}

        \inferrule[S-Cast]{(unchecked)}{\stepo{\ket{\qstate}}{\Cast{\purity}{\Qpurity{q}{\purity'}}}{\emptyset}{1}{\ket{\qstate}}{\Qpurity{q}{\purity}}}

        \inferrule[S-CastS]{\stepo{\ket{\qstate}}{e}{O}{p}{\ket{\qstate'}}{e'}}{\stepo{\ket{\qstate}}{\Cast{\purity}{e}}{O}{p}{\ket{\qstate'}}{\Cast{\purity}{e'}}}
    \end{mathpar}
\caption{Operational semantics of quantum constructs in \LangName{}.} \label{fig:full-dynamics}
\end{figure}
\section{Full Static Analysis Type System} \label{sec:app-analysis}

\Cref{fig:full-analysis} contains the full type system. There is no rule for $\Qpurity{q}{\purity}$ which is not written by the user.

\begin{figure}
    \footnotesize
    \begin{mathpar}
        \inferrule[A-Var]{\vphantom{\ctx}}{\hastypea{x : \type}{x}{\type}}

        \inferrule[A-Fun]{\hastypea{\ctx, x : \type_1}{e}{\type_2}}{\hastypea{\ctx}{\Lam{x}{e}}{\type_1 \to \type_2}}

        \inferrule[A-App]{\hastypea{\ctx_1}{e_1}{\type_1 \to \type_2} \\ \hastypea{\ctx_2}{e_2}{\type_1}}{\hastypea{\ctx_1, \ctx_2}{\App{e_1}{e_2}}{\type_2}}

        \inferrule[A-Pair]{\hastypea{\ctx_1}{e_1}{\type_1} \\ \hastypea{\ctx_2}{e_2}{\type_2}}{\hastypea{\ctx_1, \ctx_2}{(e_1, e_2)}{\tPair{\type_1}{\type_2}}}

        \inferrule[A-Let]{\hastypea{\ctx_2}{e}{\tPair{\type_1}{\type_2}} \\ \hastypea{\ctx_1, x : \type_1, y : \type_2}{e'}{\type}}{\hastypea{\ctx_1, \ctx_2}{\Let{x}{y}{e}{e'}}{\type}}

        \inferrule[A-If]{\hastypea{\ctx_1}{e}{\tBool} \quad \hastypea{\ctx_2}{e_1}{\tPurity{\qtype}{f}} \quad \hastypea{\ctx_2}{e_2}{\tPurity{\qtype}{g}}}{\hastypea{\ctx_1, \ctx_2}{\If{e}{e_1}{e_2}}{\tPurity{\qtype}{\mixed}}}

        \inferrule[A-Bool]{\vphantom{\ctx}}{\hastypea{\cdot}{b}{\tBool}}

        \inferrule[A-Qinit]{\vphantom{\ctx}}{\hastypea{\cdot}{\Qinit\ ()}{\tPurity{\qubit}{\pure}}}

        \inferrule[A-U1]{\hastypea{\ctx}{e}{\tPurity{\qubit}{f}}}{\hastypea{\ctx}{\Uone{e}}{\tPurity{\qubit}{f}}}

        \inferrule[A-U2]{\hastypea{\ctx}{e}{\tPurity{(\tEpair{\qubit}{\qubit})}{f}}}{\hastypea{\ctx}{\Utwo{e}}{\tPurity{(\tEpair{\qubit}{\qubit})}{f}}}

        \inferrule[A-Measure]{\hastypea{\ctx}{e}{\tPurity{\qubit}{f}}}{\hastypea{\ctx}{\Measure{e}}{\tBool}}

        \inferrule[A-Entangle]{\hastypea{\ctx}{e}{\tPair{\tPurity{\qtype_1}{f}}{\tPurity{\qtype_2}{g}}} \quad h = \CombinePoly{f}{g}}{\hastypea{\ctx}{\EntangleOne{\purity}{e}}{\tPurity{(\tEpair{\qtype_1}{\qtype_2})}{h}}}

        \inferrule[A-SplitMixed]{\hastypea{\ctx}{e}{\tPurity{(\tEpair{\qtype_1}{\qtype_2})}{f}} \\ j\ \mathsf{fresh} \\ g = \SplitPoly{f}{j}}{\hastypea{\ctx}{\SplitA{\mixed}{e}}{\tPair{\tPurity{\qtype_1}{g}}{\tPurity{\qtype_2}{g}}}}

        \inferrule[A-SplitPure]{\hastypea{\ctx}{e}{\tPurity{(\tEpair{\qtype_1}{\qtype_2})}{\pure}}}{\hastypea{\ctx}{\SplitA{\pure}{e}}{\tPair{\tPurity{\qtype_1}{\pure}}{\tPurity{\qtype_2}{\pure}}}}

        \inferrule[A-CastMixed]{\hastypea{\ctx}{e}{\tPurity{\qtype}{f}}}{\hastypea{\ctx}{\Cast{\mixed}{e}}{\tPurity{\qtype}{f}}}

        \inferrule[A-CastPure]{\hastypea{\ctx}{e}{\tPurity{\qtype}{\pure}}}{\hastypea{\ctx}{\Cast{\pure}{e}}{\tPurity{\qtype}{\pure}}}

    \end{mathpar}
    \caption{Full static analysis type system.} \label{fig:full-analysis}
\end{figure}
\section{Proofs of Semantic Properties} \label{sec:app-proofs}

In this section, we provide proofs of progress, preservation, and purity soundness theorems for \LangName{}.

\subsection{\LangName{} Language}

\subsubsection{Progress (\Cref{thm:progress})} \label{sec:progress-proof}

First, we state canonical forms lemmas for values:
\begin{lemma}
    If $\hastype{\ctx}{\rctx}{e}{\type}$ and $\val{e}$, then:
    \begin{itemize}
        \item If $\type$ is $\tPurity{\qtype}{\purity}$, then $e$ is $\Qpurity{q}{\purity}$ where $\qhastype{\rctx}{q}{\qtype}$.
        \item If $\type$ is $\tBool$, then $e$ is one of \True{} or \False{}.
        \item If $\type$ is $\tPair{\type_1}{\type_2}$, then $e$ is $(e_1, e_2)$ where $\val{e_1}$ and $\val{e_2}$.
        \item If $\type$ is $\type_1 \to \type_2$, then $e$ is $f$ and $\phi(f) = \Lam{x}{e}$.
    \end{itemize}
\end{lemma}
\begin{proof}
    By inversion of $\val{e}$.
\end{proof}

\begin{lemma}
    If $\qhastype{\rctx}{q}{\qtype}$, then if $\qtype$ is $\qubit$ then $q$ is $\Qref{\alpha}$ for some $\alpha$ and if $\qtype$ is $\Epair{\qtype_1}{\qtype_2}$ then $q$ is $\Epair{q_1}{q_2}$ for some $q_1, q_2$.
\end{lemma}
\begin{proof}
    By inversion of $\qhastype{\rctx}{q}{\qtype}$.
\end{proof}

Next, we prove the main theorem:
\begin{proof}
    Proceed by induction on the derivation of $\hastype{\cdot}{\rctx}{e}{\type}$. \rulename{T-Var} does not apply.

    In cases \rulename{T-Abs}, \rulename{T-Qval}, \rulename{T-Bool} we have $\val{e}$.

    In cases \rulename{T-App}, \rulename{T-Pair}, by IH either $\canstep{\ket{\qstate}}{e_1}$ or $\val{e_1}$. In the former case, apply \rulename{S-AppL} or \rulename{S-PairL} to obtain $\canstep{\ket{\qstate}}{e}$. In the latter case, by IH either $\canstep{\ket{\qstate}}{e_2}$ or $\val{e_2}$. In the former case, apply \rulename{S-AppR} or \rulename{S-PairR} and in the latter case apply \rulename{S-App} or \rulename{V-Pair}.

    In cases \rulename{T-Let}, \rulename{T-U1}, \rulename{T-U2}, and \rulename{T-Entangle}, by IH either $\canstep{\ket{\qstate}}{e}$ or $\val{e}$. In the former case, apply \rulename{S-LetS}, \rulename{S-U1S}, \rulename{S-U2S}, or \rulename{S-EntangleS}. In the latter case, apply \rulename{S-Let}, \rulename{S-U1}, \rulename{S-U2}, or \rulename{S-Entangle}.

    In case \rulename{T-Qinit}, apply \rulename{S-Qinit}. In case \rulename{T-If}, by IH either $\canstep{\ket{\qstate}}{e}$ or $\val{e}$. In the former case, apply \rulename{S-IfS}. In the latter case, apply \rulename{S-IfT} or \rulename{S-IfF}.

    In case \rulename{T-Measure}, by IH either $\canstep{\ket{\qstate}}{e}$ or $\val{e}$. In the former case, apply \rulename{S-MeasureS}. In the latter case, apply one of \rulename{S-MeasureT} or \rulename{S-MeasureF}, whose probabilities resulting from a two-outcome quantum measurement add to 1.

    In case \rulename{T-Split}, by IH we have $\canstep{\ket{\qstate}}{e}$ or $\val{e}$. In the former case, apply \rulename{S-SplitS}. In the latter case, either apply \rulename{S-SplitMixed}, or if the premises of \rulename{S-SplitPure} are true, apply \rulename{S-SplitPure}, otherwise apply \rulename{S-SplitFail}.

    In case \rulename{T-Cast}, by IH we have $\canstep{\ket{\qstate}}{e}$ or $\val{e}$. In the former case, apply \rulename{S-CastS}. In the latter case, apply \rulename{S-Cast}.
\end{proof}

\subsubsection{Preservation (\Cref{thm:preservation})} \label{sec:preservation-proof}

In the following, we use $\sempure{q}{\ket{\qstate'}}$ as shorthand for $\sempure{\Qpurity{q}{\pure}}{\ket{\qstate'}}$.
We first state a helpful lemma. Its proof follows from the fact that purity of quantum values is equivalent to separability, and separable qubits remain separable under irrelevant measurement or unitary operations:
\begin{lemma}
    If $\sempure{q}{\ket{\qstate}}$ and $\ket{\qstate'}$ is any sequence of $M_A$ or $U_A$ applied to $\ket{\qstate}$ where $A$ contains no qubits owned by $q$, then $\sempure{q}{\ket{\qstate'}}$.
\end{lemma}

Next, we prove the main theorem:
\begin{proof}
    Proceed by induction on the derivation of $\step{\ket{\qstate}}{e}{}{\ket{\qstate'}}{e'}$. In each case, to show that $\iscompat{\ket{\qstate'}}{e'}$, we show that every pure annotation inside a term that the step introduces is pure. If the step does not introduce pure annotations and does not have any effect on $\ket{\qstate}$, then $\iscompat{\ket{\qstate}}{e}$ immediately implies $\iscompat{\ket{\qstate'}}{e'}$.

    In case \rulename{S-App}, by inversion of $\hastype{\ctx_1, \ctx_2}{\rctx_1, \rctx_2}{\App{(\Lam{x}{e})}{e'}}{\type}$ we have that $\hastype{\ctx_1, x : \type'}{\rctx_1}{e}{\type}$ and $\hastype{\ctx_2}{\rctx_2}{e'}{\type'}$. Thus we have $\hastype{\ctx_1, \ctx_2}{\rctx_1, \rctx_2}{[e'/x]e}{\type}$.

    In case \rulename{S-Let}, by inversion of typing we have that $\hastype{\ctx, x : \type_1, y : \type_2}{\rctx}{e'}{\type}$ and $\hastype{\ctx_1}{\rctx_1}{e_1}{\type_1}$ and $\hastype{\ctx_2}{\rctx_2}{e_2}{\type_2}$. Thus we have $\hastype{\ctx, \ctx_1, \ctx_2}{\rctx, \rctx_1, \rctx_2}{[e_1, e_2/x, y]e}{\type}$.

    In case \rulename{S-Qinit}, the introduced term $\Qpurity{\Qref{\alpha}}{\pure}$ has type $\tPurity{\qubit}{\pure}$ and is separable and thus pure in $\ket{\qstate} \otimes \ket{0}_\alpha$.

    In case \rulename{S-SplitMixed}, there are no introduced pure annotations. By inversion of \linebreak $\hastype{\ctx}{\rctx_1, \rctx_2}{\SplitA{\mixed}{\Qpurity{\Epair{q_1}{q_2}}{\mixed}}}{\tPurity{(\tEpair{\qtype_1}{\qtype_2})}{\mixed}}$, we have that $\qhastype{\rctx_1}{q_1}{\qtype_1}$ and $\qhastype{\rctx_2}{q_2}{\qtype_2}$, meaning that $\hastype{\ctx}{\rctx_1, \rctx_2}{(\Qpurity{q_1}{\mixed}, \Qpurity{q_2}{\mixed})}{\tPurity{(\tEpair{\qtype_1}{\qtype_2})}{\mixed}}$.

    In case \rulename{S-Entangle}, by inversion of $\hastype{\ctx}{\rctx_1, \rctx_2}{\Entangle{\purity}{\Qpurity{q_1}{\purity}}{\Qpurity{q_2}{\purity}}}{\tPurity{(\tEpair{\qtype_1}{\qtype_2})}{\purity}}$ we have $\qhastype{\rctx_1}{q_1}{\qtype_1}$ and $\qhastype{\rctx_2}{q_2}{\qtype_2}$, meaning $\hastype{\ctx}{\rctx_1, \rctx_2}{\Qpurity{\Epair{q_1}{q_2}}{\purity}}{\tPurity{(\tEpair{\qtype_1}{\qtype_2})}{\purity}}$.
    If $\purity$ is $\mixed$ then no pure annotations are introduced.
    If $\purity$ is $\pure$, by the IH we know that $\sempure{q_1}{\ket{\qstate}}$ and $\sempure{q_2}{\ket{\qstate}}$, meaning $\sempure{\Epair{q_1}{q_2}}{\ket{\qstate}}$ and we have $\iscompat{\ket{\qstate}}{\Qpurity{\Epair{q_1}{q_2}}{\pure}}$.

    In case \rulename{S-Cast}, where $\purity$ is $\mixed$, no pure annotations are introduced. By inversion of $\hastype{\ctx}{\rctx}{\Cast{\mixed}{\Qpurity{q}{\purity}}}{\tPurity{\qtype}{\mixed}}$ we have $\qhastype{\rctx}{q}{\qtype}$, meaning $\hastype{\ctx}{\rctx}{\Qpurity{q}{\mixed}}{\tPurity{\qtype}{\mixed}}$.

    In case \rulename{S-U1}, by inversion of $\hastype{\ctx}{\rctx}{\Uone{\Qpurity{\Qref{\alpha}}{\purity}}}{\tPurity{\qtype}{\purity}}$ we have $\qhastype{\rctx}{\Qref{\alpha}}{\qtype}$, meaning $\hastype{\ctx}{\rctx}{\Qpurity{\Qref{\alpha}}{\purity}}{\tPurity{\qtype}{\purity}}$. If $\purity$ is $\mixed$ no pure annotations are introduced. If $\purity$ is $\pure$, then by the IH we know that $\sempure{q_1}{\ket{\qstate}}$. Because the unitary operator only acts on $\alpha$, we have $\sempure{q_1}{U_\alpha\ket{\qstate}}$ meaning $\iscompat{U_\alpha\ket{\qstate}}{\Qpurity{\Qref{\alpha}}{\pure}}$.

    In case \rulename{S-U2}, the same reasoning applies as for \rulename{S-U1}, except that if $\purity$ is pure then by the IH we know that $\sempure{\Epair{\Qref{\alpha}}{\Qref{\beta}}}{\ket{\qstate}}$ and the unitary operator only acts on $\alpha, \beta$, meaning we have $\sempure{\Epair{\Qref{\alpha}}{\Qref{\beta}}}{U_{\alpha, \beta}\ket{\qstate}}$ and $\iscompat{U_{\alpha, \beta}\ket{\qstate}}{\Qpurity{\Epair{\Qref{\alpha}}{\Qref{\beta}}}{\pure}}$.

    In cases \rulename{S-IfT} and \rulename{S-IfF}, by inversion of $\hastype{\ctx}{\rctx}{\If{e}{e_1}{e_2}}{\tPurity{\qtype}{\mixed}}$, where $e$ is $\True$ or $\False$, we have $\hastype{\ctx}{\rctx}{e_1}{\tPurity{\qtype}{\purity}}$ and $\hastype{\ctx}{\rctx}{e_2}{\tPurity{\qtype}{\purity}}$, meaning $\hastype{\ctx}{\rctx}{\Cast{\mixed}{e_1}}{\tPurity{\qtype}{\mixed}}$ and $\hastype{\ctx}{\rctx}{\Cast{\mixed}{e_2}}{\tPurity{\qtype}{\mixed}}$.

    In cases \rulename{S-MeasureT} and \rulename{S-MeasureF}, no pure annotations are introduced. By inversion of typing, $\Measure{\Qpurity{\Qref{\alpha}}{\purity}}$ also has Boolean type.

    In case \rulename{S-SplitPure}, by inversion of $\hastype{\ctx}{\rctx_1, \rctx_2}{\SplitA{\pure}{\Qpurity{\Epair{q_1}{q_2}}{\pure}}}{\tPair{\tPurity{\qtype_1}{\pure}}{\tPurity{\qtype_2}{\pure}}}$ we have that $\qhastype{\rctx_1}{q_1}{\qtype_1}$ and $\qhastype{\rctx_2}{q_2}{\qtype_2}$, meaning $\hastype{\ctx}{\rctx_1, \rctx_2}{(\Qpurity{q_1}{\pure}, \Qpurity{q_2}{\pure})}{\tPair{\tPurity{\qtype_1}{\pure}}{\tPurity{\qtype_2}{\pure}}}$. The premise is the separability condition that implies compatibility.

    In cases \rulename{S-AppL}, \rulename{S-AppR}, \rulename{S-PairL}, \rulename{S-PairR}, \rulename{S-LetS}, \rulename{S-SplitS}, \rulename{S-EntangleS}, \rulename{S-CastS}, \rulename{S-U1S}, \rulename{S-U2S}, \rulename{S-IfS}, \rulename{S-MeasureS}, and \rulename{S-SplitS}, the IH directly implies type preservation and state/expression compatibility.
\end{proof}

\subsubsection{Purity Soundness (\Cref{thm:purity})} \label{sec:purity-proof}

First, we state two properties of the language that hold as a consequence of it being a variant of the linear simply-typed $\lambda$-calculus with type safety.
\begin{lemma}[Strong Normalization]
    If $\hastype{\cdot}{\rctx}{e}{\type}$ then for all $\ket{\qstate}$ such that $\rctx \subseteq \domain{\ket{\qstate}}$, there exists $v$ such that $\stepv{\ket{\qstate}}{e}{\ket{\qstate'}}{v}$ and $\val{v}$.
\end{lemma}
This lemma implies the existence of an evaluation $\stepstar{\ket{\qstate}}{e}{}{\ket{\qstate'}}{v}$ for some $\ket{\qstate'}$ and $v$.
\begin{lemma}[Call-By Equivalence] \label{thm:calling-equivalence}
    $\stepv{\ket{\qstate}}{\App{(\Lam{x}{e})}{e'}}{\ket{\qstate'}}{v}$ such that $\val{v}$ if and only if $\stepv{\ket{\qstate}}{[e'/x]e}{\ket{\qstate'}}{v}$. Likewise, $\stepv{\ket{\qstate}}{\Let{x}{y}{(e_1, e_2)}{e}}{\ket{\qstate'}}{v}$ where $\val{v}$ if and only if $\stepv{\ket{\qstate}}{[e_1, e_2 / x, y]e}{\ket{\qstate'}}{v}$.
\end{lemma}
In this statement, $e'$, $e_1$ and $e_2$ need not be values, and thus this lemma asserts the equivalence of call-by-name and call-by-value evaluation strategies.
Because we operate in a linear $\lambda$-calculus where every variable occurs once without discarding or duplication, any effects (measurement or unitary operator) of the eagerly-evaluated argument of a function application or \texttt{let}-binding still occur exactly once if they are instead substituted before evaluation. Furthermore, the order of effects of two components of a pair does not matter, because linearity requires them to refer to disjoint sets of qubits, and reversing the order of measurements or unitary operators on disjoint sets of qubits cannot change the computation outcome.

We now give a proof for the purity soundness theorem by logical relations, strengthening the induction hypothesis to describe types other than $\tPurity{\qtype}{\pure}$ and typing judgments that involve non-empty contexts. For our relation, we define the notion of purity at a type $\type$, denoted $\sempuret{e}{\ket{\qstate}}{\type}$.
\begin{align*}
    \sempuret{e}{\ket{\qstate}}{\tPurity{\qtype}{\pure}} &= \sempure{e}{\ket{\qstate}} \\
    \sempuret{e}{\ket{\qstate}}{\tPurity{\qtype}{\mixed}} &= \top \\
    \sempuret{e}{\ket{\qstate}}{\tBool} &= \top \\
    \sempuret{e}{\ket{\qstate}}{\type_1 \to \type_2} &= \forall e', \sempuret{e'}{\ket{\qstate}}{\type_1} \Rightarrow \sempuret{\App{e}{e'}}{\ket{\qstate}}{\type_2} \\
    \sempuret{e}{\ket{\qstate}}{\tPair{\type_1}{\type_2}} &= \forall e_1, e_2, e', (\sempuret{e_1}{\ket{\qstate}}{\type_1} \Rightarrow \sempuret{e_2}{\ket{\qstate}}{\type_2} \Rightarrow \sempuret{[e_1, e_2 / x, y]e'}{\ket{\qstate}}{\type}) \\
    & \quad \Rightarrow \sempuret{\Let{x}{y}{e}{e'}}{\ket{\qstate}}{\type}
\end{align*}
Purity at type $\tPurity{\qtype}{\pure}$ simply invokes the existing definition. At mixed or Boolean type, purity contains no information. A function $\type_1 \to \type_2$ is pure when applying it to a pure argument at $\type_1$ yields an output pure at $\type_2$. Finally, a pure product $\tPair{\type_1}{\type_2}$ contains two elements that are pure at $\type_1$ and $\type_2$ respectively. We represent this by stating the elimination form of a product, a \texttt{let}-expression, is pure at $\type$ if its body $e'$ is pure at $\type$ after being substituted with two expressions pure at $\type_1$ and $\type_2$.

The following lemma states that reversal of deterministic steps preserves the purity relation.
\begin{lemma}
    If $\step{\ket{\qstate}}{e}{1}{\ket{\qstate'}}{e'}$ deterministically and $\sempure{e'}{\ket{\qstate'}}$ then $\sempure{e}{\ket{\qstate}}$. \label{thm:reversal}
\end{lemma}
The proof follows directly from the definition of purity.
Next, we define the standard notion of a substitution $\gamma$ mapping an open expression $e$ to a closed expression $\gamma(e)$. Define the compatibility judgment $\gamma \vDash \ctx, \rctx, \ket{\qstate}$ to mean that $\gamma$ maps every variable $x_i : \type_i$ in $\ctx$ to a closed term $e_i$ such that $\hastype{\cdot}{\rctx_i}{e_i}{\type_i}$ where $\rctx$ and all $\rctx_i$ are disjoint, and $\iscompat{\ket{\qstate}}{\rctx_i}$ and $\sempuret{e_i}{\ket{\qstate}}{\type_i}$.

Now we state the strengthened purity soundness theorem.
When $\type$ is $\tPurity{\qtype}{\pure}$ and $\ctx$ is empty, the strengthening implies the original theorem.
\begin{theorem} \label{thm:strengthened-purity-soundness}
    If $\hastype{\ctx}{\rctx}{e}{\type}$ and $\gamma \vDash \ctx, \rctx, \ket{\qstate}$ and $\iscompat{\ket{\qstate}}{\gamma(e)}$,
    then $\sempuret{\gamma(e)}{\ket{\qstate}}{\type}$.
\end{theorem}
\begin{proof}
    Proceed by induction on the derivation of $\hastype{\ctx}{\rctx}{e}{\type}$.

    The case for \rulename{T-Var} holds directly from the hypothesis.
    In case \rulename{T-Qval}, no substitution can occur. If $a = \mixed$ then purity is trivial. Otherwise, because $\iscompat{\ket{\qstate}}{e}$ we have $\sempure{e}{\ket{\qstate}}$.
    In case \rulename{T-Qinit}, no substitution can occur, and $e$ has only one deterministic transition and purity is trivial.

    In case \rulename{T-U1}, by the IH, $\hastype{\cdot}{\rctx}{\gamma(e)}{\tPurity{\qubit}{\pure}}$ and $\iscompat{\ket{\qstate}}{\gamma(e)}$ so $\sempure{\gamma(e)}{\ket{\qstate}}$.
    Any evaluation $\stepv{\ket{\qstate}}{\Uone{\gamma(e)}}{\ket{\qstate'}}{e'}$ where $\val{e'}$ is of the form
    $\ket{\qstate}; \Uone{\gamma(e)} \mapsto \qstate_1; \Uone{e_1} \mapsto \ldots \mapsto \ket{\qstate'}; \Uone{\Qpurity{\Qref{\alpha}}{\pure}} \mapsto_1 U_\alpha\ket{\qstate}; \Qpurity{\Qref{\alpha}}{\pure}$. Inverting each step, because $\sempure{\gamma(e)}{\ket{\qstate}}$ we have that the execution up to $\ket{\qstate'}$ is unique. By preservation, $\iscompat{U_\alpha\ket{\qstate'}}{\Qpurity{\Qref{\alpha}}{\pure}}$ so $\sempure{\Qpurity{\Qref{\alpha}}{\pure}}{U_\alpha\ket{\qstate'}}$. By~\Cref{thm:reversal}, we have $\sempure{\Uone{\Qpurity{\Qref{\alpha}}{\pure}}}{\ket{\qstate'}}$. Stitching the two deterministic executions together, we conclude that $\sempure{\gamma(\Uone{e})}{\ket{\qstate}}$.

    Case \rulename{T-U2} follows by the same reasoning as for \rulename{T-U1}, with the only difference being that we execute a two-qubit operator $U_{\alpha, \beta}$ on a pair of qubits in the deterministic step.

    Case \rulename{T-Entangle} also follows by similar reasoning. If $a = \mixed$ then purity is trivial. Otherwise, every execution first evaluates $\gamma(e)$ to a compatible final state and value $\iscompat{\ket{\qstate}}{\Qpurity{\Epair{q_1}{q_2}}{\pure}}$, implying $
    \sempure{\Epair{q_1}{q_2}}{\ket{\qstate}}$. Reverse this step and stitch together the deterministic executions of $\gamma(e_1)$ and $\gamma(e_2)$ to obtain $\sempure{\gamma(\Entangle{\purity}{e_1}{e_2})}{\ket{\qstate}}$.

    In case \rulename{T-Abs}, we have that $\hastype{x : \type}{\rctx}{\gamma(e)}{\type'}$. By the IH we have that if $\sempuret{e'}{\ket{\qstate}}{\type}$ then $\sempuret{[e'/x]\gamma(e)}{\ket{\qstate}}{\type'}$. From this we have $\sempuret{\App{(\Lam{x}{\gamma(e)})}{e'}}{\ket{\qstate}}{\type'}$ by~\Cref{thm:calling-equivalence}. We also have that $\sempuret{\gamma(\Lam{x}{e})}{\ket{\qstate}}{\type \to \type'}$.

    In case \rulename{T-Pair}, we have that $\hastype{\cdot}{\rctx_1}{\gamma(e_1)}{\type_1}$ and $\hastype{\cdot}{\rctx_2}{\gamma(e_2)}{\type_2}$. By the IH we have that $\sempuret{\gamma(e_1)}{\ket{\qstate}}{\type_1}$ and $\sempuret{\gamma(e_2)}{\ket{\qstate}}{\type_2}$. Thus, if we assume $\sempuret{[\gamma(e_1), \gamma(e_2) / x, y]e}{\ket{\qstate}}{\type}$, then we have $\sempuret{\Let{x}{y}{(\gamma(e_1), \gamma(e_2))}{e}}{\ket{\qstate}}{\type}$ by~\Cref{thm:calling-equivalence}, and so $\sempuret{(\gamma(e_1), \gamma(e_2))}{\ket{\qstate}}{\tPair{\type_1}{\type_2}}$.

    In case \rulename{T-App}, we have that $\hastype{\cdot}{\rctx_1}{\gamma(e_1)}{\type_1 \to \type_2}$ and $\hastype{\cdot}{\rctx_2}{\gamma(e_2)}{\type_1}$. By the IH we have $\sempuret{\gamma(e_1)}{\ket{\qstate}}{\type_1 \to \type_2}$ and $\sempuret{\gamma(e_2)}{\ket{\qstate}}{\type_1}$, thus $\sempuret{\gamma(\App{e_1}{e_2})}{\ket{\qstate}}{\type_2}$.

    In case \rulename{T-Let}, we have that $\hastype{\cdot}{\rctx_1}{\gamma(e)}{\tPair{\type_1}{\type_2}}$ and $\hastype{x : \type_1, y : \type_2}{\rctx_2}{\gamma(e')}{\type}$. By the IH we have $\sempuret{\gamma(e)}{\ket{\qstate}}{\tPair{\type_1}{\type_2}}$ and that if $\sempuret{e_1}{\ket{\qstate}}{\type_1}$ and $\sempuret{e_2}{\ket{\qstate}}{\type_2}$ then $\sempuret{[e_1, e_2 / x, y]\gamma(e')}{\ket{\qstate}}{\type}$. Thus, we have $\sempuret{\gamma(\Let{x}{y}{e}{e'})}{\ket{\qstate}}{\type}$.

    In case \rulename{T-SplitPure}, we have $\hastype{\cdot}{\rctx}{\gamma(e)}{\tPurity{(\tEpair{\qtype_1}{\qtype_2})}{\pure}}$ and need to show that\linebreak $\sempuret{\SplitA{\pure}{\gamma(e)}}{\ket{\qstate}}{\tPair{\tPurity{\qtype_1}{\pure}}{\tPurity{\qtype_2}{\pure}}}$. By the IH we have that $\sempure{\gamma(e)}{\ket{\qstate}}$.

    Proceed by the same reasoning as case \rulename{T-U1} to fully evaluate $\gamma(e)$ by rule \rulename{S-SplitS} and deterministically obtain $\Qpurity{\Epair{q_1}{q_2}}{\pure}$ under the unique state $\ket{\qstate'}$.

    First assume that the separability condition holds.
    The next step is \rulename{S-SplitPure}, and inverting it yields its premises $\ket{\qstate'} = \ket{\qstate_1} \otimes \ket{\qstate_2} \otimes \ket{\qstate_0}$, $\domain{\ket{\qstate_1}} = \Refs{q_1}$, $\domain{\ket{\qstate_2}} = \Refs{q_2}$, implying that $\sempure{q_1}{\ket{\qstate'}}$ and $\sempure{q_2}{\ket{\qstate'}}$. Thus, if we assume $\sempuret{[\Qpurity{q_1}{\pure}, \Qpurity{q_2}{\pure} / x, y]e'}{\ket{\qstate'}}{\type}$, from~\Cref{thm:calling-equivalence} we have that $\sempuret{\Let{x}{y}{(\Qpurity{q_1}{\pure}, \Qpurity{q_2}{\pure})}{e'}}{\ket{\qstate'}}{\type}$.

    Now assume that the separability condition fails. Then, every evaluation of $\SplitA{\pure}{\gamma(e)}$ will next take step \rulename{S-SplitFail}, meaning that every evaluation fails and purity holds vacuously.

    The remaining cases yield outputs of mixed or Boolean type, for which purity holds trivially.
\end{proof}

\subsection{Static Analysis for Purity}

\subsubsection{Purity Soundness (\Cref{thm:analysis-soundness})} \label{sec:analysis-soundness-proof}
We have already shown that all $\pure$ annotations introduced by the semantics of \LangName{} without \CastP{} are on pure expressions. We next show that if the static analysis assigns an expression $e$ a pure quantum type, then $e$ is pure. If so, rule \rulename{S-Cast} only introduces $\pure$ annotations on pure expressions, implying the soundness theorem.

\begin{proof}
    Proceed by induction on the derivation of $\hastypea{\ctx}{e}{\tPurity{\qtype}{\pure}}$. We only need to examine new $\pure$ annotations introduced by the $\vdash_A$ judgment.

    In case \rulename{A-CastPure}, if $e$ has pure type then by the IH it is pure, meaning $\Cast{\pure}{e}$ is also pure.

    In case \rulename{A-SplitMixed}, $g = \SplitPoly{f}{j}$ is never $\pure$, so purity is trivially satisfied.

    In case \rulename{A-Entangle}, if $h$ is not $\pure$ then purity is trivial. In addition, if $f$ and $g$ are both $\pure$ then the same reasoning as in the original proof applies.
    If $f$ and $g$ are not $\pure$ but $h$ is pure, we must show that $\Entangle{\purity}{e_1}{e_2}$ is pure. Suppose it is not. The first possibility is that it evaluates to an entangled pair $\Epair{e_1}{e_2}$ that is entangled with some qubit $\alpha$ not owned by $e_1$ or $e_2$. But this means one of $e_1$ or $e_2$ is entangled with $\alpha$; without loss of generality let it be $e_1$. Then, some $\mathtt{split}_\mixed$ occurred with $\alpha$ on one side and $e_1$ on the other, introducing a term for $x_i$ in the history $f$. Because $\alpha$ is not owned by $e_2$, $x_i$ is absent from the history $g$ and must be present in $h = \CombinePoly{f}{g}$, contradicting the fact that $h$ is $\pure$. The second possibility is that the evaluation of $\Entangle{\purity}{e_1}{e_2}$ encounters an \texttt{if}-expression that evaluates to a mixed state that is not discarded. But then this \texttt{if}-expression appears in the evaluation of $e_1$ or $e_2$, meaning by preservation that either $f$ or $g$ is $\mixed$ and thus $h = \CombinePoly{f}{g}$ is $\mixed$, which is also a contradiction.

    All other rules follow by the same reasoning as in the proof of~\Cref{thm:strengthened-purity-soundness}.
\end{proof}
\section{Full Benchmark Descriptions} \label{sec:app-examples}

\paragraph{Teleport-Deferred} This program implements the deferred-measurement teleportation from~\Cref{sec:examples}. We also implemented an erroneous variant, \emph{Teleport-NoCZ}, which replaces the CZ gate with a CNOT and results in the final qubit actually being entangled, and a third variant \emph{Teleport-Measure} that measures the ancillas and uses classical rather than quantum conditioning.

\paragraph{AndOracle} This program inverts the phase of the state conditioned on two qubits, as may be seen in an oracle for Grover's algorithm. We also implement an erroneous variant \emph{AndOracle-NotUncomputed} that does not correctly uncompute the ancilla.

\paragraph{Bell-GHZ} This program attempts to illegally substitute a sub-state of a Greenberger-Horne-Zeilinger state~\citep{ghz} for a Bell state, by dropping an entangled ancilla.

\paragraph{Deutsch} This program implements Deutsch's algorithm~\citep{deutsch}, which determines whether a black-box function $f : \{0, 1\} \to \{0, 1\}$ is the constant function. We also implemented an erroneous variant \emph{Deutsch-BadResultBasis}, which omits an essential Hadamard gate, causing the result qubit to be in the incorrect basis, and attempts to drop an entangled ancilla.

\paragraph{DeutschJozsa} This program implements the Deutsch-Jozsa algorithm~\citep{deutsch}, a generalization of Deutsch's algorithm, which determines whether a black-box function $f : \{0, 1\}^n \to \{0, 1\}$ is constant or balanced (returns 1 for exactly half of the domain). We also implemented an erroneous variant \emph{DeutschJozsa-MixedInit}, which uses an incorrect initial state of ${(\ket{00} + \ket{11})}/{\sqrt{2}}$ rather than ${(\ket{00} + \ket{01} + \ket{10} + \ket{11})}/{2}$ and will produce an incorrect result. This variant uses a purifying cast to try to force the algorithm to accept this incorrect state.

\paragraph{Grover} This program implements Grover's search algorithm~\citep{grover} on a two-qubit, four-element database, locating a distinguished element in cell $\ket{11}$. We also implemented an erroneous variant \emph{Grover-BadOracle}, which uses \emph{AndOracle-NotUncomputed} and drops an entangled ancilla. It uses a purifying cast to try to force the algorithm to accept this incorrect oracle.

\paragraph{QFT} This program implements the quantum Fourier transform~\citep{qft}, a building block for Simon's algorithm and Shor's algorithm, on three qubits. The purity specification of the QFT is that its output has the same purity annotation as its input.

\paragraph{ShorCode} This program implements encoding and decoding operations for Shor's error correcting code~\citep{shor-code}, which uses nine physical qubits to correct an arbitrary single error on one qubit. The program also implements a phase flip error channel and runs error correction on a qubit subject to phase-flip error. The purity specifications on the encoding and decoding operations allows the program to discard the extra parity bits that are separable from the decoded data. We also implemented an erroneous variant \emph{ShorCode-Drop}, which is the result of a programmer error that shadows a function argument that is not known to be pure.

\paragraph{ModMul($n$)} We followed the scheme of~\citet{markov2015constantoptimized} for quantum circuits for modular multiplication, which generalizes a predictable structural pattern for arbitrary $n$, with the number of gates linear in $n$. For each $n$, we implemented conditional multiplication mod $2^n-1$ by some $k$ and also by $k^{-1}$. Applying the first operation on a $n$-qubit register and a condition qubit entangles the register and qubit. Then, applying the inverse operation must disentangle the register and condition. Thus, the program uses a purifying-split operator to verify that the condition qubit is pure at program termination. We also implemented an erroneous variant \emph{ModMul($n$)-NotInverse} where multiplication by $k^{-1}$ is defective, resulting in a register and condition that are still entangled.

\subsection{Analysis Result Descriptions}

\paragraph{Teleport-Deferred} We determine that the teleportation circuit satisfies its purity specification, taking a pure qubit to another pure qubit. We also determine that the erroneous variant that substitutes a CZ gate does not correctly separate the output qubit from the temporaries, meaning that measuring the temporaries causes the qubit to enter a mixed state.

\paragraph{AndOracle} We determine that the oracle correctly uncomputes ancillas and yields a pure output. We also determine that the variant that does not uncompute an ancilla may yield a mixed output, violating its specification.

\paragraph{Bell-GHZ} We determine that the program is ill-typed due to its attempt to directly return mixed qubits in a function whose output is specified to be pure.

\paragraph{Deutsch} We determine that the algorithm satisfies its specification, including the fact that an ancilla is separable from the output and can be safely dropped. We also determine that the erroneous variant results in an entangled ancilla and rejects the attempt to drop it at runtime.

\paragraph{DeutschJozsa} We determine that the algorithm satisfies its specification. For the erroneous variant with a defective initial state, we determine at the use site of the state that it has an incorrect purity and the attempt to directly cast the type of the state fails the static analysis.

\paragraph{Grover} We determine that the algorithm satisfies its specification. For the erroneous variant with a defective oracle, the attempt to directly cast the result of the oracle fails the static analysis.

\paragraph{QFT} We determine that the algorithm satisfies its purity specification, taking pure inputs to pure outputs. The result indicates that the output is correct and lacks entanglement with any other qubit in the system.

\paragraph{ShorCode} We determine that the algorithm satisfies its purity specification, including the fact that at the end of the decoding process, the extra check bits may be safely discarded without disrupting the decoded data. For the erroneous variant with a programming error, the type checker detects that a value not known to be pure cannot be shadowed.

\paragraph{ModMul($n$)} We determine that the algorithm satisfies its purity specification that the condition bit is separable from the output, implying that the inverse operation was implemented correctly. For the erroneous variant with incorrect inverse, the runtime verification fails, indicating the condition bit is still entangled and that the algorithm is incorrect.
\section{Full Benchmark Programs} \label{sec:app-full-sources}

In this section, we present the full source code for each benchmark. Several benchmarks use syntactic features of \LangName{} not discussed in the main paper, including inference of purity assertions and polymorphic purity annotations, which are described in~\Cref{sec:higher-level}.

\subsection{Teleport-Deferred}
\begin{lstlisting}[numbers=left,style=tiny]
fun bell_pair () : (qubit & qubit)<P> =
  CNOT (H (qinit ()), qinit ())

fun teleport (q1 : qubit<P>) : qubit<P> =
  let (q2 : qubit<M>, q3 : qubit<M>) = bell_pair () in
  let (q1 : qubit<M>, q2 : qubit<M>) = CNOT (q1, q2) in
  let q1 = H (q1) in
  let (q2 : qubit<M>, q3 : qubit<M>) = CNOT (q2, q3) in
  let (q1 : qubit<M>, q3 : qubit<M>) = CZ (q1, q3) in
  let all : ((qubit & qubit) & qubit)<P> = ((q1, q2), q3) in
  let (_ : (qubit & qubit)<P>, q3 : qubit<P>) = all in q3

fun main () : qubit<P> = teleport (H (qinit ()))
\end{lstlisting}

\subsection{Teleport-NoCZ}

All other functions are identical to the \emph{Teleport-Deferred} example.

\begin{lstlisting}[numbers=left,style=tiny]
fun teleport (q1 : qubit<P>) : qubit<P> =
  let (q2 : qubit<M>, q3 : qubit<M>) = bell_pair () in
  let (q1 : qubit<M>, q2 : qubit<M>) = CNOT (q1, q2) in
  let q1 : qubit<M> = H (q1) in
  let (q2 : qubit<M>, q3 : qubit<M>) = CNOT (q2, q3) in
  let (q1 : qubit<M>, q3 : qubit<M>) = CNOT (q1, q3) in
  let all : ((qubit & qubit) & qubit)<P> = ((q1, q2), q3) in
  (* Dynamic separability check failure: argument entangled *)
  let (discard : (qubit & qubit)<P>, q3 : qubit<P>) = all in
  let discard = measure (discard) in q3
\end{lstlisting}

\subsection{Teleport-Measure}

All other functions are identical to the \emph{Teleport-Deferred} example.

\begin{lstlisting}[numbers=left,style=tiny]
fun teleport (q1 : qubit<P>) : qubit<P> =
  let (q2 : qubit<M>, q3 : qubit<M>) = bell_pair () in
  let (q1 : qubit<M>, q2 : qubit<M>) = CNOT (q1, q2) in
  let q1 : qubit<M> = H (q1) in
  let q3 = if measure (q2) then X (q3) else q3 in
  let q3 = if measure (q1) then Z (q3) else q3 in
  cast<P>(q3) (* Static analysis failure: q1 and q2 not covered *)
\end{lstlisting}

\subsection{AndOracle}
\begin{lstlisting}[numbers=left,style=tiny]
fun and_oracle (p0 : qubit<P>, p1 : qubit<P>) : (qubit & qubit)<P> =
  let x = qinit () in
  let (p0 : qubit<M>, (p1 : qubit<M>, x : qubit<M>)) = TOF (p0, (p1, x)) in
  let (p0 : qubit<M>, (p1 : qubit<M>, x : qubit<M>)) = TOF (p0, (p1, x)) in
  let qs : (qubit & (qubit & qubit))<P> = (x, (p0, p1)) in
  let (x : qubit<P>, rest : (qubit & qubit)<P>) = qs in rest
fun main () : (qubit & qubit)<P> = and_oracle (H (qinit ()), X (qinit ()))
\end{lstlisting}

\subsection{AndOracle-NotUncomputed}
\begin{lstlisting}[numbers=left,style=tiny]
fun and_oracle (p0 : qubit<P>, p1 : qubit<P>) : (qubit & qubit)<P> =
  let x = qinit () in
  let (p0 : qubit<M>, (p1 : qubit<M>, x : qubit<M>)) = TOF (p0, (p1, x)) in
  let p0 = Z (p0) in
  let _ = measure x in
  entangle<P>(p0, p1) (* Type error: p0 and p1 are mixed *)
fun main () : (qubit & qubit)<P> = and_oracle (H (qinit ()), X (qinit ()))
\end{lstlisting}

\subsection{Bell-GHZ}
\begin{lstlisting}[numbers=left,style=tiny]
fun main () : (qubit & qubit)<P> =
  let q1 = H (qinit ()) in
  let (q1 : qubit<M>, q2 : qubit<M>) = CNOT (q1, qinit ()) in
  let (q1 : qubit<M>, q3 : qubit<M>) = CNOT (q1, qinit ()) in
  let _ = measure q3 in
  entangle<P>(q1, q2) (* Type error: q1 and q2 are mixed *)
\end{lstlisting}

\subsection{Deutsch}
\begin{lstlisting}[numbers=left,style=tiny]
fun deutsch (uf : (qubit & qubit)<P> -> (qubit & qubit)<P>) : bool =
  let input : (qubit & qubit)<P> = (H (qinit ()), H (X (qinit ()))) in
  let (x : qubit<P>, _ : qubit<P>) = uf (input) in
  measure (H (x))

fun cnot (xy : (qubit & qubit)<P>) : (qubit & qubit)<P> = CNOT (xy)

fun always_true (xy : (qubit & qubit)<P>) : (qubit & qubit)<P> =
  let (x : qubit<M>, y : qubit<M>) = xy in
  cast<P>(entangle<M>(x, X (y)))

fun always_false (xy : (qubit & qubit)<P>) : (qubit & qubit)<P> = xy

fun main () : ((bool * bool) * bool) =
  ((deutsch (always_false), deutsch (always_true)), deutsch (cnot))
\end{lstlisting}

\subsection{Deutsch-BadResultBasis}

All other functions are identical to the \emph{Deutsch} example.

\begin{lstlisting}[numbers=left,style=tiny]
fun deutsch (uf : (qubit & qubit)<P> -> (qubit & qubit)<P>) : bool =
  (* Dynamic separability check failure: argument entangled *)
  let (x : qubit<P>, y : qubit<P>) =
    uf (entangle<P>(H (qinit ()), (X (qinit ())))) in
  let _ = measure (y) in
  measure (H (x))
\end{lstlisting}

\subsection{DeutschJozsa}
\begin{lstlisting}[numbers=left,style=tiny]
type oracle = ((qubit & qubit) & qubit)<P> -> ((qubit & qubit) & qubit)<P>
type domain = (qubit & qubit)<P>

(* An (entangled) domain-codomain pair *)
type graph_pt = ((qubit & qubit) & qubit)<P>

(* Prepare domain qubits in |0> + |1> *)
fun init_domain () : (qubit<P> * qubit<P>) = (H (qinit ()), H (qinit ()))

(* Prepare output qubit in |0> - |1> *)
fun init_output () : qubit<P> = H (X qinit ())

fun test_oracle (f : oracle) : graph_pt =
  let out : qubit<P> = init_output () in
  let dom : (qubit<P> * qubit<P>) = init_domain () in
  let all : graph_pt = (dom, out) in
  let inout : graph_pt = f (all) in
  let ((d0 : qubit<M>, d1 : qubit<M>), out: qubit<M>) = inout in
  (* Hadamard the domain qubits *)
  let (inout_post : ((qubit & qubit) & qubit)<M>) =
    (((H d0), (H d1)), out) in
  cast<P>(inout_post)

(* A balanced function {0, 1}^2 -> {0, 1} that selects states with second
    qubit |1> *)
fun is_odd (pt : graph_pt) : graph_pt =
  let ((d0 : qubit<M>, d1 : qubit<M>), out : qubit<M>) = pt in
  let (d1 : qubit<M>, out : qubit<M>) = (CNOT (d1, out)) in
  let (inout : ((qubit & qubit) & qubit)<M>) = ((d0, d1), out) in
  cast<P>(inout)

fun main () : graph_pt = test_oracle (is_odd)
\end{lstlisting}

\subsection{DeutschJozsa-MixedInit}

The other functions are identical to the \emph{DeutschJozsa} example.

\begin{lstlisting}[numbers=left,style=tiny]
fun init_domain () : (qubit<M> * qubit<M>) =
  let (x : qubit<M>, y : qubit<M>) = CNOT (H (qinit ()), qinit ()) in
  let _ = measure (y) in
  (x, cast<M>(qinit ()))

fun test_oracle (f : oracle) : graph_pt =
  let out : qubit<P> = init_output () in
  let dom : (qubit<M> * qubit<M>) = init_domain () in
  let all : graph_pt = (dom, out) in
  let inout : graph_pt = f (all) in
  let ((d0 : qubit<M>, d1 : qubit<M>), out: qubit<M>) = inout in
  (* Hadamard the domain qubits *)
  let (inout_post : ((qubit & qubit) & qubit)<M>) =
    (((H d0), (H d1)), out) in
  (* Static analysis failure: mixed output of init_domain () not covered *)
  cast<P>(inout_post)
\end{lstlisting}

\subsection{Grover}
\begin{lstlisting}[numbers=left,style=tiny]
type addr = (qubit & qubit)<P>
type oracle = (qubit & qubit)<P> -> (qubit & qubit)<P>
fun init_addr () : addr = entangle<P>(H qinit(), H qinit ())
fun diffuse (p : addr) : addr =
  let (p0 : qubit<M>, p1 : qubit<M>) = p in
  let (p0 : qubit<M>, p1 : qubit<M>) = (H p0, H p1) in
  let (p0 : qubit<M>, p1 : qubit<M>) = (Z p0, Z p1) in
  let (p0 : qubit<M>, p1 : qubit<M>) = CZ (p0, p1) in
  let (p0 : qubit<M>, p1 : qubit<M>) = (H p0, H p1) in
  let p = entangle<M>(p0, p1) in
  cast<P>(p)

fun grover (f : oracle) : addr =
  let addr = init_addr () in
  let addr = f (addr) in
  let addr = diffuse (addr) in
  addr

fun final_addr (p : addr) : addr = (CZ (p))
fun main () : addr = grover (final_addr)
\end{lstlisting}

\subsection{Grover-BadOracle}

The other functions are identical to the \emph{Grover} example.

\begin{lstlisting}[numbers=left,style=tiny]
fun final_addr (p : addr) : addr =
  let (p0 : qubit<M>, p1 : qubit<M>) = p in
  let x = qinit () in
  let (p0 : qubit<M>, (p1 : qubit<M>, x : qubit<M>)) = TOF (p0, (p1, x)) in
  let (x : qubit<M>, p0 : qubit<M>) = CZ (x, p0) in
  let _ = measure x in
  cast<P>(entangle<M>(p0, p1)) (* Static analysis failure: x not covered *)
\end{lstlisting}

\subsection{QFT}
\begin{lstlisting}[numbers=left,style=tiny]
fun qft_sub_1 (q : qubit<'p>) : qubit<'p> = H q

fun qft_sub_2 (qs : (qubit & qubit)<'p>) : (qubit & qubit)<'p> =
  let (q0 : qubit<M>, q1 : qubit<M>) = cast<M>(qs) in
  let q0 = qft_sub_1 (q0) in
  (* Controlled phase of 2 pi / 2 ** 2 *)
  let (qs : (qubit & qubit)<M>) = CPHASE 0.250 (q1, q0) in
  let (q1 : qubit<M>, q0 : qubit<M>) = qs in
  let (qs : (qubit & qubit)<M>) = (q0, q1) in
  cast<'p>(qs)

fun qft_sub_3 (qs : ((qubit & qubit) & qubit)<'p>) :
  ((qubit & qubit) & qubit)<'p> =
  let (qs : (qubit & qubit)<M>, q2 : qubit<M>) = qs in
  let qs = qft_sub_2 (qs) in
  let (q0 : qubit<M> , q1 : qubit<M>) = qs in
  (* Controlled phase of 2 pi / 2 ** 3 *)
  let (q2 : qubit<M>, q0 : qubit<M>) = CPHASE 0.125 (q2, q0) in
  let (qs : ((qubit & qubit) & qubit)<M>) = ((q0, q1), q2) in
  cast<'p>(qs)

fun qft_1 (q : qubit<'p>) : qubit<'p> = qft_sub_1 (q)

fun qft_2 (qs : (qubit & qubit)<'p>) : (qubit & qubit)<'p> =
  let qs = qft_sub_2 (qs) in
  let (q0 : qubit<M>, q1 : qubit<M>) = qs in
  let q1 = qft_1 (q1) in
  let (qs: (qubit & qubit)<M>) = (q0, q1) in
  cast<'p>(qs)

fun qft_3 (qs : ((qubit & qubit) & qubit)<'p>) :
  ((qubit & qubit) & qubit)<'p> =
  let qs = qft_sub_3 (qs) in
  let ((q0 : qubit<M>, q1 : qubit<M>), q2 : qubit<M>) = qs in
  let tail : (qubit & qubit)<M> = (q1, q2) in
  let (q1 : qubit<M>, q2 : qubit<M>) = qft_2 (tail) in
  let (qs: ((qubit & qubit) & qubit)<M>) = ((q0, q1), q2) in
  cast<'p>(qs)

fun main () : ((qubit & qubit) & qubit)<P> =
  let qs : ((qubit & qubit) & qubit)<pure> = ((qinit(), X qinit()), qinit()) in
  qft_3 (qs)
\end{lstlisting}

\subsection{ShorCode}
\begin{lstlisting}[numbers=left,style=tiny]
type triple_p = (qubit & (qubit & qubit))<P>
type triple_m = (qubit & (qubit & qubit))<M>
type nonuple_p = ((qubit & (qubit & qubit)) &
                 ((qubit & (qubit & qubit)) &
                  (qubit & (qubit & qubit))))<P>
type nonuple_m = ((qubit & (qubit & qubit)) &
                 ((qubit & (qubit & qubit)) &
                  (qubit & (qubit & qubit))))<M>

(* Encode a qubit with the three-qubit bit-flip code *)
fun enc_bit (q : qubit<'p>) : (qubit & (qubit & qubit))<'p> =
  let (a1 : qubit<M>) = qinit () in
  let (a2 : qubit<M>) = qinit () in
  let (q : qubit<M>, a2 : qubit<M>) = (CNOT (q, a2)) in
  let (q : qubit<M>, a1 : qubit<M>) = (CNOT (q, a1)) in
  let (out : (qubit & (qubit & qubit))<M>) = (q, (a1, a2)) in
  cast<'p>(out)

(* Encode a qubit with the three-qubit phase-flip code *)
fun enc_phase (q : qubit<'p>) : (qubit & (qubit & qubit))<'p> =
  let (x : qubit<M>, (y : qubit<M>, z : qubit<M>)) = enc_bit (q) in
  let (out : (qubit & (qubit & qubit))<M>) = (H x, (H y, H z)) in
  cast<'p>(out)

(* Encode a qubit with the nine-qubit Shor code by concatenating the bit- and
 * phase-flip codes *)
fun enc_shor (q : qubit<P>) : nonuple_p =
  let (x : qubit<M>, (y : qubit<M>, z : qubit<M>)) = enc_phase (cast<M>(q)) in
  let (out : nonuple_m) = (enc_bit (x), (enc_bit (y), enc_bit (z))) in
  cast<P>(out)

(* Decode the three-qubit bit-flip code *)
fun dec_bit (enc : (qubit & (qubit & qubit))<'p>) :
  (qubit & (qubit & qubit))<'p> =
  let (q0 : qubit<M>, tail : (qubit & qubit)<M>) = enc in
  let (q1 : qubit<M>, q2 : qubit<M>) = tail in
  let (q0 : qubit<M>, q1 : qubit<M>) = (CNOT (q0, q1)) in
  let (q0 : qubit<M>, q2 : qubit<M>) = (CNOT (q0, q2)) in
  let (q2 : qubit<M>, (q1 : qubit<M>, q0 : qubit<M>)) = TOF (q2, (q1, q0)) in
  let env : (qubit & qubit)<M> = entangle<M>(q1, q2) in
  let out : (qubit & (qubit & qubit))<M> = entangle<M>(q0, env) in
  cast<'p>(out)

(* Decode the three-qubit phase-flip code *)
fun dec_phase (enc : (qubit & (qubit & qubit))<'p>) :
  (qubit & (qubit & qubit))<'p> =
  let (x : qubit<M>, (y : qubit<M>, z : qubit<M>)) = enc in
  let qs : (qubit & (qubit & qubit))<M> = (H x, (H y, H z)) in
  let out : (qubit & (qubit & qubit))<'p> = dec_bit (cast<'p>(qs)) in
  out

(* Decode the Shor code, without discarding the extra bits *)
fun dec_shor (enc : nonuple_p) : nonuple_p =
  let (x : triple_m, (y : triple_m, z : triple_m)) = enc in
  let (x : triple_m, (y : triple_m, z : triple_m)) =
    (dec_bit (x), (dec_bit (y), dec_bit(z))) in

  let (x0 : qubit<M>, (x1 : qubit<M>, x2 : qubit<M>)) = x in
  let (y0 : qubit<M>, (y1 : qubit<M>, y2 : qubit<M>)) = y in
  let (z0 : qubit<M>, (z1 : qubit<M>, z2 : qubit<M>)) = z in
  let heads : triple_m = (x0, (y0, z0)) in
  let (x0 : qubit<M>, (y0 : qubit<M>, z0 : qubit<M>)) = dec_phase (heads) in

  let x = (x0, (x1, x2)) in
  let y = (y0, (y1, y2)) in
  let z = (z0, (z1, z2)) in
  let dec : nonuple_m = (x, (y, z)) in
  cast<P>(dec)

fun test_bitflip (q : qubit<P>) : qubit<P> =
  let (q0 : qubit<M>, tail : (qubit & qubit)<M>) = enc_bit (q) in
  let (q1 : qubit<M>, q2 : qubit<M>) = tail in
  let qs : (qubit & (qubit & qubit))<M> = (q0, (q1, q2)) in
  let (q : qubit<P>, env : (qubit & qubit)<P>) =
    dec_bit (cast<P>(qs)) in
  q

(* Accept a qubit and a noise operation. Encode the qubit, apply the given noise
 * operation on the physical qubits, then decode and discard the extra bits. *)
fun shor_ecc (qop : (qubit<P> * (nonuple_p -> nonuple_p))) : qubit<P> =
  (* Encode the qubit *)
  let (q : qubit<P>, op : nonuple_p -> nonuple_p) = qop in
  let enc = enc_shor (q) in
  (* Disturb the encoded state with the noise operation *)
  let enc = op (enc) in
  (* Decode and discard the parity bits *)
  let dec = dec_shor (enc) in
  let (x : triple_p, (y : triple_p, z : triple_p)) = dec in
  let (dec : qubit<P>, others : (qubit & qubit)<P>) = x in
  dec

(* A unitary channel that produces a single phase flip *)
fun phaseflip_channel (enc : nonuple_p) : nonuple_p =
  let (x : triple_m, (y : triple_m, z : triple_m)) = enc in
  let (x0 : qubit<M>, x_tail : (qubit & qubit)<M>) = x in
  let x0 = Z x0 in
  let x : triple_m = (x0, x_tail) in
  let enc : nonuple_p = (x, (y, z)) in
  cast<P>(enc)

fun main () : qubit<P> = shor_ecc ((H (qinit ()), phaseflip_channel))
\end{lstlisting}

\subsection{ShorCode-Drop}

All other functions are identical to the \emph{ShorCode} example.

\begin{lstlisting}[numbers=left,style=tiny]
fun enc_bit (q : qubit<'p>) : ((qubit & qubit) & qubit)<'p> =
  let (a1 : qubit<M>) = qinit () in
  let (a2 : qubit<M>) = qinit () in
  let (q : qubit<M>) = qinit () in (* Type error: drops shadowed q *)
  let (q : qubit<M>, a2 : qubit<M>) = (CNOT (q, a2)) in
  let (q : qubit<M>, a1 : qubit<M>) = (CNOT (q, a1)) in
  let (out : ((qubit & qubit) & qubit)<M>) = ((q, a1), a2) in
  cast<'p>(out)
\end{lstlisting}

\subsection{ModMul(4)}

We next show the modular multiplication benchmark for $n = 4$. All other benchmarks in this family are very similar.

\begin{lstlisting}[numbers=left,style=tiny]
type five = (qubit & (((qubit & qubit) & qubit) & qubit))<P>
type four_m = (((qubit & qubit) & qubit) & qubit)<M>
type four_p = (((qubit & qubit) & qubit) & qubit)<P>

(* controlled multiply by 7 mod 15,
 * using negation followed by three controlled swaps *)
fun mult7 (cqs : five) : five =
  let (c : qubit<M>, qs : four_m) = cqs in
  let (((q1 : qubit<M>, q2 : qubit<M>), q3 : qubit<M>), q4 : qubit<M>) = qs in
  let (c : qubit<M>, q1 : qubit<M>) = CNOT (c, q1) in
  let (c : qubit<M>, q2 : qubit<M>) = CNOT (c, q2) in
  let (c : qubit<M>, q3 : qubit<M>) = CNOT (c, q3) in
  let (c : qubit<M>, q4 : qubit<M>) = CNOT (c, q4) in
  let (c : qubit<M>, (q2 : qubit<M>, q3 : qubit<M>)) = FRED (c, (q2, q3)) in
  let (c : qubit<M>, (q1 : qubit<M>, q2 : qubit<M>)) = FRED (c, (q1, q2)) in
  let (c : qubit<M>, (q1 : qubit<M>, q4 : qubit<M>)) = FRED (c, (q1, q4)) in
  let res : five = (c, (((q1, q2), q3), q4)) in
  res

(* controlled multiply by 13 mod 15 *)
fun mult13 (cqs : five) : five =
  let (c : qubit<M>, qs : four_m) = cqs in
  let (((q1 : qubit<M>, q2 : qubit<M>), q3 : qubit<M>), q4 : qubit<M>) = qs in
  let (c : qubit<M>, q1 : qubit<M>) = CNOT (c, q1) in
  let (c : qubit<M>, q2 : qubit<M>) = CNOT (c, q2) in
  let (c : qubit<M>, q3 : qubit<M>) = CNOT (c, q3) in
  let (c : qubit<M>, q4 : qubit<M>) = CNOT (c, q4) in
  let (c : qubit<M>, (q1 : qubit<M>, q4 : qubit<M>)) = FRED (c, (q1, q4)) in
  let (c : qubit<M>, (q1 : qubit<M>, q2 : qubit<M>)) = FRED (c, (q1, q2)) in
  let (c : qubit<M>, (q2 : qubit<M>, q3 : qubit<M>)) = FRED (c, (q2, q3)) in
  let res : five = (c, (((q1, q2), q3), q4)) in
  res

fun z () : qubit<P> = qinit ()
fun o () : qubit<P> = H (qinit ())
fun main () : (qubit<P> * four_p) =
  let c = o () in
  (* 0b1001 = 9 *)
  let num : four_p = ((((o ()), z ()), z ()), o ()) in
  let (c : qubit<P>, rest : four_p) = mult13 (mult7 (entangle<P>(c, num))) in
  (* restored to 0b1001 *)
  (c, rest)
\end{lstlisting}

\subsection{ModMul(4)-NotInverse}

All other functions are identical to the \emph{ModMul(4)} example.

\begin{lstlisting}[numbers=left,style=tiny]
fun mult13 (cqs : (qubit & (((qubit & qubit) & qubit) & qubit))<P>) :
  (qubit & (((qubit & qubit) & qubit) & qubit))<P> =
  let (c : qubit<M>, qs : (((qubit & qubit) & qubit) & qubit)<M>) = cqs in
  let (((q1 : qubit<M>, q2 : qubit<M>), q3 : qubit<M>), q4 : qubit<M>) = qs in
  let (c : qubit<M>, q1 : qubit<M>) = CNOT (c, q1) in
  let (c : qubit<M>, q2 : qubit<M>) = CNOT (c, q2) in
  let (c : qubit<M>, q3 : qubit<M>) = CNOT (c, q3) in
  let (c : qubit<M>, q4 : qubit<M>) = CNOT (c, q4) in
  let (c : qubit<M>, (q1 : qubit<M>, q4 : qubit<M>)) = FRED (c, (q1, q4)) in
  let (c : qubit<M>, (q1 : qubit<M>, q3 : qubit<M>)) = FRED (c, (q1, q3)) in (* WRONG *)
  let (c : qubit<M>, (q2 : qubit<M>, q3 : qubit<M>)) = FRED (c, (q2, q3)) in
  let res : (qubit & (((qubit & qubit) & qubit) & qubit))<P> =
    (c, (((q1, q2), q3), q4)) in
  res

fun main () : (qubit<P> * four_p) =
  let c = o () in
  let num : four_p = ((((o ()), z ()), z ()), o ()) in
  (* Dynamic separability check failure: argument entangled *)
  let (c : qubit<P>, rest : four_p) = mult13 (mult7 (entangle<P>(c, num))) in
  (c, rest)
\end{lstlisting}

\end{document}